\RequirePackage{snapshot}
\documentclass[11pt,twoside,english,british]{article}
\usepackage[T1]{fontenc}
\usepackage[latin9]{inputenc}
\usepackage[a4paper]{geometry}
\usepackage{fancyhdr}
\usepackage{color}
\usepackage{babel}
\usepackage{verbatim}
\usepackage{refstyle}
\usepackage{booktabs}
\usepackage{mathrsfs}
\usepackage{mathtools}
\usepackage{enumitem}
\usepackage{amsmath}
\usepackage{amsthm}
\usepackage{amssymb}
\usepackage{graphicx}
\usepackage{setspace}
\usepackage[authoryear,longnamesfirst]{natbib}
\usepackage{xargs}[2008/03/08]
\usepackage[pdfusetitle,
 bookmarks=true,bookmarksnumbered=false,bookmarksopen=false,
 breaklinks=false,pdfborder={0 0 0},pdfborderstyle={},backref=false,colorlinks=false]
 {hyperref}

\makeatletter

\AtBeginDocument{\providecommand\thmref[1]{\ref{thm:#1}}}
\AtBeginDocument{\providecommand\eqref[1]{\ref{eq:#1}}}
\AtBeginDocument{\providecommand\secref[1]{\ref{sec:#1}}}
\AtBeginDocument{\providecommand\subsecref[1]{\ref{subsec:#1}}}
\AtBeginDocument{\providecommand\remref[1]{\ref{rem:#1}}}
\AtBeginDocument{\providecommand\appref[1]{\ref{app:#1}}}
\AtBeginDocument{\providecommand\assref[1]{\ref{ass:#1}}}
\AtBeginDocument{\providecommand\exaref[1]{\ref{exa:#1}}}
\AtBeginDocument{\providecommand\propref[1]{\ref{prop:#1}}}
\AtBeginDocument{\providecommand\figref[1]{\ref{fig:#1}}}
\AtBeginDocument{\providecommand\tabref[1]{\ref{tab:#1}}}
\AtBeginDocument{\providecommand\lemref[1]{\ref{lem:#1}}}
\providecommand{\tabularnewline}{\\}
\RS@ifundefined{subsecref}
  {\newref{subsec}{name = \RSsectxt}}
  {}
\RS@ifundefined{thmref}
  {\def\RSthmtxt{theorem~}\newref{thm}{name = \RSthmtxt}}
  {}
\RS@ifundefined{lemref}
  {\def\RSlemtxt{lemma~}\newref{lem}{name = \RSlemtxt}}
  {}

\numberwithin{equation}{section}
\numberwithin{figure}{section}
\numberwithin{table}{section}
\theoremstyle{remark}
\newtheorem*{notation*}{\protect\notationname}
\theoremstyle{plain}
\newtheorem{thm}{\protect\theoremname}[section]
\theoremstyle{remark}
\newtheorem{rem}{\protect\remarkname}[section]
\theoremstyle{definition}
\newtheorem{example}{\protect\examplename}[section]
\theoremstyle{plain}
\newtheorem{assumption}{\protect\assumptionname}
\theoremstyle{plain}
\newtheorem{prop}{\protect\propositionname}[section]
\theoremstyle{plain}
\newtheorem{lem}{\protect\lemmaname}[section]

\setlist[enumerate,1]{label=\upshape{(\roman*)}, ref=(\roman*)}
\setlist[enumerate,2]{label=\upshape{(\alph*)}, ref=(\alph*)}
\setlist[enumerate,3]{label=\upshape{\roman*.}, ref=\roman*}

\makeatletter
  \@ifpackageloaded{refstyle}{}{\usepackage{refstyle}}
\makeatother

\newref{thm}{name = Theorem~}
\newref{prop}{name = Proposition~}
\newref{lem}{name = Lemma~}
\newref{cor}{name = Corollary~}
\newref{ass}{name = }
\newref{exa}{name = Example~}
\newref{rem}{name = Remark~}
\newref{def}{name = Definition~}

\newref{eq}{refcmd = (\ref{#1})}

\newref{sec}{name = Section~}
\newref{sub}{name = Section~}
\newref{subsec}{name = Section~}
\newref{app}{name = Appendix~}
\newref{supp}{name = Online Appendix~}

\newref{fig}{name = Figure~}
\newref{tbl}{name = Table~}
\newref{tab}{name = Table~}

\makeatletter
  \@ifpackageloaded{mathrsfs}{}{\usepackage{mathrsfs}}
\makeatother
\usepackage{amsbsy}

\date{}

\usepackage{nextpage}

\allowdisplaybreaks[1]

\usepackage{scalefnt}
\usepackage[group-separator={,}]{siunitx}
\newcommand\smaller[2][0.85]{{\scalefont{#1}#2}}
\newcommand{\ass}[1]{{\upshape{\smaller[0.76]{#1}}}}

\newcommand{\assumpname}[1]{%
  \renewcommand{\theassumption}{\ass{#1}}%
}

\makeatletter
\newsavebox{\@brx}
\newcommand{\dbllangle}[1][]{\savebox{\@brx}{\(\m@th{#1\langle}\)}%
  \mathopen{\copy\@brx\kern-0.5\wd\@brx\usebox{\@brx}}}
\newcommand{\dblrangle}[1][]{\savebox{\@brx}{\(\m@th{#1\rangle}\)}%
  \mathclose{\copy\@brx\kern-0.5\wd\@brx\usebox{\@brx}}}
\makeatother

\usepackage{changepage}

\theoremstyle{definition}
\newtheorem{examplex}{Example}
\renewenvironment{example}
  {\pushQED{\qed}\examplex}
  {\popQED\endexamplex}

\numberwithin{examplex}{section}

\newcommand{\exname}[1]{%
  \renewcommand{\theexamplex}{{#1}}%
}

\newcounter{subremark}[rem]
\renewcommand{\thesubremark}{(\roman{subremark})}

\newcommand{\subremark}{%
  \refstepcounter{subremark}%
  \thesubremark{}.%
}

\newcounter{savedexnumber}

\newcommand{\saveexamplex}{%
  \setcounter{savedexnumber}{\value{examplex}}
}

\newcommand{\restoreexamplex}{%
  \setcounter{examplex}{\value{savedexnumber}}
  \numberwithin{examplex}{section}
}

\usepackage{mathdots}

\usepackage[usestackEOL]{stackengine}
\setstackgap{S}{5pt}
\newcommand{\authaffil}[2]{\Shortunderstack{#1\\\small{#2}}}

\usepackage{needspace}

\newref{lr}{name = \ass{LR.}}
\newref{lrp}{name = \ass{LR.}}

\makeatother

\addto\captionsbritish{\renewcommand{\assumptionname}{Assumption}}
\addto\captionsbritish{\renewcommand{\examplename}{Example}}
\addto\captionsbritish{\renewcommand{\lemmaname}{Lemma}}
\addto\captionsbritish{\renewcommand{\notationname}{Notation}}
\addto\captionsbritish{\renewcommand{\propositionname}{Proposition}}
\addto\captionsbritish{\renewcommand{\remarkname}{Remark}}
\addto\captionsbritish{\renewcommand{\theoremname}{Theorem}}
\addto\captionsenglish{\renewcommand{\assumptionname}{Assumption}}
\addto\captionsenglish{\renewcommand{\examplename}{Example}}
\addto\captionsenglish{\renewcommand{\lemmaname}{Lemma}}
\addto\captionsenglish{\renewcommand{\notationname}{Notation}}
\addto\captionsenglish{\renewcommand{\propositionname}{Proposition}}
\addto\captionsenglish{\renewcommand{\remarkname}{Remark}}
\addto\captionsenglish{\renewcommand{\theoremname}{Theorem}}
\providecommand{\assumptionname}{Assumption}
\providecommand{\examplename}{Example}
\providecommand{\lemmaname}{Lemma}
\providecommand{\notationname}{Notation}
\providecommand{\propositionname}{Proposition}
\providecommand{\remarkname}{Remark}
\providecommand{\theoremname}{Theorem}

\begin{document}


\global\long\def\uwrite#1#2{\underset{#2}{\underbrace{#1}} }%

\global\long\def\blw#1{\ensuremath{\underline{#1}}}%

\global\long\def\abv#1{\ensuremath{\overline{#1}}}%

\global\long\def\vect#1{\mathbf{#1}}%


\global\long\def\smlseq#1{\{#1\} }%

\global\long\def\seq#1{\left\{  #1\right\}  }%

\global\long\def\smlsetof#1#2{\{#1\mid#2\} }%

\global\long\def\setof#1#2{\left\{  #1\mid#2\right\}  }%


\global\long\def\goesto{\ensuremath{\rightarrow}}%

\global\long\def\ngoesto{\ensuremath{\nrightarrow}}%

\global\long\def\uto{\ensuremath{\uparrow}}%

\global\long\def\dto{\ensuremath{\downarrow}}%

\global\long\def\uuto{\ensuremath{\upuparrows}}%

\global\long\def\ddto{\ensuremath{\downdownarrows}}%

\global\long\def\ulrto{\ensuremath{\nearrow}}%

\global\long\def\dlrto{\ensuremath{\searrow}}%


\global\long\def\setmap{\ensuremath{\rightarrow}}%

\global\long\def\elmap{\ensuremath{\mapsto}}%

\global\long\def\compose{\ensuremath{\circ}}%

\global\long\def\cont{C}%

\global\long\def\cadlag{D}%

\global\long\def\Ellp#1{\ensuremath{\mathcal{L}^{#1}}}%


\global\long\def\naturals{\ensuremath{\mathbb{N}}}%

\global\long\def\reals{\mathbb{R}}%

\global\long\def\complex{\mathbb{C}}%

\global\long\def\rationals{\mathbb{Q}}%

\global\long\def\integers{\mathbb{Z}}%


\global\long\def\abs#1{\ensuremath{\left|#1\right|}}%

\global\long\def\smlabs#1{\ensuremath{\lvert#1\rvert}}%
 
\global\long\def\bigabs#1{\ensuremath{\bigl|#1\bigr|}}%
 
\global\long\def\Bigabs#1{\ensuremath{\Bigl|#1\Bigr|}}%
 
\global\long\def\biggabs#1{\ensuremath{\biggl|#1\biggr|}}%

\global\long\def\norm#1{\ensuremath{\left\Vert #1\right\Vert }}%

\global\long\def\smlnorm#1{\ensuremath{\lVert#1\rVert}}%
 
\global\long\def\bignorm#1{\ensuremath{\bigl\|#1\bigr\|}}%
 
\global\long\def\Bignorm#1{\ensuremath{\Bigl\|#1\Bigr\|}}%
 
\global\long\def\biggnorm#1{\ensuremath{\biggl\|#1\biggr\|}}%

\global\long\def\floor#1{\left\lfloor #1\right\rfloor }%
\global\long\def\smlfloor#1{\lfloor#1\rfloor}%

\global\long\def\ceil#1{\left\lceil #1\right\rceil }%
\global\long\def\smlceil#1{\lceil#1\rceil}%


\global\long\def\Union{\ensuremath{\bigcup}}%

\global\long\def\Intsect{\ensuremath{\bigcap}}%

\global\long\def\union{\ensuremath{\cup}}%

\global\long\def\intsect{\ensuremath{\cap}}%

\global\long\def\pset{\ensuremath{\mathcal{P}}}%

\global\long\def\clsr#1{\ensuremath{\overline{#1}}}%

\global\long\def\symd{\ensuremath{\Delta}}%

\global\long\def\intr{\operatorname{int}}%

\global\long\def\cprod{\otimes}%

\global\long\def\Cprod{\bigotimes}%


\global\long\def\smlinprd#1#2{\ensuremath{\langle#1,#2\rangle}}%

\global\long\def\inprd#1#2{\ensuremath{\left\langle #1,#2\right\rangle }}%

\global\long\def\orthog{\ensuremath{\perp}}%

\global\long\def\dirsum{\ensuremath{\oplus}}%


\global\long\def\spn{\operatorname{sp}}%

\global\long\def\rank{\operatorname{rk}}%

\global\long\def\proj{\operatorname{proj}}%

\global\long\def\tr{\operatorname{tr}}%

\global\long\def\vek{\operatorname{vec}}%

\global\long\def\diag{\operatorname{diag}}%

\global\long\def\col{\operatorname{col}}%


\global\long\def\smpl{\ensuremath{\Omega}}%

\global\long\def\elsmp{\ensuremath{\omega}}%

\global\long\def\sigf#1{\mathcal{#1}}%

\global\long\def\sigfield{\ensuremath{\mathcal{F}}}%
\global\long\def\sigfieldg{\ensuremath{\mathcal{G}}}%

\global\long\def\flt#1{\mathcal{#1}}%

\global\long\def\filt{\mathcal{F}}%
\global\long\def\filtg{\mathcal{G}}%

\global\long\def\Borel{\ensuremath{\mathcal{B}}}%

\global\long\def\cyl{\ensuremath{\mathcal{C}}}%

\global\long\def\nulls{\ensuremath{\mathcal{N}}}%

\global\long\def\gauss{\mathfrak{g}}%

\global\long\def\leb{\mathfrak{m}}%


\global\long\def\prob{\ensuremath{\mathbb{P}}}%

\global\long\def\Prob{\ensuremath{\mathbb{P}}}%

\global\long\def\Probs{\mathcal{P}}%

\global\long\def\PROBS{\mathcal{M}}%

\global\long\def\expect{\ensuremath{\mathbb{E}}}%

\global\long\def\Expect{\ensuremath{\mathbb{E}}}%

\global\long\def\probspc{\ensuremath{(\smpl,\filt,\Prob)}}%


\global\long\def\iid{\ensuremath{\textnormal{i.i.d.}}}%

\global\long\def\as{\ensuremath{\textnormal{a.s.}}}%

\global\long\def\asp{\ensuremath{\textnormal{a.s.p.}}}%

\global\long\def\io{\ensuremath{\ensuremath{\textnormal{i.o.}}}}%

\newcommand\independent{\protect\mathpalette{\protect\independenT}{\perp}}
\def\independenT#1#2{\mathrel{\rlap{$#1#2$}\mkern2mu{#1#2}}}

\global\long\def\indep{\independent}%

\global\long\def\distrib{\ensuremath{\sim}}%

\global\long\def\distiid{\ensuremath{\sim_{\iid}}}%

\global\long\def\asydist{\ensuremath{\overset{a}{\distrib}}}%

\global\long\def\inprob{\ensuremath{\overset{p}{\goesto}}}%

\global\long\def\inprobu#1{\ensuremath{\overset{#1}{\goesto}}}%

\global\long\def\inas{\ensuremath{\overset{\as}{\goesto}}}%

\global\long\def\eqas{=_{\as}}%

\global\long\def\inLp#1{\ensuremath{\overset{\Ellp{#1}}{\goesto}}}%

\global\long\def\indist{\ensuremath{\overset{d}{\goesto}}}%

\global\long\def\eqdist{=_{d}}%

\global\long\def\wkc{\ensuremath{\rightsquigarrow}}%

\global\long\def\wkcu#1{\overset{#1}{\ensuremath{\rightsquigarrow}}}%

\global\long\def\plim{\operatorname*{plim}}%


\global\long\def\var{\operatorname{var}}%

\global\long\def\lrvar{\operatorname{lrvar}}%

\global\long\def\cov{\operatorname{cov}}%

\global\long\def\corr{\operatorname{corr}}%

\global\long\def\bias{\operatorname{bias}}%

\global\long\def\MSE{\operatorname{MSE}}%

\global\long\def\med{\operatorname{med}}%

\global\long\def\avar{\operatorname{avar}}%

\global\long\def\se{\operatorname{se}}%

\global\long\def\sd{\operatorname{sd}}%


\global\long\def\nullhyp{H_{0}}%

\global\long\def\althyp{H_{1}}%

\global\long\def\ci{\mathcal{C}}%


\global\long\def\simple{\mathcal{R}}%

\global\long\def\sring{\mathcal{A}}%

\global\long\def\sproc{\mathcal{H}}%

\global\long\def\Wiener{\ensuremath{\mathbb{W}}}%

\global\long\def\sint{\bullet}%

\global\long\def\cv#1{\left\langle #1\right\rangle }%

\global\long\def\smlcv#1{\langle#1\rangle}%

\global\long\def\qv#1{\left[#1\right]}%

\global\long\def\smlqv#1{[#1]}%


\global\long\def\trans{\mathsf{T}}%

\global\long\def\indic{\ensuremath{\mathbf{1}}}%

\global\long\def\Lagr{\mathcal{L}}%

\global\long\def\grad{\nabla}%

\global\long\def\pmin{\ensuremath{\wedge}}%
\global\long\def\Pmin{\ensuremath{\bigwedge}}%

\global\long\def\pmax{\ensuremath{\vee}}%
\global\long\def\Pmax{\ensuremath{\bigvee}}%

\global\long\def\sgn{\operatorname{sgn}}%

\global\long\def\argmin{\operatorname*{argmin}}%

\global\long\def\argmax{\operatorname*{argmax}}%

\global\long\def\Rp{\operatorname{Re}}%

\global\long\def\Ip{\operatorname{Im}}%

\global\long\def\deriv{\ensuremath{\mathrm{d}}}%

\global\long\def\diffnspc{\ensuremath{\deriv}}%

\global\long\def\diff{\ensuremath{\,\deriv}}%

\global\long\def\i{\ensuremath{\mathrm{i}}}%

\global\long\def\e{\mathrm{e}}%

\global\long\def\sep{,\ }%

\global\long\def\defeq{\coloneqq}%

\global\long\def\eqdef{\eqqcolon}%

\global\long\def\err{\varepsilon}%

\global\long\def\mset#1{\mathcal{#1}}%

\global\long\def\largedec#1{\mathbf{#1}}%

\global\long\def\z{\largedec z}%

\newcommandx\Ican[1][usedefault, addprefix=\global, 1=]{I_{#1}^{\ast}}%

\global\long\def\jsr{{\scriptstyle \mathrm{JSR}}}%

\global\long\def\cjsr{{\scriptstyle \mathrm{CJSR}}}%

\global\long\def\rsr{{\scriptstyle \mathrm{RJSR}}}%

\global\long\def\ctspc{\mathscr{M}}%

\global\long\def\b#1{\boldsymbol{#1}}%

\global\long\def\pseudy{\text{\ensuremath{\abv y}}}%

\global\long\def\regcoef{\kappa}%

\global\long\def\adj{\operatorname{adj}}%

\global\long\def\llangle{\dbllangle}%

\global\long\def\rrangle{\dblrangle}%

\global\long\def\smldblangle#1{\ensuremath{\llangle#1\rrangle}}%

\newcommand{\casecens}{{\upshape{(i)}}}

\newcommand{\caseclas}{{\upshape{(ii)}}}

\newcommand{\casestat}{{\upshape{(iii)}}}

\global\long\def\delcens{\mathrm{(i)}}%

\global\long\def\delclas{\mathrm{(ii)}}%

\global\long\def\smlf{f}%

\global\long\def\bigf{\b f}%

\global\long\def\fe{f}%

\global\long\def\z{\boldsymbol{z}}%

\global\long\def\ga{g}%

\global\long\def\co{\chi}%

\global\long\def\CO{\mathrm{X}}%

\global\long\def\EQ{\Theta}%

\global\long\def\eq{\theta}%

\global\long\def\ct{\psi}%

\global\long\def\cn{\mu}%

\global\long\def\cvar{{\cal M}}%

\global\long\def\codf{\chi}%

\global\long\def\eqdf{\chi_{2}}%

\global\long\def\ctdf{\chi_{1}}%

\global\long\def\set#1{\mathscr{#1}}%

\global\long\def\srp{\mathrm{H}}%

\global\long\def\srfn{h}%

\global\long\def\ch{\operatorname{co}}%

\global\long\def\im{\operatorname{im}}%

\global\long\def\lin{\mathrm{lin}}%

\global\long\def\state{\mathfrak{z}}%

\global\long\def\fespc{\mathscr{F}}%

\global\long\def\bigfspc{\pmb{\mathscr{F}}}%

\global\long\def\denspc{\mathscr{R}}%

\global\long\def\parspc{\mathscr{P}}%

\global\long\def\den{\varrho}%

\global\long\def\cden{\rho}%

\global\long\def\orths{\mathbb{O}}%

\global\long\def\sorths{\mathbb{O}^{+}}%

\global\long\def\vol{s}%

\global\long\def\trans{\top}%

\title{Identification in (Endogenously) Nonlinear SVARs\\ Is Easier Than
You Think}
\author{\authaffil{James A.\ Duffy\footnotemark[1]{}}{University of Oxford}\hspace{3cm}
\authaffil{Sophocles Mavroeidis\footnotemark[2]{}}{University
of Oxford}}
\date{\vspace*{0.3cm}April 2026}

\maketitle
\renewcommand*{\thefootnote}{\fnsymbol{footnote}}

\footnotetext[1]{Department\ of Economics and Corpus Christi College;
\texttt{james.duffy@economics.ox.ac.uk}}

\footnotetext[2]{Department of Economics and University College;
\texttt{sophocles.mavroeidis@economics.ox.ac.uk}}

\renewcommand*{\thefootnote}{\arabic{footnote}}

\setcounter{footnote}{0}
\begin{abstract}
\noindent We study identification in structural vector autoregressions
(SVARs) in which the endogenous variables enter nonlinearly on the
left-hand side of the model, a feature we term \emph{endogenous} \emph{nonlinearity},
to distinguish it from the more familiar case in which nonlinearity
arises only through exogenous or predetermined variables. This class
of models accommodates asymmetric impact multipliers, endogenous regime
switching, and occasionally binding constraints. We show that, under
weak regularity conditions, the model parameters and structural shocks
are (nonparametrically) identified up to an orthogonal transformation,
exactly as in a linear SVAR. Our results have the powerful implication
that most existing identification schemes for linear SVARs extend
directly to our nonlinear setting, with the number of restrictions
required to achieve exact identification remaining unchanged. We specialise
our results to piecewise affine SVARs, which provide a convenient
framework for the modelling of endogenous regime switching, and their
smooth transition counterparts. We illustrate our methodology with
an application to the nonlinear Phillips curve, providing a test for
the presence of nonlinearity that is robust to the choice of identifying
assumptions, and finding significant evidence for state-dependent
inflation dynamics.
\end{abstract}
\vfill

\noindent{}\thmref{shockid} supersedes a result that first appeared,
with rather stronger assumptions, as Theorem~A.1 in \citet[v2]{DM24}.

\thispagestyle{plain}

\pagenumbering{roman}

\newpage{}

\thispagestyle{plain}

\setcounter{tocdepth}{2}

{\singlespacing

\tableofcontents{}

}

\newpage{}

\newpage{}

\pagestyle{fancy}

\pagenumbering{arabic}

\section{Introduction}

For more than four decades, following the seminal contribution of
\citet{Sims80}, the linear structural vector autoregression (SVAR)
\begin{align}
\Phi_{0}z_{t} & =c+\sum_{i=1}^{k}\Phi_{i}z_{t-i}+\err_{t}, & \err_{t} & \distiid[0,I_{p}],\label{eq:svar-intro}
\end{align}
has played a central role in empirical macroeconomics. This is a
dynamic linear simultaneous equations model (SEM), in which the $p$
endogenous variables $z_{t}$ are jointly determined as a function
of their past values and the $p$ (mutually uncorrelated) structural
shocks $\err_{t}$. The latter are regarded as the exogenous inputs
to the system, so that causality is understood to run from these shocks
to current and future values of $z_{t}$, and a key object of interest
is the mapping between $\err_{t}$ and $z_{t+h}$ for $h\geq0$: the
(structural) impulse response function.

In this context, a fundamental result that characterises the extent
to which the data is informative about the model parameters, and thus
also about those impulse responses, may be phrased heuristically as
follows:
\begin{quote}
\begin{enumerate}[label=(ID),ref=(ID)]
\item \label{enu:id} Data on $\{z_{t}\}$ is sufficient to identify the
linear SVAR parameters $(c,\{\Phi_{i}\}_{i=0}^{p})$, and the structural
shocks $\err_{t}$, up to, and only up to, an orthogonal matrix $Q$.
\end{enumerate}
\end{quote}
In light of this, what might be termed the `SVAR identification problem'
becomes one of finding sufficient additional restrictions on that
matrix $Q$, so as to pin down, wholly or partially, the model parameters.
The literature since has developed a variety of ways of using macroeconomic
theory to generate such restrictions, based e.g.\ on the relative
timing of shocks, the signs of their effects on impact, their medium-
and long-run effects, and their correlation with external instruments
(for a textbook discussion of which, see \citealp{KL17book}).

The linearity of \eqref{svar-intro} is convenient, but inherently
limiting as to the nature of the dynamics that can be modelled. In
particular, it has the rather undesirable implication that the response
of the economy to shocks must be the same \emph{irrespective} of the
phase of the business cycle: so that e.g.\ an aggregate demand shock
has exactly the same effect on unemployment and inflation in the depths
of a recession, when there is considerable slack in the labour market,
as it would during periods of expansion. The substantial literature
on nonlinear (S)VAR models has arisen partly to address these limitations
(see e.g.\ \citealp{Chan09}; \citealp{TTG10}; \citealp{HT13},
for surveys). These allow the parameters of the SVAR at time $t$
to depend on an exogenous (or if not wholly exogenous, at least \emph{predetermined})
regime-switching process $s_{t-1}$, as e.g.\ in\footnote{In many treatments of these models, the regime indicator in \eqref{msvar}
is denoted as $s_{t}$, rather than $s_{t-1}$. However, a feature
of these models is that the regime is always determined prior to the
realisation of $\err_{t}$, and may thus be regarded as measurable
with respect to time-$(t-1)$ information; we have written $s_{t-1}$
to make this clearer.}
\begin{equation}
\Phi_{0}(s_{t-1})z_{t}=c(s_{t-1})+\sum_{i=1}^{k}\Phi_{i}(s_{t-1})z_{t-i}+\err_{t},\label{eq:msvar}
\end{equation}
where often $s_{t}\in\{1,\ldots,L\}$ takes finitely many values,
and each $\Phi_{i}(s)=\sum_{\ell=1}^{L}\pi_{\ell}(s)\Phi_{i}^{(\ell)}$
switches, or smoothly transitions, between the parameters of the $L$
`regimes'; here each $\pi_{\ell}(s)\in[0,1]$, with $\sum_{\ell=1}^{L}\pi_{\ell}(s)=1$.
The evolution of $\{s_{t}\}$ may be modelled as an exogenous Markov
chain (as in a Markov switching model), possibly with state-dependent
transition probabilities, or as a function of certain predetermined
variables (such as an element of $z_{t-i}$ for some $i\geq1$, as
in a typical `threshold autoregressive' model); but in any case,
$s_{t-1}$ must be determined prior to the realisation of $\err_{t}$.
We therefore refer to these henceforth as \emph{exogenous} regime-switching
SVARs. (This characterisation applies to time-varying parameter VARs,
in which $\{s_{t}\}$ is also some exogenous but possibly nonstationary
process, such as a random walk.)

While models of the form \eqref{msvar} enjoy greatly enriched dynamics
relative to \eqref{svar-intro}, here the possibility of regime switching
exacerbates the identification problem. Indeed the counterpart of
\ref{enu:id} for general Markov-switching models is that, \emph{conditional
on the regime} $s_{t-1}=s$, the parameters of \eqref{msvar} are
identified up to an orthogonal matrix $Q(s)$. Since this matrix may
vary with $s\in\{1,\ldots,L\}$, the number of unidentified parameters,
and thus the number of restrictions needed to deliver (exact) identification,
scales proportionally with the number of states $L$. In practice,
this may necessitate replicating a common set of $p(p-1)/2$ restrictions
across \emph{all} $L$ regimes (see e.g.\  \citealp{RRWZ05}, Sec.\ II;
\citealp{SZ06AER}, Sec.\ III), yielding $Lp(p-1)/2$ restrictions
in total. Similarly, in their two-regime STVAR model, \citet[p.~4]{AG12AEJ}
impose a Cholesky ordering on the elements of $z_{t}$ in each regime.

The exogeneity of the regime (i.e.\ of $s_{t-1}$) moreover restricts
the kinds of nonlinearities that may be exhibited by the model's
impulse responses. Notably, since each regime is itself a linear SVAR,
the immediate effects of the shocks (i.e.\ the impact multipliers)
must be \emph{linear} in $\err_{t}$: which in particular rules out
the possibility of sign-dependent asymmetries. It also renders \eqref{msvar}
unable to accommodate occasionally binding constraints, such as the
zero lower bound (ZLB) constraint on short-term nominal interest rates,
because the model requires the regime (whether `constrained' or
`unconstrained') to be determined \emph{prior} to realising the
value of the potentially constrained variable -- whereas, as a matter
of logic, it ought to be the value of that variable which determines
whether the model is in fact in the constrained or unconstrained regime
(see \citealp{AMSV21}).

Recently, \citet{SM21} and \citet{AMSV21} introduced the first SVAR
models involving what we here refer to as \emph{endogenous} regime
switching, which are notably distinguished from the earlier literature
on the nonlinear SVARs of the form \eqref{msvar} in that they permit
the autoregressive `regime' to be determined \emph{jointly} with
the values of the endogenous variables. For example, the `censored
and kinked SVAR' (CKSVAR) of \citet{SM21} takes the form 
\[
\phi_{0}^{+}y_{t}^{+}+\phi_{0}^{-}y_{t}^{-}+\Phi_{0}^{x}x_{t}=c+\sum_{i=1}^{k}[\phi_{i}^{+}y_{t-i}^{+}+\phi_{i}^{-}y_{t-i}^{-}+\Phi_{i}^{x}x_{t-i}]+\err_{t}.
\]
where $y_{t}^{+}$ and $y_{t}^{-}$ denote the positive and negative
parts of $y_{t}$ (a scalar), and $x_{t}$ is $(p-1)$-dimensional.
In this model there are two contemporaneous regimes: one associated
with $y_{t}>0$ (the `unconstrained' regime, in the ZLB setting),
and the other with $y_{t}\leq0$ (the `constrained' regime), and
in every period the model is solved \emph{simultaneously} for the
current values of $y_{t}$ and $x_{t}$, and for the applicable regime
(as depends on the sign of $y_{t}$). Thus in situations where the
$\err_{t}=0$ would entail a solution of $y_{t}=0$ (or approximately
so), this allows the impact multipliers of $\err_{t}$ to be asymmetric,
being dependent on which regime they push the model into.

Building on these developments, this paper proposes a new class of
nonlinear SVAR models, which have the general form
\begin{align}
\fe_{0}(z_{t}) & =\sum_{i=1}^{k}\fe_{i}(z_{t-i})+\err_{t}\label{eq:SVAR-intro}
\end{align}
where each $\fe_{i}:\reals^{p}\setmap\reals^{p}$ is a continuous,
possibly nonlinear function, with $\fe_{0}$ being invertible; we
refer to these models as `endogenously nonlinear', in view of the
nonlinearities on the l.h.s. Because it is not tied to any particular
functional form, \eqref{SVAR-intro} also offers a great deal of flexibility
in its dynamics, comparable to that offered by \eqref{msvar}. This
framework readily encompasses the CKSVAR, which corresponds to a special
case in which each $\fe_{i}$ is piecewise linear. More general models
with endogenous switching between several regimes may be straightforwardly
encompassed within the framework \eqref{SVAR-intro}, by specifying
\begin{equation}
\fe_{0}(z)=\sum_{\ell=1}^{L}\indic\{z\in\set Z^{(\ell)}\}(\bar{\phi}_{0}^{(\ell)}+\Phi_{0}^{(\ell)}z)\label{eq:f0-regime}
\end{equation}
to be an invertible, (continuous) piecewise affine function, where
$\{\set Z^{(\ell)}\}_{\ell=1}^{L}$ is a convex partition of $\reals^{p}$,
and the current regime $\ell_{t}$ corresponds to the element of that
partition for which $z_{t}\in\set Z^{(\ell_{t})}$.

The principal contribution of this paper is to characterise observational
equivalence in the setting of the following, slightly more general
formulation of \eqref{SVAR-intro},
\begin{align}
\fe_{0}(z_{t}) & =\b{\fe}_{1}(\b z_{t-1})+\err_{t}, & \err_{t} & \distiid[0,I_{p}],\label{eq:SVAR-nonsep}
\end{align}
where $\b z_{t-1}=(z_{t-1}^{\trans},\ldots,z_{t-k}^{\trans})^{\trans}$,
and $\b{\fe}_{1}:\reals^{kp}\setmap\reals^{p}$ need not be separable
in the lags of $z_{t}$ (\secref{svarident}). Remarkably, despite
the far greater flexibility afforded by the parametrisation of \eqref{SVAR-nonsep}
relative to the linear SVAR \eqref{svar-intro}, the fundamental identification
result \ref{enu:id} carries over to \eqref{SVAR-nonsep} essentially
unchanged. Under weak conditions on the functions $(\fe_{0},\b{\fe}_{1})$
and the distribution of the shocks, we have (\thmref{shockid}):
\begin{quote}
\begin{enumerate}[label=(ID$^\prime$),ref=(ID$^\prime$)]
\item \label{enu:nl-id} Data on $\{z_{t}\}$ is sufficient to identify
the nonlinear SVAR parameters $(\fe_{0},\b{\fe}_{1})$, and the structural
shocks $\err_{t}$, up to, and only up to, an orthogonal matrix $Q$.
\end{enumerate}
\end{quote}
This is a \emph{nonparametric} identification result, in the sense
that we do not suppose that $(\fe_{0},\b{\fe}_{1})$ have any particular
(known) parametric form. While its proof draws upon the microeconometrics
literature on nonlinear SEMs (see in particular \citealp{Matz09Ecta,Matz15Ecta};
\citealp{BH18Ecta}; \citealp{CGHP21JPE}), it constitutes a genuinely
novel result within that setting. \ref{enu:nl-id} has the powerful
implication that most of the existing identification results for linear
SVARs apply directly to the endogenously nonlinear SVAR, since in
both cases exact identification is a matter of imposing $p(p-1)/2$
restrictions sufficient to pin down $Q$.

There follows a discussion of the $L$-regime endogenous regime-switching
SVAR, which arises when $\fe_{0}$ is specified to have the piecewise
affine form \eqref{f0-regime}, and of how to verify the conditions
of \thmref{shockid} in this case (\secref{piecewise-affine}). (Here
we also suppose, mostly to provide a practically convenient parametrisation,
that the SVAR has the time-separable form \eqref{SVAR-intro}, with
each $\{\fe_{i}\}_{i=1}^{k}$ also specified to have the same functional
form as \eqref{f0-regime}.) In this context, our results imply that
it is sufficient, for the purposes of exact identification, to impose
identifying restrictions in \emph{only one of those $L$ regimes},
or even to distribute these in some way across those regimes. To obtain
smooth transitions between adjacent regimes, we propose to convolve
$\fe_{0}$ with a smooth kernel. This has the considerable advantage
of preserving the invertibility of $\fe_{0}$, whereas this may fail
if one attempts to smooth $\fe_{0}$ by the usual device of replacing
each indicator function in \eqref{f0-regime} by a smooth, cdf-like
function (as is very commonly done to produce `smooth transition'
(S)VARs).

Our methodology is illustrated by estimating an endogenously regime-switching
SVAR (in the log vacancy--unemployment ratio and inflation), to investigate
the possibility of a nonlinear Phillips curve relationship (\secref{phillipscurve})
that was recently proposed by \citet{BenignoEggertsson2023} to explain
the recent post-pandemic inflation surge. In particular, our identification
results allow us to examine the evidence for nonlinearity in a manner
that is robust to alternative identification assumptions, thus shedding
light on the recent debate between \citet{BenignoEggertsson2023}
and \citet{BeaudryHouPortier2025}.

Finally, we provide an extension of our results to the augmented model
\begin{align*}
\fe_{0}(z_{t}) & =\b{\fe}_{1}(\b z_{t-1}^{(1)},\b z_{t-1}^{(2)},v_{t-1})+\sigma(\b z_{t-1}^{(2)},v_{t-1})\err_{t}, & \err_{t} & \distiid[0,I_{p}],
\end{align*}
where $(\b z_{t-1}^{(1)},\b z_{t-1}^{(2)})$ is some partitioning
of $\b z_{t-1}$, and $\{v_{t}\}$ is a strictly exogenous process,
in the sense of being independent of $(\b z_{0},\{\err_{t}\})$ (\secref{hetero}).
Here $\sigma(\cdot)$ is a diagonal matrix (with strictly positive
entries), which allows the conditional variances of the structural
shocks to depend on certain predetermined variables. In this setting,
we show that \ref{enu:nl-id} continues to provide a valid characterisation
of the identification of $(\fe_{0},\b{\fe}_{1})$, and that $Q$ may
moreover be subject to further restrictions, if there is sufficient
variability in the (diagonal) entries of $\sigma(\b z_{t-1}^{(2)},v_{t-1})$;
these correspond to exactly the restrictions familiar from the \emph{linear}
SVAR literature on `identification by heteroskedasticity'. Our main
result here (\thmref{svarhetero}) not only accommodates both: (i)
ARCH-type heteroskedasticity; and (ii) the possible dependence of
the r.h.s.\ of the SVAR on an exogenous process $\{v_{t}\}$; but
also (iii) permits $\b{\fe}_{1}$ to be discontinuous in some of its
arguments (specifically, $\b z_{t-1}^{(2)}$ and $v_{t-1}$).

\begin{notation*}
$e_{m,i}$ denotes the $i$th column of an $m\times m$ identity matrix;
when $m$ is clear from the context, we write this simply as $e_{i}$.
$\smlnorm{\cdot}$ denotes the Euclidean norm on $\reals^{m}$. Matrix
norms are always those induced by the corresponding vector norm. For
a function $g:\reals^{m}\setmap\reals^{n}$, $Dg(u_{0})=[(\partial g_{i}/\partial u_{j})(u_{0})]$
denotes the $(n\times m)$ Jacobian (matrix) of $g(u)$ at $u=u_{0}$.
A `density' is always a density with respect to Lebesgue measure,
unless otherwise stated.
\end{notation*}

\section{Observational equivalence and identification}

\label{sec:svarident}

\subsection{The linear SVAR: a brief review}

\label{subsec:linearsvar}

Our point of departure is the linear SVAR, in which the observed series
$\{z_{t}\}$ are regarded as being generated linearly from an underlying
$p$-dimensional sequence of \emph{structural shocks} $\{\err_{t}\}$,
each of which have an economic interpretation (as e.g.\ an aggregate
supply shock, a monetary policy shock, etc.) That is, for some $k\in\naturals$,
\begin{equation}
\Phi_{0}z_{t}=c+\sum_{i=1}^{k}\Phi_{i}z_{t-i}+\err_{t}\eqdef c+\b{\Phi}_{1}\b z_{t-1}+\err_{t}\label{eq:linSVAR}
\end{equation}
where $z_{t}$ and $\err_{t}$ are $\reals^{p}$-valued, and to permit
a more compact presentation we have defined $\b{\Phi}_{1}\defeq[\Phi_{1},\ldots,\Phi_{k}]$
and $\b z_{t-1}\defeq(z_{t-1}^{\trans},\ldots,z_{t-k}^{\trans})^{\trans}$.

Observational equivalence in this setting being well understood (see
e.g.\ \citealp{Hamilton94}, Ch.~11; \citealp{Lut07}, Ch.~9),
our purpose here is to briefly review this in a manner that facilitates
the comparison with our results for the endogenously nonlinear SVAR,
which are developed in \subsecref{nonlinearsvar} below. To simplify
the problem, we suppose that $\{\err_{t}\}_{t\in\integers}$ is i.i.d.,
with a (Lebesgue) density $\den\in\denspc$, normalised to have $\expect\err_{t}=0$
and $\expect\err_{t}\err_{t}^{\trans}=I_{p}$. By the Markov property
the joint density of $\{z_{t}\}_{t=1}^{T}$, conditional on $\b z_{0}$,
is simply the product of the conditional densities of $z_{t}\mid\b z_{t-1}$,
for $t\in\{1,\ldots,T\}$. Under our assumptions, this density is
time-invariant, and equals
\[
\varphi_{z_{t}\mid\b z_{t-1}}(\xi\mid\b{\xi}_{-1})=\den(\Phi_{0}\xi-\b{\Phi}_{1}\b{\xi}_{-1})\cdot\smlabs{\det\Phi_{0}},
\]
where $\xi\in\reals^{p}$, $\b{\xi}_{-1}\in\reals^{kp}$. Accordingly,
we say that two alternative parameterisations of the linear SVAR,
$(c,\Phi_{0},\b{\Phi}_{1},\den)$ and $(\tilde{c},\tilde{\Phi}_{0},\tilde{\b{\Phi}}_{1},\tilde{\den})$,
are \emph{observationally equivalent} if they imply identical conditional
densities $\varphi_{z_{t}\mid\b z_{t-1}}$; in which case they also
yield identical (conditional) likelihoods, for every possible realisation
of $\{z_{t}\}$.

We then have the following well-known result, that data on $\{z_{t}\}$
identifies the SVAR coefficients $(c,\Phi_{0},\b{\Phi}_{1})$ up to,
and only up to, an orthogonal transformation. Let $\orths(p)$ denote
the set of $p\times p$ orthogonal matrices.
\begin{thm}
\label{thm:linSVAR}Let $(\tilde{c},\tilde{\Phi}_{0},\tilde{\b{\Phi}}_{1})\in\reals^{p}\times\reals^{p\times p}\times\reals^{p\times kp}$.
Then there exists a $\tilde{\den}\in\denspc$ such that $(\tilde{c},\tilde{\Phi}_{0},\tilde{\b{\Phi}}_{1},\tilde{\den})$
is observationally equivalent to $(c,\Phi_{0},\b{\Phi}_{1},\den)$
in the model \eqref{linSVAR}, if and only if there exists a $Q\in\orths(p)$
such that
\begin{align*}
\tilde{c} & =Qc & \tilde{\Phi}_{0} & =Q\Phi_{0} & \tilde{\b{\Phi}}_{1} & =Q\b{\Phi}_{1}.
\end{align*}
\end{thm}
\begin{rem}
\label{rem:linearid}\subremark{} Versions of this result, or equivalent
characterisations thereof, have long been utilised in the analysis
of linear SVARs, and linear simultaneous equations models (SEMs).
This particular characterisation leads naturally to the `orthogonal
reduced-form parametrisation' (\citealp{ARRW18Ecta}, Sec.\ 2.3)
of the SVAR, in terms of the (unidentified) $Q\in\orths(p)$ and the
(identified) reduced form parameters ($\Phi_{0}^{-1}\b{\Phi}_{1}$
and $\Phi_{0}^{\trans}\Phi_{0}$), which has proved fruitful for the
analysis of sign-restricted SVARs (\citealp{Faust98CR}; \citealp{Uhlig98CR,Uhlig05JME};
\citealp{ARRW18Ecta}), and the formulation of rank conditions for
global identification (\citealp{RRWZ10REStud}).

\subremark{}\label{subrem:wn} The preceding follows as a corollary
to \thmref{shockid} below, albeit that result is proved under stronger
regularity conditions on the allowable set of densities $\denspc$.
Because of the linearity of \eqref{linSVAR}, the same result also
holds under weaker conditions on the model than we have maintained
here. For example, we could require $\{\err_{t}\}$ merely to be stationary
white noise, since all that is really needed to identify the reduced
form parameters is the orthogonality of $\err_{t}$ from $\b z_{t-1}$.
On the other hand, the assumption that $\{\err_{t}\}$ is an i.i.d.\ process,
often with a \emph{known} (often Gaussian) distribution is common
in empirical work, particularly in the context of Bayesian SVARs,
and even in discussions of identification in these models (as in e.g.\ \citealp{RRWZ10REStud}).

\subremark{} Here we have maintained only that the individual elements
of $\err_{t}=(\err_{1t},\ldots,\err_{pt})^{\trans}$ are contemporaneously
\emph{orthogonal}, rather than being independent. We thereby exclude
the possibility, highlighted in a strand of the linear SVAR literature
(e.g.\ \citealp{LMS17JoE}; \citealp{GMR20REStud}), of exploiting
the additional restrictions available when the shocks are independent
\emph{and} non-Gaussian, to strengthen the above result to one in
which the SVAR coefficients are identified up to an unknown (signed)
permutation matrix.

\subremark{} Let $k_{0}$ denote the true lag order of the SVAR,
i.e.\ the greatest $i\in\naturals$ such that $\Phi_{i}\neq0$. We
have implicitly maintained that this is less than or equal to $k$,
which may therefore be interpreted as an upper bound on the true lag
order of the model. In this sense, \thmref{linSVAR} does not assume
knowledge of the true lag order $k_{0}$ of the SVAR, but merely of
some finite upper bound $k$ thereof.
\end{rem}

\subsection{The (endogenously) nonlinear SVAR}

\label{subsec:nonlinearsvar}

We now seek to extend \thmref{linSVAR} to the setting of the following
(endogenously) \emph{nonlinear} SVAR
\begin{equation}
\fe_{0}(z_{t})=\b{\fe}_{1}(\z_{t-1})+\err_{t}\label{eq:nlVAR}
\end{equation}
where $\fe_{0}:\reals^{p}\setmap\reals^{p}$ is invertible, and $\b{\fe}_{1}:\reals^{kp}\setmap\reals^{p}$.
(As a convenient location normalisation, we set $\fe_{0}(0)=0$.)
This model evidently nests \eqref{linSVAR}, by taking $\fe_{0}(z_{t})=\Phi_{0}z_{t}$
and $\b{\fe}_{1}(\b z_{t-1})=c+\sum_{i=1}^{k}\Phi_{i}z_{t-i}$. Another
important special case arises when $\b{\fe}_{1}(\b z_{t-1})=\sum_{i=1}^{k}\fe_{i}(z_{t-i})$
is additively time-separable, as considered in \citet{DM24}. But
while this separability facilitates an extension of the Granger--Johansen
representation theorem to these nonlinear SVARs, it is not necessary
for the results that follow. The only separability that we require
here is between $z_{t}$, $\b z_{t-1}$ and $\err_{t}$.

We develop the following running example throughout the rest of the
paper.
\begin{example}[nonlinear Phillips curve]
\label{exa:phillips} The Phillips curve is a key component of any
macroeconomic model. It provides a causal link between aggregate output
and prices, and is thus essential in modelling the monetary policy
transmission mechanism. Its name derives from the seminal contribution
of \citet{Phillips1958}, who proposed the following simple \emph{nonlinear}
relationship between (wage) inflation, $\pi_{t}^{w}$, and labour
market tightness as measured by the unemployment rate, $u_{t}$,
\begin{equation}
\pi_{t}^{w}=a+b\left(\frac{1}{u_{t}}\right)^{c}.\label{eq:Phillips}
\end{equation}
Several recent contributions have used the vacancy-to-unemployment
ratio, $\theta_{t}\defeq v_{t}/u_{t}$, as an alternative measure
of tightness, and price (instead of wage) inflation, $\pi_{t}$, see
e.g.\ \citet{BallLeighMishra2022}, \citet{BenignoEggertsson2023},
and \citet{BeaudryHouPortier2025}. These papers employ alternative
functional forms for \eqref{Phillips}, and introduce additional dynamics,
inflation expectations, and other shocks.\footnote{\citet{BallLeighMishra2022} use a third order polynomial in $\log\theta_{t}$,
\citet{BenignoEggertsson2023} a piecewise linear function in $\log\theta_{t}$
with a kink at $\theta_{t}=1$, and \citet{BeaudryHouPortier2025}
consider both these specifications.}

Here we consider a stylised version of the piecewise linear Phillips
curve proposed in \citet{BenignoEggertsson2023}, 
\begin{equation}
\pi_{t}-\pi=\begin{cases}
\kappa\log\theta_{t}+\eta_{\pi,t}, & \text{if }\theta_{t}\leq1\text{ (`normal')}\\
\kappa^{\mathrm{tight}}\log\theta_{t}+\eta_{\pi,t}, & \text{if }\theta_{t}>1\text{ (\text{`labour shortage')}}
\end{cases}\label{eq:NLPC}
\end{equation}
where $\pi$ denotes steady state or target inflation, and $\eta_{\pi,t}$
an exogenous shock. Despite differences in specification, the fundamental
identification problem in all such models remains the same. Insofar
as inflation and tightness may plausibly be determined \emph{simultaneously},
the r.h.s.\ of \eqref{NLPC} cannot (in general) be identified as
though it were a (nonlinear) regression. Simultaneous causation can
instead be addressed by incorporating both \eqref{NLPC}, and the
corresponding reverse (causal) model for the effect of inflation on
tightness, into an ($\reals^{2}$-valued) nonlinear function $\fe_{0}(z_{t})$,
where $z_{t}=(\log\theta_{t},\pi_{t})^{\trans}$, yielding a specification
for the l.h.s.\ of \eqref{nlVAR}. 
\end{example}

As noted in the introduction, we term \eqref{nlVAR} an \emph{endogenously}
\emph{nonlinear} SVAR, because of the possible nonlinearity on the
l.h.s.\ of the model, i.e.\ in the endogenous variables $z_{t}$.
Were the model instead required to be linear in the endogenous variables,
so that $\fe_{0}(z)=\Phi_{0}z$, identification of the model parameters
would be as straightforward as it is in the linear SVAR; and along
the lines of \remref{linearid}\ref{subrem:wn} above, the assumption
that $\{\err_{t}\}$ is independent across time could be weakened
to one of $\{\err_{t}\}$ being merely a martingale difference sequence
(with respect to the filtration generated by $\{z_{t}\}$). However,
in imposing linearity on the l.h.s., we would lose the possibilities
for endogenous regime switching, asymmetric impact multipliers, and
of handling occasionally binding constraints, which the general model
\eqref{nlVAR} affords. We accordingly want to permit $\fe_{0}$ to
be nonlinear: a consequence of which is that \emph{independence across
time} of $\{\err_{t}\}$ becomes necessary for the parameters of \eqref{nlVAR}
to be identified. (But note that there is \emph{no} requirement of
contemporaneous independence between the elements of $\err_{t}$.)

We thus continue to maintain that the structural shocks $\{\err_{t}\}$
are i.i.d.\ with $\expect\err_{t}=0$ and $\expect\err_{t}\err_{t}^{\trans}=I_{p}$,
and a (Lebesgue) density $\den\in\denspc$. The parameter space for
the model \eqref{nlVAR} then consists of collections $\fespc_{0}\ni\fe_{0}$
and $\bigfspc_{1}\ni\b{\fe}_{1}$ of functions $\reals^{p}\setmap\reals^{p}$
and $\reals^{kp}\setmap\reals^{p}$, and a collection of densities
$\denspc\ni\den$ supported on $\reals^{p}$, which under our regularity
conditions, together determine the conditional density
\begin{equation}
\varphi_{z_{t}\mid\b z_{t-1}}(\xi\mid\b{\xi}_{-1})=\den[\fe_{0}(\xi)-\b{\fe}_{1}(\b{\xi}_{-1})]\cdot\smlabs{\det D\fe_{0}(\xi)}.\label{eq:nlVARdens}
\end{equation}
We continue to regard two alternative parametrisations of the model
as being observationally equivalent if they yield the same conditional
density. For convenience, we shall suppose throughout that $\b z_{0}$
is continuously distributed, with a density that is a.e.\ strictly
positive on $\reals^{kp}$. Our assumptions below then ensure that
this is also true for every successive $\b z_{t}$, and $\varphi_{z_{t}\mid\b z_{t-1}}(\xi\mid\b{\xi}_{-1})$
is thus well defined for almost every $\xi\in\reals^{p}$ and $\b{\xi}_{-1}\in\reals^{kp}$,
for all $t\geq1$.

Our regularity conditions on the model parameter space, which are
sufficient to ensure that the conditional density \eqref{nlVARdens}
exists (and is unique up to the usual a.e.\ equivalence), are as
follows.

\needspace{4\baselineskip}

\assumpname{PS}
\begin{assumption}
\label{ass:svarpar}$\fespc_{0}$, $\bigfspc_{1}$ and $\denspc$
collect every function such that:
\begin{enumerate}[label=\ass{\arabic*.}, ref=\ass{\arabic*}]
\item \label{enu:shockid:smooth}$\tilde{f}_{0}\in\fespc_{0}$ and $\tilde{\b{\fe}}_{1}\in\bigfspc_{1}$
are locally Lipschitz (continuous);
\item \label{enu:shockid:bij}$\tilde{\fe}_{0}\in\fespc_{0}$ is a bijection
$\reals^{p}\setmap\reals^{p}$, $\tilde{\fe}_{0}(0)=0$, and $\det D\tilde{\fe}_{0}(z)\neq0$
for almost every $z\in\reals^{p}$;
\item \label{enu:shockid:dens}$\tilde{\den}\in\denspc$ is continuously
differentiable, with $\tilde{\den}(\err)>0$ for all $\err\in\reals^{p}$,
and
\begin{align*}
\int_{\reals^{p}}\tilde{\den}(\err)\diff\err & =1, & \int_{\reals^{p}}\err\tilde{\den}(\err)\diff\err & =0, & \int_{\reals^{p}}\err\err^{\trans}\tilde{\den}(\err)\diff\err & =I_{p}.
\end{align*}
\end{enumerate}
\end{assumption}
\begin{rem}
\subremark{} Local Lipschitzness implies that $\tilde{\fe}_{0}$
and $\tilde{\b{\fe}}_{1}$ are differentiable almost everywhere (a.e.).
The r.h.s.\ of \eqref{nlVARdens} is therefore defined at least almost
everywhere, which is sufficient to pin down the conditional density
$\varphi_{z_{t}\mid\b z_{t-1}}$, since the latter is itself only
uniquely defined up to an a.e.\ equivalence. (See \appref{nlsem}
for further details.) Our smoothness conditions and support conditions
on the density $\tilde{\den}$ (which accord with those of \citealp{Matz09Ecta})
are maintained only for convenience, and could very likely also be
relaxed in this same direction.

\subremark{} Since the nonlinear SVAR \eqref{nlVAR} is a (dynamic)
nonlinear SEM, our work relates closely to the literature on identification
in such models: particularly \citet{Matz09Ecta,Matz15Ecta} and \citet{BH18Ecta}.
Here we have deliberately relaxed the assumption that the functions
$\tilde{\fe}_{0}$ and $\tilde{\b{\fe}}_{1}$ are (at least once)
continuously differentiable, which is standard in that literature,
to allow our results to accommodate models that are continuous but
merely piecewise differentiable, such as the piecewise affine SVARs
introduced in \secref{piecewise-affine} below.

\subremark{} We naturally require $\tilde{\fe}_{0}$ to be invertible,
which ensures that the model always yields a solution for the endogenous
variables $z_{t}$, irrespective of the values of the predetermined
variables $\b z_{t-1}$ and the structural shocks $\err_{t}$. Requiring
$\det D\tilde{\fe}_{0}(z)\neq0$ a.e.\ merely excludes certain `irregular'
cases (our assumptions also imply that this quantity must have the
same sign a.e.).
\end{rem}

Regarding the parameters $(\fe_{0},\bigf_{1},\den)$ that generated
$\{z_{t}\}$ in \eqref{nlVAR}, as distinct from the \emph{entirety}
of the model parameter space, we also maintain the following.

\assumpname{DGP}
\begin{assumption}
\label{ass:svardgp}$(\fe_{0},\bigf_{1},\den)$ are such that:
\begin{enumerate}[label=\ass{\arabic*.}, ref=\ass{\arabic*}]
\item \label{enu:shockid:surj}$\bigf_{1}:\reals^{kp}\setmap\reals^{p}$
is surjective, with $\rank D\b{\fe}_{1}(\b z)=p$ for almost every
$\b z\in\reals^{kp}$;
\item \label{enu:invlip}$\fe_{0}^{-1}$ is locally Lipschitz; and
\item \label{enu:shockid:densrank}$\den$ has a \emph{local} maximum at
some $\err^{\ast}\in\reals^{p}$, and is twice continuously differentiable
in a neighbourhood of $\err^{\ast}$, with negative definite Hessian
there.
\end{enumerate}
\end{assumption}
\begin{rem}
\subremark{} We interpret \assref{svardgp}\ass{.}\ref{enu:shockid:surj}
as requiring that there be sufficient dependence of the r.h.s.\ of
the model (i.e.\ on the conditional mean of $\fe_{0}(z_{t})$) on
the predetermined variables $\b z_{t-1}$, in both a `global' and
`local' sense. (Note that this is only a requirement on the $\b{\fe}_{1}$
that actually generated the data, which need \emph{not} be satisfied
by all members of $\bigfspc_{1}$). For a simple illustration of why
some such condition cannot be avoided, consider an extreme case in
which $\b{\fe}_{1}(\b z)=0$ for all $\b z\in\reals^{kp}$, so that
the r.h.s.\ of \eqref{nlVAR} does not depend on $\b z_{t-1}$ at
all. Then because $z_{t}=\fe_{0}^{-1}(\err_{t})$ will be i.i.d.\ and
independent of $\b z_{t-1}$, so too will be 
\[
\tilde{\err}_{t}\defeq\tilde{\fe}_{0}(z_{t})=\tilde{\fe}_{0}[\fe_{0}^{-1}(\err_{t})]
\]
for \emph{every} $\tilde{\fe}_{0}\in\fespc_{0}$. Beyond requiring
$\tilde{\fe}_{0}$ to be scale- and location-normalised such that
$\expect\tilde{\err}_{t}=0$ and $\expect\tilde{\err}_{t}\tilde{\err}_{t}^{\trans}=I_{p}$,
the model would therefore yield no meaningful identifying restrictions
on $\tilde{\fe}_{0}$.

\subremark{} \assref{svardgp}\ass{.}\ref{enu:invlip} is a weak
regularity condition on the inverse of $\fe_{0}$, which would e.g.\ be
automatically satisfied if $\fe_{0}$ were continuously differentiable
with $\det D\fe_{0}(z)\neq0$ for all $z\in\reals^{p}$.

\subremark{} \assref{svardgp}\ass{.}\ref{enu:shockid:densrank}
would clearly be satisfied if $\err_{t}$ were Gaussian; but note
that only a well-behaved \emph{local} maximum is required for this
condition to hold. The main purpose of this assumption is to allow
us to deduce that $u=u^{\ast}$ from merely the equality $\fe_{U}(u)=f_{U}(u^{\ast})$,
and further regulate the behaviour of $\fe_{U}$ in the vicinity of
$u^{\ast}$. Though their model and proofs differ significantly from
ours -- in particular, because their counterpart of our $\b{\fe}_{1}$
has the property that each component depends on a variable (an `instrument')
that is special to that component -- it is noteworthy that a similar
assumption is introduced by \citet{BH18Ecta} as their Condition~M
(see also their Corollary~2).
\end{rem}

Remarkably, despite the far greater flexibility afforded by the nonlinear
parametrisation of \eqref{nlVAR}, under the foregoing conditions
we obtain the following, effectively identical characterisation of
observational equivalence to that of the linear SVAR \eqref{linSVAR},
the proof of which appears in \appref{nlsem}.
\begin{thm}
\label{thm:shockid}Suppose \ref{ass:svarpar} and \ref{ass:svardgp}
hold, and let $\tilde{\fe}_{0}\in\fespc_{0}$ and $\tilde{\b{\fe}_{1}}\in\bigfspc_{1}$.
Then there exists a $\tilde{\den}\in\denspc$ such that $(\tilde{\fe}_{0},\tilde{\bigf}_{1},\tilde{\den})$
is observationally equivalent to $(\fe_{0},\bigf_{1},\den)$, if and
only if there exists a $Q\in\orths(p)$ such that
\begin{align}
\tilde{\fe}_{0}(z) & =Q\fe_{0}(z)\sep\forall z\in\reals^{p}, & \tilde{\bigf}_{1}(\z) & =Q\bigf_{1}(\z)\sep\forall\z\in\reals^{kp}.\label{eq:orthogtr}
\end{align}
\end{thm}

\begin{rem}
\subremark{} Here we are asking whether for given candidate functions
$(\tilde{\fe}_{0},\tilde{\b{\fe}}_{1})$, it is possible to find a
distribution $\tilde{\den}\in\denspc$ for the structural shocks such
that
\[
\den[\fe_{0}(\xi)-\b{\fe}_{1}(\b{\xi}_{-1})]\cdot\smlabs{\det D\fe_{0}(\xi)}=\tilde{\den}[\tilde{\fe}_{0}(\xi)-\tilde{\b{\fe}}_{1}(\b{\xi}_{-1})]\cdot\smlabs{\det D\tilde{\fe}_{0}(\xi)}
\]
holds for almost every $\xi\in\reals^{p}$ and $\b{\xi}_{-1}\in\reals^{kp}$.
The $\tilde{\den}$ delivering this equivalence will, for $Q$ as
in \eqref{orthogtr}, be given by the density of 
\[
\tilde{\err}_{t}=\tilde{\fe}_{0}(z_{t})-\tilde{\b{\fe}}_{1}(\z_{t-1})=Q[\fe_{0}(z_{t})-\b{\fe}_{1}(\z_{t-1})]=Q\err_{t},
\]
which under our assumptions will also lie in $\denspc$. This implies
that the introduction of further (e.g.\ parametric) assumptions on
the set $\denspc$ of allowable densities would not yield any further
tightening of our characterisation of observational equivalence, provided
that $\denspc$ remains closed under orthogonal transformations of
the variables: as would e.g.\ be the case even if $\denspc$ were
restricted to the set of Gaussian densities on $\reals^{p}$ (with
mean zero and identity covariance).

\subremark{} The foregoing is a nonparametric identification result,
in the sense that neither $(\fe_{0},\b{\fe}_{1})$, nor the distribution
$\den$ of the shocks, are assumed to have any particular (known)
parametric form. In practice, however, we would expect the model \eqref{nlVAR}
to be formulated parametrically, if only because the limited length
of the time series available, for most macroeconomic applications,
render genuine nonparametric estimation infeasible. In the abstract
setting of \thmref{shockid}, these parametric functional form and/or
distributional assumptions can be understood as restrictions on the
sets $\fespc_{0}$, $\bigfspc_{1}$ and $\denspc$. The conclusion
of the theorem continues to hold in such cases, provided that $\denspc$
is not so (unusually) constrained that it fails to satisfy the invariance
condition noted in the previous remark. See \secref{piecewise-affine}
below for the discussion of a class of parametric models (for $\fe_{0}$
and $\b{\fe}_{1}$) for which the conditions required by the theorem
may be verified straightforwardly.

\subremark{} As noted above, a consequence of the Markov property
of the SVAR is that the notion of observational equivalence appropriate
to our setting refers only to the distribution $z_{t}\mid\b z_{t-1}$
of the endogenous variables conditional on the exogenous variables;
it therefore coincides exactly with that employed by \citet{Matz09Ecta}
in the context of a (non-dynamic) nonlinear SEM: see her (3.1), in
particular. This allows the proof of \thmref{shockid} to be approached
just as if we were analysing identification in a nonlinear SEM, a
connection that we draw out more fully in \appref{nlsem}. Relative
to the results in the existing SEM literature, we obtain a much tighter
characterisation of observational equivalence because of the separability
between $z_{t}$ and $\b z_{t-1}$.

\subremark{} Should \eqref{orthogtr} fail to hold, then there will
be at least some realisations of $\{z_{t}\}$ for which the likelihoods
of $(\tilde{\fe}_{0},\tilde{\bigf}_{1},\tilde{\den})$ and $(\fe_{0},\bigf_{1},\den)$
will be distinct, and so the data will to this extent be informative
about these two alternative parametrisations of the model. However,
we would not\emph{ }claim, on the basis of this theorem \emph{alone},
that the parameters of the model are consistently estimable up to
an orthogonal transformation. While it seems reasonable to suppose
that consistent nonparametric estimation of the model would be possible
(under suitable regularity conditions) when $\{z_{t}\}$ is stationary
and ergodic, the familiar connection between identification and consistent
estimation is attenuated when $\{z_{t}\}$ possesses stochastic (or
indeed, deterministic) trends, because of the non-recurrence of those
trends in higher dimensions (see \citealp[Sec.~6]{Bing01Hdbk}; \citealp[p.~62]{GP13JoE}).
Consistent estimation of the model parameters (up to $Q$) would in
such cases likely require further restrictions on $(\fe_{0},\b{\fe}_{1})$,
such as those sufficient to ensure that $\{z_{t}\}$ is indeed stationary
and ergodic (for a discussion of such conditions in this context,
see \citealp{DMW23stat}, and the references cited therein). 

\end{rem}

\subsection{Orthogonal reduced-form parametrisation}

\label{subsec:orthogreducedform}

Analogously (though not identically) to the `orthogonal reduced-form
parametrisation' (\citealp{ARRW18Ecta}, Sec.\ 2.3) of the linear
SVAR, \thmref{shockid} suggests the following convenient reparametrisation
of the endogenously nonlinear SVAR. Let $z_{0}\in\reals^{p}$ be fixed,
and a point at which $\fe_{0}$ is (assumed to be) differentiable,
with full rank Jacobian $D\fe_{0}(z_{0})$. By the QR decomposition,
we have $D\fe_{0}(z_{0})=Q^{\trans}L$, where $L$ is lower triangular,
and $Q\in\orths(p)$; define $(\ga_{0},\b{\ga}_{1})\defeq(Q\fe_{0},Q\b{\fe}_{1})$.
Multiplying \eqref{nlVAR} through by $Q$, we may reformulate the
model as
\begin{equation}
\ga_{0}(z_{t})=\b{\ga}_{1}(z_{t-1})+Q\err_{t}\label{eq:svar-equiv}
\end{equation}
where now $\ga_{0}$ is restricted such that $D\ga_{0}(z_{0})$ is
lower triangular (which need hold \emph{only} at that chosen $z_{0}$),
and $Q\in\orths(p)$. 

This yields an equivalent parametrisation of the model, in which the
parameter spaces for $\b{\ga}_{1}\in\b{\fespc}_{1}$ and $\den\in\denspc$
remain as before, but now $\fespc_{0}$ is additionally restricted
(beyond \assref{svarpar}\ass{.}\ref{enu:shockid:smooth}\ass{--}\ref{enu:shockid:bij})
to functions $\ga_{0}:\reals^{p}\setmap\reals^{p}$ for which $D\ga_{0}(z_{0})$
is lower triangular (at the nominated $z_{0}\in\reals^{p}$); let
$\fespc_{0}^{(z_{0})}$ denote the resulting parameter space for $\ga_{0}$.
To exactly offset this restriction, we now have the additional parameter
$Q\in\orths(p)$, so that we may equivalently regard the nonlinear
SVAR as being parametrised by $(\ga_{0},\b{\ga}_{1},Q,\den)\in\fespc_{0}^{(z_{0})}\times\bigfspc_{1}\times\orths(p)\times\denspc$.
The import of \thmref{shockid} here is that the parameters $(\ga_{0},\b{\ga}_{1})\in\fespc_{0}^{(z_{0})}\times\bigfspc_{1}$
are \emph{exactly} identified by data on $\{z_{t}\}$, with the non-identified
part of the structural parameters being transferred entirely to $Q$.
The `nonlinear SVAR identification problem' can thus be framed precisely
as one of finding sufficient restrictions to pin down $Q$, from which
the structural parameters may then be recovered, via $(\fe_{0},\b{\fe}_{1})=(Q^{\trans}\ga_{0},Q^{\trans}\b{\ga}_{1})$.

The reparametrisation \eqref{svar-equiv} provides a convenient perspective
from which to import various approaches to identifying $Q$ from the
linear SVAR literature. For the most part, these apply directly to
the present setting, with little modification required. The following
example illustrates how it remains possible to identify impulse responses
via external instruments, without requiring any additional assumptions
relative to those needed to identify the linear VAR. 
\begin{example}[external instruments]
 Suppose that $w_{t}$ is a (scalar) `external instrument': an
observed (stationary) process that is assumed to be contemporaneously
correlated with the first structural shock, but not with any of the
others (see e.g.\ \citealp{SW18EJ}, p.\ 931). Defining
\begin{equation}
u_{t}\defeq\ga_{0}(z_{t})-\b{\ga}_{1}(z_{t-1})=Q\err_{t}\label{eq:redformerr}
\end{equation}
which by \thmref{shockid} is identified, we must have 
\[
\delta\defeq\expect u_{t}w_{t}=Q\expect\err_{t}w_{t}=Qe_{1}\alpha=q_{1}\alpha,
\]
where $q_{1}$ denotes the first column of $Q$, and $\alpha=\expect\err_{1t}w_{t}\neq0$.
Since $\delta=\expect u_{t}w_{t}$ is identified, so too is $q_{1}=\delta/\smlnorm{\delta}$,
and we can further recover $\err_{1t}=q_{1}^{\trans}u_{t}$.

Since the distribution of $u_{t}$ in \eqref{redformerr} is identified,
so too is the conditional distribution
\[
u_{t}\mid\{\err_{1t}=\bar{\err}_{1}\}\eqdist u_{t}\mid\{q_{1}^{\trans}u_{t}=\bar{\err}_{1}\}.
\]
For given values of $\b z_{t-1}=\bar{\b z}$ and $\bar{\err}_{1}$,
the distribution of the counterfactual quantity
\[
z_{t}(\bar{\b z},\bar{\err}_{1})\mid\{\b z_{t-1}=\bar{\b z},\err_{1t}=\bar{\err}_{1}\}\eqdist\ga_{0}^{-1}(\b{\ga}_{1}(\bar{\b z})+u_{t})\mid\{q_{1}^{\trans}u_{t}=\bar{\err}_{1}\}
\]
depends only on $\ga_{0}$, $\b{\ga}_{1}$ and the distribution of
$u_{t}\mid\{q_{1}^{\trans}u_{t}=\bar{\err}_{1}\}$, all of which are
identified. In this way, the impact multipliers of $\err_{1t}$ may
be recovered; the impulse responses at further horizons depend, by
the Markov property, additionally only on the conditional distribution
of $z_{t}\mid\b z_{t-1}$, which is trivially identified.
\end{example}

\section{Piecewise affine SVARs}

\label{sec:piecewise-affine}

\subsection{Endogenous regime switching}

Here we introduce a class of endogenously regime-switching models,
in which the conditions required for our results may be verified relatively
straightforwardly. Models of this form have been used recently to
study monetary policy under an occasionally binding constraint on
nominal interest rates: see \citet{SM21}, \citet{AMSV21}, \citet{ILMZ20},
and \citet{CCMM25}.

Suppose now that the l.h.s.\ of the nonlinear SVAR
\begin{equation}
\fe_{0}(z_{t})=\b{\fe}_{1}(\z_{t-1})+\err_{t}\label{eq:pwa0}
\end{equation}
is specified as
\begin{equation}
\fe_{0}(z)=\sum_{\ell=1}^{L}\indic\{z\in\set Z^{(\ell)}\}(\bar{\phi}_{0}^{(\ell)}+\Phi_{0}^{(\ell)}z),\label{eq:pwa}
\end{equation}
for $\{\set Z^{(\ell)}\}_{\ell=1}^{L}$ a collection of convex sets
that partition $\reals^{p}$, $\{\bar{\phi}_{0}^{(\ell)}\}_{\ell=1}^{L}\subset\reals^{p}$
and $\{\Phi_{0}^{(\ell)}\}_{\ell=1}^{L}\subset\reals^{p\times p}$.
When these parameters are such that $\fe_{0}$ is continuous, we shall
say that $\fe_{0}$ is a \emph{piecewise affine }function. (We do
not consider cases in which $\fe_{0}$ may be discontinuous, so continuity
should always be taken as implied.) The model may then be regarded
as consisting of $L$ `regimes' demarcated by the sets $\{\set Z^{(\ell)}\}_{\ell=1}^{L}$.
Which of those regimes is operative in period $t$, i.e.\ the value
of $\ell_{t}\in\{1,\ldots,L\}$ such that
\[
\fe_{0}(z_{t})=\bar{\phi}_{0}^{(\ell_{t})}+\Phi_{0}^{(\ell_{t})}z_{t}
\]
is determined \emph{jointly} with the value of $z_{t}$. For this
reason, we say that there is \emph{endogenous} switching between the
$L$ regimes, as distinct from the \emph{exogenous} regime switching
that would result if $\ell_{t}$ were determined prior to the realisation
of $z_{t}$. The situation here is thus markedly different from the
regime-switching SVARs considered in the previous literature, which
as noted in the introduction, can generally be written in the form
\[
\Phi_{0}(s_{t-1})z_{t}=c(s_{t-1})+\sum_{i=1}^{k}\Phi_{i}(s_{t-1})z_{t-i}+\err_{t},
\]
where $s_{t-1}$ is determined prior to $\err_{t}$ and $z_{t}$ (see
e.g.\ \citealp{AG12AEJ}; \citealp{CCCN15EJ}; \citealp{BP24EctJ}).

\saveexamplex{}

\exname{\ref*{exa:phillips}}
\begin{example}[nonlinear Phillips curve; ctd]
 The nonlinear Phillips curve of \citet{BenignoEggertsson2023},
in \eqref{NLPC} above, is piecewise affine (and continuous) with
a kink at $\log\theta_{t}=0$, which the authors refer to as the `Beveridge
threshold'. Their model thus delineates two distinct labour market
regimes: a `normal' regime ($\ell_{t}=1$), when the labour market
is slack, $\log\theta_{t}\leq0$, and a `labour shortage' regime ($\ell_{t}=2$)
in which $\log\theta_{t}>0$. The regime $\ell_{t}$ is entirely driven
by the \emph{contemporaneous} value of the endogenous variable $\log\theta_{t}$,
and so the regime-switching is genuinely endogenous. Their Phillips
curve \eqref{NLPC} can also be written as
\begin{equation}
\pi_{t}=\pi+\kappa^{(\ell_{t})}\log\theta_{t}+\eta_{t}.\label{eq:nlPC=000020compact}
\end{equation}
where $\kappa^{(1)}=\kappa$ and $\kappa^{(2)}=\kappa^{\mathrm{tight}}$.
Contrast this with an alternative specification in which the slope
of the Phillips curve is determined by \emph{past} values of $\log\theta_{t}$,
for example
\begin{equation}
\pi_{t}=\pi+\kappa^{(\ell{}_{t-1})}\log\theta_{t}+\eta_{t}.\label{eq:predet=000020PC}
\end{equation}
Conditional on the past (i.e.\ on time $t-1$), \eqref{predet=000020PC}
is linear in $z_{t}=(\log\theta_{t},\pi_{t})^{\trans}$, and so shocks
to $\log\theta_{t}$ will have the same proportional effect $\kappa^{(\ell{}_{t-1})}$
on $\pi_{t}$, irrespective of their sign; whereas in \eqref{nlPC=000020compact}
the impact of the shocks will vary additionally (and nonlinearly)
depending on the initial (i.e.\ pre-shock) proximity of tightness
to the Beveridge threshold.
\end{example}
\restoreexamplex{}

\subsection{Identification}

We would like primitive conditions that ensure $\fe_{0}$ in \eqref{pwa}
satisfies the requirements \assref{svarpar}\ass{.}\ref{enu:shockid:smooth}\ass{-}\ref{enu:shockid:bij}
and \ref{ass:svardgp}\ass{.}\ref{enu:invlip} of \thmref{shockid}:
namely, that both it and its inverse should be locally Lipschitz,
and that it should be (globally) invertible, with $\det D\fe_{0}(z)\neq0$
a.e. Two important special cases of \eqref{pwa}, for which these
conditions may be readily verified, are those of:
\begin{itemize}
\item a (continuous) \emph{piecewise linear }function, in which there exists
a basis $\{a_{i}\}_{i=1}^{p}$ for $\reals^{p}$ such that each $\set Z^{(\ell)}$
can be written as a union of cones of the form
\begin{equation}
\set C_{{\cal I}}\defeq\{z\in\reals^{p}\mid a_{i}^{\trans}z\geq0\sep\forall i\in{\cal I}\text{ and }a_{i}^{\trans}z<0\sep\forall i\notin{\cal I}\}\label{eq:pwlbasis}
\end{equation}
where ${\cal I}$ ranges over the subsets of $\{1,\ldots,p\}$, and
$\bar{\phi}_{0}^{(\ell)}=0$ for all $\ell\in\{1,\ldots,L\}$; and
\item a (continuous) \emph{threshold affine }function, in which there exists
an $a\in\reals^{p}\backslash\{0\}$ and thresholds $\{\tau_{\ell}\}_{\ell=0}^{L}$
with $\tau_{\ell}<\tau_{\ell+1}$, $\tau_{0}=-\infty$ and $\tau_{L}=+\infty$,
such that
\[
\set Z^{(\ell)}=\{z\in\reals^{p}\mid a^{\trans}z\in(\tau_{\ell-1},\tau_{\ell}]\},
\]
i.e.\ the sets $\{\set Z^{(\ell)}\}$ take the forms of `bands'
in $\reals^{p}$. (In typical examples, $a=e_{p,i}$, i.e.\ it picks
out one `threshold variable' from the elements of $z_{t}$.)
\end{itemize}

Because the boundaries between the regimes are then affine subspaces
(of $\reals^{p}$), ensuring the continuity of $\fe_{0}$ is a straightforward
matter of linearly restricting the elements of $\{\bar{\phi}_{0}^{(\ell)}\}_{\ell=1}^{L}$
and $\{\Phi_{0}^{(\ell)}\}_{\ell=1}^{L}$ such that the values prescribed
by adjacent regimes agree on those boundaries; see the next appearance
of \exaref{phillips} for an illustration. Regarding our remaining
requirements on $\fe_{0}$, for these it is necessary and sufficient
that
\begin{equation}
\sgn\det\Phi_{0}^{(\ell)}=\sgn\det\Phi_{0}^{(1)}\neq0\sep\forall\ell\in\{1,\ldots,L\}.\label{eq:detcond}
\end{equation}
See \propref{pwacondition} below; we note that equivalence of the
preceding with the invertibility of $\fe_{0}$ follows directly from
Theorems~1 and 4 of \citet{GLM80Ecta}, and that since $Df_{0}(z)=\sum_{\ell=1}^{L}\indic\{z\in\set Z^{(\ell)}\}\Phi_{0}^{(\ell)}$
a.e., the Jacobian is then clearly invertible a.e.

\saveexamplex{}

\exname{\ref*{exa:phillips}}
\begin{example}[nonlinear Phillips curve; ctd]
 The nonlinear Phillips curve \eqref{NLPC} is piecewise linear,
with $z_{t}=(\log\theta_{t},\pi_{t})^{\trans}$ and two regimes 
\begin{align*}
\mathscr{Z}^{(1)} & =\{z\in\reals^{2}\mid e_{1}^{\trans}z\leq0\} & \mathscr{Z}^{(2)} & =\{z\in\reals^{2}\mid e_{1}^{\trans}z>0\}
\end{align*}
which can each be written as unions of cones of the form \eqref{pwlbasis}
(e.g.\ by taking $a_{1}=-e_{1}$ and $a_{2}=e_{2}$). \eqref{NLPC}
specifies only the first component of the bivariate map $\fe_{0}(z_{t})$.
If the second component is also modelled as piecewise linear, with
regimes also determined by the sign of $\log\theta_{t}$ (thus linear
on each of the sets $\mathscr{Z}^{(1)}$ and $\mathscr{Z}^{(2)}$),
then $\fe_{0}$ admits the representation \eqref{pwa}. To ensure
continuity at the regime boundary where $\log\theta_{t}=0$, we need
the equality
\[
\bar{\phi}_{0}^{(1)}+\begin{bmatrix}\Phi_{0,1}^{(1)} & \Phi_{0,2}^{(1)}\end{bmatrix}\begin{bmatrix}0\\
\pi_{t}
\end{bmatrix}=\bar{\phi}_{0}^{(2)}+\begin{bmatrix}\Phi_{0,1}^{(2)} & \Phi_{0,2}^{(2)}\end{bmatrix}\begin{bmatrix}0\\
\pi_{t}
\end{bmatrix}
\]
to hold for all values of $\pi_{t}\in\reals$, where $\Phi_{0}^{(\ell)}=[\Phi_{0,1}^{(\ell)},\Phi_{0,2}^{(\ell)}]$.
This entails
\begin{align*}
\bar{\phi}_{0}^{(1)}-\bar{\phi}_{0}^{(2)} & =0 & \Phi_{0,2}^{(1)}-\Phi_{0,2}^{(2)} & =0,
\end{align*}
and we may also impose $\bar{\phi}_{0}^{(1)}=0$, for the location
normalisation $\fe_{0}(0)=0$. To put it another way, continuity requires
that only the coefficients on the regime-determining variable $\log\theta_{t}$
may change at the threshold, leading to the (non-redundant) specification
\begin{equation}
\Phi_{0}^{(\ell)}=[\Phi_{0,1}^{(\ell)},\Phi_{0,2}]\sep\ell\in\{1,2\}\label{eq:pc-cont}
\end{equation}
 in which the second column of the coefficient matrix is regime-invariant.
\end{example}
\restoreexamplex{}

The SVAR specification \eqref{pwa0}--\eqref{pwa} thus provides
a flexible but tractable means of introducing nonlinearity into an
SVAR model. This is especially the case if we also specify that the
r.h.s.\ should be additively time-separable, and of the same functional
form as the l.h.s., so that
\begin{equation}
\fe_{0}(z_{t})=\b{\fe}_{1}(\z_{t-1})+\err_{t}=c+\sum_{i=1}^{k}\fe_{i}(z_{t-i})+\err_{t}\label{eq:pwaVAR}
\end{equation}
where now, for \emph{every} $i\in\{0,\ldots,k\}$,
\begin{equation}
\fe_{i}(z)=\sum_{\ell=1}^{L}\indic\{z\in\set Z^{(\ell)}\}(\bar{\phi}_{i}^{(\ell)}+\Phi_{i}^{(\ell)}z),\label{eq:pwa-fi}
\end{equation}
is (continuous) piecewise affine. (Note that there is no need to additionally
index the regimes $\set Z^{(\ell)}$ by $i$ here, since if the partitions
$\{\set Z_{i}^{(\ell)}\}_{\ell=1}^{L_{i}}$ did vary across $i$,
we could always find a mutual refinement such that \eqref{pwa-fi}
held for all $i$.) We term this model a \emph{piecewise affine SVAR};
with \emph{piecewise linear} and \emph{threshold affine} \emph{SVARs}
corresponding to those cases where the $\fe_{i}$'s are either all
piecewise linear or all threshold affine functions, respectively. 

The conditions \assref{svarpar}\ass{.}\ref{enu:shockid:smooth}
and \assref{svardgp}\ass{.}\ref{enu:shockid:surj} imposed by \thmref{shockid}
on $\b{\fe}_{1}(\b z_{t-1})=\sum_{i=1}^{k}\fe_{i}(z_{t-i})$ are rather
less taxing than those imposed on $\fe_{0}$. Under the specification
\eqref{pwa-fi}, continuity is readily imposed, and then automatically
implies Lipschitz continuity. Moreover, $D\b f_{1}(\b z_{t-1})$ a.e.\ exists
and satisfies
\[
D\b{\fe}_{1}(\b z)=\begin{bmatrix}\Phi_{1}^{(\ell_{1})} & \Phi_{2}^{(\ell_{2})} & \cdots & \Phi_{k}^{(\ell_{k})}\end{bmatrix}
\]
for some $\ell_{i}\in\{1,\ldots,L\}$ depending on $\b z$, and so
it is easy to verify whether $\rank D\b{\fe}_{1}(\b z)=p$ a.e.\ (or,
since this holds generically, to test the null hypothesis of a deficient
rank). In practice, this may be analysed more straightforwardly on
the basis of the coefficients on the first lag or two of $z_{t}$,
which may themselves be sufficient to satisfy this condition.

Verifying the (global) surjectivity condition on $\b{\fe}_{1}$ is
a little more challenging, because of the apparent absence of a counterpart
to \eqref{detcond} for this case. In the special case of a model
with only one lag, surjectivity of $z_{t-1}\elmap\fe_{1}(z_{t-1})$
is equivalent to \eqref{detcond}. Though easy to check, this is far
more than is necessary for surjectivity when additional lags are present.
Alternatively, if some elements of $z_{t}$ enter $\fe_{i}$ linearly,
as will often be the case in practice (as in our next example), then
surjectivity holds so long as the coefficient vectors associated with
(at least) $p$ of these variables (drawn from across the $k$ lags
of $z_{t}$ appearing on the r.h.s.)\ form a rank $p$ matrix.

\begin{example}[occasionally binding constraint]
 \citet{SM21} proposed the censored and kinked structural VAR (CKSVAR),
to model the effects of the zero lower bound (ZLB) constraint on monetary
policy: see also \citet{AMSV21} and \citet{CCMM25}. In his setting,
$y_{t}$ is a scalar variable whose positive part $y_{t}^{+}\defeq\max\{y_{t},0\}$
coincides with the central bank's policy rate (constrained to be
non-negative), while its (latent) negative part $y_{t}^{-}\defeq\min\{y_{t},0\}$
is the `shadow rate', which summarises the stance of monetary policy
desired by the central bank when the ZLB binds, to be engineered via
`unconventional' policy, such as asset purchases. The remaining
variables in the model are collected in the $(p-1)$-dimensional vector
$x_{t}$, in his case the inflation and unemployment rates; we then
set $z_{t}=(y_{t},x_{t}^{\trans})^{\trans}$.

To allow for possibility that the ZLB might actually constrain monetary
policy, $y_{t}^{+}$ and $y_{t}^{-}$ are permitted to enter the model
with different coefficients (in possibly all $p$ equations), 
\begin{equation}
\phi_{0}^{+}y_{t}^{+}+\phi_{0}^{-}y_{t}^{-}+\Phi_{0}^{x}x_{t}=c+\sum_{i=1}^{k}[\phi_{i}^{+}y_{t-i}^{+}+\phi_{i}^{-}y_{t-i}^{-}+\Phi_{i}^{x}x_{t-i}]+u_{t}\label{eq:cksvar}
\end{equation}
where $\phi_{i}^{\pm}\in\reals^{p}$ and $\Phi_{i}^{x}\in\reals^{p\times(p-1)}$,
for $i\in\{0,\ldots,k\}$. This may be rendered as an instance of
a threshold affine SVAR by defining\begin{subequations}\label{eq:CKSVAR-sets}
\begin{align}
\set Z^{(1)} & \defeq\set Z^{-}=\{(y,x)\in\reals^{p}\mid y\leq\tau_{1}\} & \Phi_{i}^{(1)} & \defeq[\phi_{i}^{-},\Phi_{i}^{x}]\\
\set Z^{(2)} & \defeq\set Z^{+}=\{(y,x)\in\reals^{p}\mid y>\tau_{1}\} & \Phi_{i}^{(2)} & \defeq[\phi_{i}^{+},\Phi_{i}^{x}],
\end{align}
\end{subequations} with $\tau_{1}=0$, and then setting $\fe_{i}(z)=\sum_{\ell=1}^{2}\indic\{z\in\set Z^{(\ell)}\}\Phi_{i}^{(\ell)}z$.
(Because there are only two regimes, it may also be equivalently cast
as a piecewise linear SVAR.) Here continuity of each $\fe_{i}$ is
guaranteed by the fact that $\Phi_{i}^{(1)}$ and $\Phi_{i}^{(2)}$
only differ by their first column; or equivalently by the linear restrictions
$(\Phi_{i}^{(1)}-\Phi_{i}^{(2)})E_{-1}=0$, for $E_{-1}$ the final
$p-1$ columns of $I_{p}$.

In \citet{SM21}, identification of the parameters of this model are
complicated by the fact that $y_{t}$ is only observed when $y_{t}>0$;
it is in effect censored at zero. His results therefore do not fall
within the framework of \thmref{shockid}, which implicitly assumes
that $z_{t}$ and $\b z_{t-1}$ are observed on the entirety of their
supports. However, the model \eqref{cksvar}--\eqref{CKSVAR-sets}
may (of course) also be applied to settings in which $y_{t}$ is observed
on both sides of the threshold $\tau_{1}$, which may be treated as
an additional unknown parameter to be identified and estimated. From
the foregoing discussion, for $\fe_{0}$ to satisfy the conditions
of \thmref{shockid}, we would need only to verify that $\det\Phi_{0}^{(1)}$
and $\det\Phi_{0}^{(2)}$ are both nonzero, and have the same sign.
Regarding $\b{\fe}_{1}$: if $k\geq2$ then it is sufficient to check
(or rather, test) whether the $p\times k(p-1)$ matrix $[\Phi_{1}^{x},\ldots,\Phi_{k}^{x}]$,
formed from the coefficients on the lags of $x_{t}$, has rank $p$;
whereas if $k=1$, then we would need $\{\Phi_{1}^{(\ell)}\}$ to
satisfy the same determinantal condition as $\{\Phi_{0}^{(\ell)}\}$
(a condition also sufficient when $k\geq2$).
\end{example}

\subsection{Smooth transitions}

The piecewise affine SVAR \eqref{pwaVAR}--\eqref{pwa-fi} may be
extended to allow for `smooth transitions' between the $L$ regimes.
In the literature on smooth transition (vector) autoregressive models,
the conventional approach (e.g.\ \citealp[Sec.~3.3]{HT13}) is to
replace the indicator functions $\indic\{z\in\set Z^{(\ell)}\}$ in
\eqref{pwa-fi} by smooth maps $\pi^{(\ell)}(z)$, so that now
\[
\fe_{i}^{\mathrm{ST}}(z)=\sum_{\ell=1}^{L}\pi^{(\ell)}(z)(\bar{\phi}_{i}^{(\ell)}+\Phi_{i}^{(\ell)}z),
\]
where $\pi^{(\ell)}(z)\in[0,1]$ and $\sum_{\ell=1}^{L}\pi^{(\ell)}(z)=1$
for all $z\in\reals^{p}$, so that $\fe_{i}^{\mathrm{ST}}(z)$ is
always a smooth, convex combination of the affine functions $z\elmap\bar{\phi}_{i}^{(\ell)}+\Phi_{i}^{(\ell)}z$,
for $\ell\in\{1,\ldots,L\}$. However, the fact that the \emph{gradient}
of $\fe_{i}^{\mathrm{ST}}$ is \emph{not} a convex combination of
those underlying affine regimes makes it difficult to reduce the high-level
conditions of \thmref{shockid} to a set of verifiable conditions
on the underlying regime-specific coefficient matrices, in the manner
of \eqref{detcond}. Indeed, as the simple example in \figref{Plot-of-functions}
illustrates, it may well be the case that $\fe_{0}^{\mathrm{ST}}$
is not invertible, even though its unsmoothed counterpart $\fe_{0}$
is.

As an alternative specification that allows for smooth transitions
between regimes, but which also retains the simplicity -- in terms
of verifying the conditions for \thmref{shockid} -- enjoyed by piecewise
affine models, consider
\begin{equation}
\fe_{i,K}(z)\defeq\int_{\reals^{p}}\fe_{i}(z+u)K(u)\diff u\label{eq:smoothed}
\end{equation}
where $\fe_{i}$ is a (continuous) piecewise affine function as in
\eqref{pwa-fi} above, and $K$ is a smooth (kernel) density function
with mean zero, with $m\geq1$ continuous derivatives that satisfy
the integrability condition
\begin{equation}
\int_{\reals^{p}}\smlnorm u\smlabs{\partial_{u_{\alpha_{1}}}\cdots\partial_{u_{\alpha_{n}}}K(u)}\diff u<\infty,\label{eq:integderiv}
\end{equation}
where $\partial_{u_{i}}$ denotes the partial derivative with respect
to the $i$th element of $u\in\reals^{p}$, for $\alpha_{i}\in\{1,\ldots,p\}$
and $1\leq n\leq m$.

\begin{figure}
\includegraphics[width=1\textwidth]{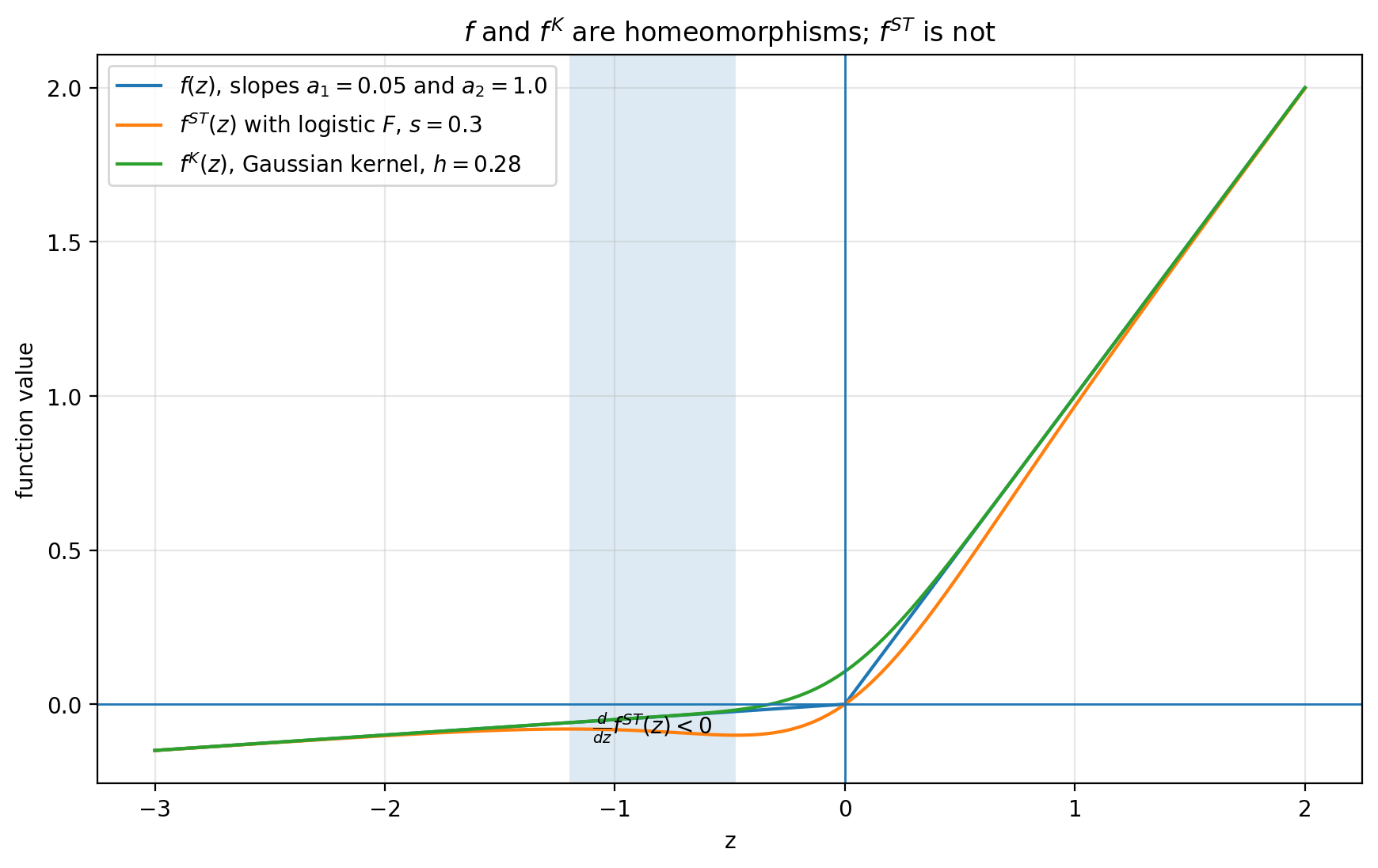}

{\footnotesize Plot of: $f(z)=a_{1}\min\{z,0\}+a_{2}\max\{z,0\}$;
$f^{\mathrm{ST}}\left(z\right)=[1-F(z)]a_{1}z+F(z)a_{2}z$, with $F(z)=(1+e^{z/s})^{-1}$;
and $f^{K}(z)=\int_{\reals}f(z+u)K(u)\diff u$, where $K(u)=h^{-1}\varphi(u/h)$
and $\varphi$ is the standard Gaussian pdf. }{\footnotesize\par}

\caption{Smooth transitions and invertibility}

\label{fig:Plot-of-functions}
\end{figure}

Our next result establishes that $\fe_{0,K}(z)$ is smooth (with as
many continuous derivatives as $K$ has), and moreover invertible
if the determinantal condition \eqref{detcond} is satisfied (its
proof appears in \appref{piecewiseaffine}). Recall that a function
is said to be bi-Lipschitz if both it and its inverse are Lipschitz
continuous.

\begin{prop}
\label{prop:pwacondition}Suppose that $\fe_{0}:\reals^{p}\setmap\reals^{p}$
is as in \eqref{pwa}, and is either a piecewise linear or threshold
affine function. Then:
\begin{enumerate}
\item \label{enu:orig}$\fe_{0}$ is invertible and bi-Lipschitz if and
only if \eqref{detcond} holds.
\end{enumerate}
Suppose that $K:\reals^{p}\setmap\reals$ is $m\geq1$ times continuously
differentiable and non-negative, satisfying $\int_{\reals^{p}}K(u)=1$,
$\int_{\reals^{p}}uK(u)\diff u=0$ and \eqref{integderiv}, and that
$\fe_{0,K}$ is formed by convolving $\fe_{0}$ with $K$, as in \eqref{smoothed}.
Then if \eqref{detcond} holds:
\begin{enumerate}[resume]
\item \label{enu:smooth}$\fe_{0,K}$ is invertible, bi-Lipschitz, and $m$
times continuously differentiable.
\end{enumerate}
\end{prop}

\section{Application: a nonlinear Phillips curve?}

\label{sec:phillipscurve}

\subsection{Formulation as an endogenous regime-switching SVAR}

The inflation surge that followed the COVID-19 pandemic reignited
academic interest in the possibility of nonlinearity in the transmission
of supply shocks to inflation. However, views on the relevance of
nonlinearity are divided. On the one hand, \citet{BallLeighMishra2022}
and \citet{BenignoEggertsson2023} find evidence of significant nonlinearity
in their formulations of the Phillips curve, and argue that nonlinearity
is needed to account for the recent inflation surge. On the other
hand, \citet{BeaudryHouPortier2025} caution that the evidence on
nonlinearity is not robust to functional form assumptions, especially
as pertains to the treatment of expectations. Reconsidering this debate,
through the lens of an endogenous regime-switching SVAR, provides
an illustrative application of the methodology developed in this paper.

Our identification result can be useful in this debate because it
highlights the following: since all observationally equivalent structures
are identified up to a (linear) orthogonal transformation, then if
one finds no (statistically significant) evidence of nonlinearity
under one specific identification scheme, then this will remain true
\emph{irrespective} of how the model is identified. Indeed, one can
see from the orthogonal reduced-form parametrisation developed in
\subsecref{orthogreducedform} above, that the structural parameters
$\fe_{0}$ (and $\b{\fe}_{1}$) will be nonlinear if and only if their
normalised (and exactly identified) counterparts $\ga_{0}$ (and $\b{\ga}_{1}$)
are also nonlinear, as will be the case for $Q\fe_{0}$ (and $Q\b{\fe}_{1}$)
for \emph{any} $Q\in\orths(p)$. Thus the presence of nonlinearity
can be tested for in a way that is robust to the identifying scheme
employed. To be clear, this is a consequence of modelling the joint
determination of $z_{t}=(\log\theta_{t},\pi_{t})^{\trans}$ in its
entirety; the argument does not carry over to the methodology employed
in the aforementioned papers, because these provide only a single-equation
analysis of the Phillips curve, and so their findings are potentially
contingent on the assumptions made in order to identify that equation.

Building on the development already given to this point in \exaref{phillips},
inspired by the recent work of \citet{BenignoEggertsson2023} we consider
the following endogenous regime-switching SVAR for $z_{t}=(\log\theta_{t},\pi_{t})^{\trans}$,
\begin{align}
\Phi_{0}^{(\ell_{t})}z_{t} & =c+\sum_{i=1}^{k}\Phi_{i}^{(\ell_{t-i})}z_{t-i}+\err_{t}, & \err_{t} & \distiid N[0,I_{2}]\label{eq:=000020model}
\end{align}
where $\theta_{t}=v_{t}/u_{t}$ is the vacancy--unemployment ratio,
$\pi_{t}$ is consumer price inflation, $\err_{t}$ are the structural
shocks, and
\begin{equation}
\ell_{t}:=\begin{cases}
1, & \text{if }z_{1t}\leq0\text{ (\textquoteleft normal\textquoteright)},\\
2, & \text{if }z_{1t}>0\text{ (\textquoteleft labour shortage\textquoteright),}
\end{cases}\label{aliknoint}
\end{equation}
where $z_{1t}=\log\theta_{t}$. This model thus has two regimes, determined
by the sign of $z_{1t}$. Following the arguments that led to \eqref{pc-cont}
above, to ensure continuity of the model in both $z_{t}$ and its
lags, we parametrise the regime-dependent coefficient matrices non-redundantly
as
\begin{equation}
\Phi_{i}^{(\ell)}=\begin{bmatrix}\Phi_{i,11}^{(\ell)} & \Phi_{i,12}^{(\ell)}\\
\Phi_{i,21}^{(\ell)} & \Phi_{i,22}^{(\ell)}
\end{bmatrix}=\begin{bmatrix}\Phi_{i,11}^{(\ell)} & \Phi_{i,12}\\
\Phi_{i,21}^{(\ell)} & \Phi_{i,22}
\end{bmatrix}\sep\ell\in\{1,2\}\label{eq:=000020continuity}
\end{equation}
so that only the coefficients of the regime-determining variable,
$z_{1t}$, are permitted to vary across the two regimes. The model
is then guaranteed to yield a solution for $z_{t}$, for every possible
value of the r.h.s.\ of \eqref{=000020model}, provided that $\det\Phi_{0}^{(1)}\cdot\det\Phi_{0}^{(2)}>0$.

To obtain a just-identified specification, by \thmref{shockid} it
suffices to impose $p(p-1)/2=1$ restrictions on the model parameters
(see also the discussion in \subsecref{orthogreducedform} above).
For some identifying schemes, this may involve imposing a restriction
on only one of the two regimes. However, the identifying assumption
in \citet{BenignoEggertsson2023} corresponds to the `recursive' or
`Cholesky' restriction under which (a shock to) inflation $\pi_{t}=z_{2t}$
has no contemporaneous effect on tightness $\log\theta_{t}=z_{1t}$,
and thus that the matrix $\Phi_{0}^{(\ell)}$ is lower triangular
for $\ell\in\{1,2\}$. In view of \eqref{=000020continuity}, this
in fact constitutes only a \emph{single} restriction on the model
parameters, that $\Phi_{0,12}=0$, and so is exactly identifying rather
than over-identifying. The second equation of the nonlinear SVAR \eqref{=000020model}
can in this case be estimated by nonlinear regression (with $\pi_{t}$
as the dependent variable), as was done by \citet{BenignoEggertsson2023}.

\subsection{Testing for linearity in the Phillips curve}

Let $\{\Gamma_{i}^{(\ell)}\}$ momentarily denote the SVAR parameters
corresponding to the recursive identification scheme of \citet{BenignoEggertsson2023}.
In light of \subsecref{orthogreducedform}, because of the lower-triangular
structure imposed on the Jacobian $\Phi_{0}^{(1)}$ of $\fe_{0}$
(at some nominated point $z_{0}$ in the `normal' regime), these
are the coefficients associated with the orthogonal reduced-form parametrisation
\eqref{svar-equiv} of the SVAR, when the $\ga_{j}$ are modelled
as piecewise linear. \thmref{shockid}, together with a sign-normalisation
of the shocks, then implies that \emph{all} observationally equivalent
models can be obtained by a \emph{common} rotation of the recursively
identified model, i.e.\  $\Phi_{i}^{(\ell)}=Q\Gamma_{i}^{(\ell)}$
for $\ell\in\{1,2\}$ and $i\in\{0,\ldots k\}$, where $Q\in\orths(p)$
with $\det Q>0$.

Because $Q$ is \emph{not} regime dependent, every observationally
equivalent parametrisation of the model obtained in this way will
exhibit regime dependence if, and only if, this is also true of the
parameters $\{\Gamma_{i}^{(\ell)}\}$ obtained under the \citet{BenignoEggertsson2023}
identification scheme. The presence of some regime dependence in $\{\Gamma_{i}^{(\ell)}\}$
is thus a necessary condition for the existence of a nonlinear Phillips
curve under \emph{any} identification scheme. Since the null hypothesis
of no regime dependence in $\{\Gamma_{i}^{(\ell)}\}$ is testable,
a failure to reject it would provide evidence, in favour of a linear
Phillips curve, that is robust to all possible identifying schemes.
(In this respect, our imposition of the \citealp{BenignoEggertsson2023},
restrictions merely provides a convenient way to normalise the system,
in the manner of \subsecref{orthogreducedform}).

Observe that the specification of \eqref{=000020model} allows for
nonlinearities in all lags of the SVAR. This permits the dynamic response
of inflation to labour market tightness shocks to be nonlinear, even
if the impact responses are linear, i.e., even if $\Phi_{0}^{(\ell)}$
is regime-invariant. We therefore consider two separate tests of linearity.
The first tests
\begin{equation}
H_{0}^{\mathrm{NS}}:\Phi_{0}^{\left(1\right)}=\Phi_{0}^{\left(2\right)}\qquad\text{v.}\qquad H_{1}^{\mathrm{NS}}:\Phi_{0}^{(1)}\neq\Phi_{0}^{(2)}.\label{eq:no-switching}
\end{equation}
The null hypothesis $H_{0}^{\mathrm{NS}}$ can be interpreted as saying
that there is no \emph{endogenous} regime switching, and implies that
the impact effect of labour tightness shocks on inflation does not
depend on the state of the labour market.

However, $H_{0}^{\mathrm{NS}}$ does not exclude the possibility that
$\Phi_{i}^{(1)}\neq\Phi_{i}^{(2)}$ for some $i\in\{1,\ldots,k\}$,
in which case the dynamic effects of tightness shocks may still be
regime dependent, at longer horizons. This motivates our second,
more restrictive hypothesis:
\begin{equation}
H_{0}^{\mathrm{lin}}:\Phi_{i}^{(1)}=\Phi_{i}^{(2)}\sep\forall i\in\{0,\ldots,k\}\qquad\text{v.}\qquad H_{1}^{\mathrm{lin}}:\Phi_{i}^{(1)}\neq\Phi_{i}^{(2)}\sep\text{for some }i,\label{eq:linearity}
\end{equation}
which under the null entails a linear SVAR. Failure to reject $H_{0}^{\mathrm{lin}}$
would suggest that a linear SVAR provides an adequate description
of the dynamic causal effects (modulo the usual invertibility caveats),
and thus that the Phillips curve is linear, in a very strong sense,
under any identification scheme.

\subsection{Results}

We use the data from the 2025 version of \citet{BenignoEggertsson2023},
available on the authors' websites. Specifically, inflation $\pi_{t}$
is the quarterly, annualised core CPI inflation (excluding food and
energy), constructed from monthly CPI data averaged to quarterly frequency
and sourced from the BLS via FRED. The vacancy-to-unemployment ratio
$\theta_{t}=v_{t}/u_{t}$ is the ratio of job vacancies to unemployed
workers, using the \citet{Barnichon2010} vacancy series (as updated
by the author), also averaged from a monthly to a quarterly frequency.
We estimate the piecewise linear SVAR \eqref{=000020model} with two
lags ($k=2$) over the sample periods 1960Q1-2024Q4 and 2008Q1-2024Q4,
to mirror the analysis of \citet{BenignoEggertsson2023}.

\subsubsection{Testing linearity}

\begin{table}
\centering %
\begin{tabular}{lccc}
\toprule 
Null Hypothesis  & Restrictions & \multicolumn{2}{c}{LR Statistic {[}$p$-value{]}}\tabularnewline
 &  & 1960Q1--2024Q4 & 2008Q1--2024Q4\tabularnewline
\midrule 
No Endogenous Switching \eqref{no-switching} & 2 & 21.7 {[}0.00{]} & 38.6 {[}0.00{]}\tabularnewline
Linear SVAR \eqref{linearity} & 6 & 38.0 {[}0.00{]} & 51.2 {[}0.00{]}\tabularnewline
\bottomrule
\end{tabular}\smallskip{}
\foreignlanguage{english}{ }%
\parbox[c]{0.9\textwidth}{%
{\small\textit{Notes:}}{\small{} The model is a bivariate SVAR in log
vacancy--unemployment ratio (log $\theta$) and wage inflation with
two lags and two regimes, determined by the sign of (log $\theta$).
Both tests are against the alternative of an unrestricted piecewise
linear SVAR. Asymptotic $p$-values based on the $\chi^{2}$ distribution
with degrees of freedom equal to the number of restrictions.}%
} 

\caption{Likelihood ratio tests of linearity hypotheses}

\label{tab:lr_tests_svar}
\end{table}

\tabref{lr_tests_svar} reports likelihood ratio (LR) tests of our
two linearity hypotheses: $H_{0}^{\mathrm{NS}}$ (no endogenous regime
switching) and $H_{0}^{\mathrm{lin}}$ (fully linear SVAR). The results
clearly reject the linearity hypothesis, both in its weak and strong
forms. The apparent deterioration in fit of the linear models is even
stronger in the shorter, more recent, sample. 

Even though failure to reject would have been conclusive evidence
against nonlinearity, these results are not enough to conclude that
the Phillips curve itself, being only one equation in our bivariate
system, is nonlinear. They imply that impulse responses to identified
structural inflation and tightness shocks will be significantly state-dependent
under any identification scheme, but it remains to be seen what this
state dependence looks like in the Phillips curve that emerges from
any specific identification scheme. We turn to this question next.

\subsubsection{Phillips curve slope}

\begin{figure}
\begin{centering}
\includegraphics[viewport=0bp 0bp 600bp 376bp,clip,width=0.9\textwidth]{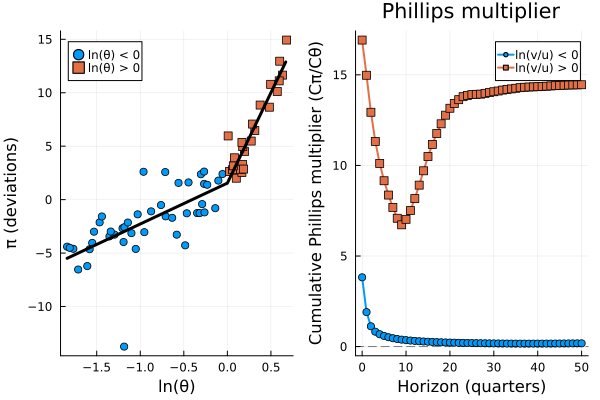}
\par\end{centering}
{\footnotesize\emph{Left panel:}}{\footnotesize{} scatter plot of inflation
deviations (inflation after removing all right-hand side contributions
except $\log\theta_{t}$) against $\log\theta_{t}$, sample 2008Q1--2024Q4.
Solid lines show the estimated regime-specific Phillips curves.}{\footnotesize\emph{
Right panel:}}{\footnotesize{} cumulative Phillips multiplier (ratio
of cumulative inflation IRF to cumulative tightness IRF) under each
regime, sample 2008Q1--2024Q4. IRFs are computed starting from 2009Q3
($\log\theta_{t}=-1.84$, loose labour market regime) and 2022Q2 ($\log\theta_{t}=0.68$,
labour shortage regime).}{\footnotesize\par}

\caption{Nonlinear Phillips curve and state-dependent multipliers}

{\footnotesize\label{fig:nonlinear_PC}}{\footnotesize\par}
\end{figure}

Further evidence on the nonlinearity of the Phillips curve is obtained
by computing estimates of its slope under both regimes. We do this
in two different ways. First, we produce a kinked Phillips curve plot
(equivalent to Figure 6(b) of \citealp{BenignoEggertsson2023}). This
is shown in the left panel of \figref{nonlinear_PC}. The scatterplot
shows inflation after removing the effects of all explanatory variables
from the supply equation in model \eqref{=000020model} except $\log\theta_{t}$.
The solid lines trace out the estimated Phillips curve in the $(\log\theta_{t},\pi_{t})$
space. In particular, the slope coefficient under each regime is computed
as $-\Phi_{0,21}^{\left(\ell\right)}/\Phi_{0,22}$, which is given
by the equation in the bottom row of \eqref{=000020model}, solved
for $z_{2,t}=\pi_{t}$, and using the fact that $\Phi_{0,22}$ is
regime-independent, as per \eqref{=000020continuity}. For the 2008Q1-{}-2024Q4
sample, the estimated slopes are $\hat{\beta}^{(1)}=3.82$ ($\log\theta_{t}\leq0$
regime) and $\hat{\beta}^{(2)}=16.92$ ($\log\theta_{t}>0$ regime).

The right panel of \figref{nonlinear_PC} shows a dynamic Phillips
curve multiplier under each regime, computed from the state-dependent
IRFs. We choose two starting points that are representative of the
two regimes: 2009Q3 ($\log\theta_{t}=-1.84$, the Great Recession
trough) for the $\log\theta_{t}\leq0$ regime, and 2022Q2 ($\log\theta_{t}=0.68$,
the post-COVID peak) for the $\log\theta_{t}>0$ regime. The multiplier
is the ratio of the cumulative inflation IRF (at horizon $h$) to
the cumulative tightness IRF following a market tightness shock which
raises $\log\theta_{t}$ by 1 unit over the next $h$ periods: 
\[
\text{Slope}_{h}^{PC}=\frac{\sum_{s=0}^{h}\frac{\partial\pi_{t+s}}{\partial\err_{\theta,t}}}{\sum_{s=0}^{h}\frac{\partial\log\theta_{t+s}}{\partial\err_{\theta,t}}}.
\]

Both approaches show a substantially steeper Phillips curve in the
tight labour market regime $\log\theta_{t}>0$ compared to the loose
labour market regime $\log\theta_{t}\leq0$. The results are qualitatively
and quantitatively consistent with \citet{BenignoEggertsson2023},
which is not surprising given that we used the same identifying assumption
as them.

\section{Extensions}

\label{sec:hetero}

The appearance of a nonlinear transformation on the l.h.s.\ of the
(endogenously) nonlinear SVAR 
\begin{equation}
\fe_{0}(z_{t})=\b{\fe}_{1}(\z_{t-1})+\err_{t}\label{eq:nlsvaragain}
\end{equation}
entails that the model automatically accommodates certain forms of
regime-dependent heteroskedasticity. This can be readily seen, for
example, when $\fe_{0}$ has the piecewise linear form
\[
\fe_{0}(z_{t})=\sum_{\ell=1}^{L}\indic\{z_{t}\in\set Z^{(\ell)}\}\Phi_{0}^{(\ell)}z_{t}.
\]
In this case, whenever the r.h.s.\ of the model is such that $z_{t}\in\set Z^{(\ell_{t})}$,
the model behaves \emph{locally} like a linear SVAR, with reduced
form
\[
z_{t}=(\Phi_{0}^{(\ell_{t})})^{-1}\b{\fe}_{1}(\z_{t-1})+(\Phi_{0}^{(\ell_{t})})^{-1}\err_{t},
\]
for all $\err_{t}$ such that $z_{t}$ continues to lie in $\set Z^{(\ell_{t})}$.
(Note that unlike in a model with \emph{exogenous} regimes, $\ell_{t}$
depends on $\err_{t}$, and so the preceding does not hold for all
$\err_{t}$.)

Nonetheless, in some situations it may be desirable to augment the
model to allow for ARCH-type conditional heteroskedasticity, in which
the variances of the structural shocks depend on certain (observed)
predetermined variables. To that end, consider the following extension
of \eqref{nlsvaragain}, to
\begin{equation}
\fe_{0}(z_{t})=\b{\fe}_{1}(\z_{t-1}^{(1)},\z_{t-1}^{(2)},v_{t-1})+\sigma(\z_{t-1}^{(2)},v_{t-1})\err_{t},\label{eq:nlsemhet}
\end{equation}
where $\b z_{t-1}^{(1)}$ and $\b z_{t-1}^{(2)}$ partition (into
vectors of dimension $d_{(1)}+d_{(2)}=kp$) the elements of $\b z_{t-1}$,
while $\{v_{t}\}$ is strictly exogenous in the sense of being independent
of $(\b z_{0},\{\err_{t}\})$, and takes values in the (possibly discrete)
set $\mathcal{V}\subset\reals^{d_{v}}$. (Rather than requiring $\{v_{t}\}$
to be stationary, we suppose that there is a measure $\mu_{v}$ on
$\mathcal{V}$ to which the distribution of $v_{t}$ is equivalent,
for every $t\geq0$.)

The skedastic function, $\sigma(\cdot)$, allows the volatilities
of the structural shocks
\[
w_{t}\defeq\sigma(\z_{t-1}^{(2)},v_{t-1})\err_{t}
\]
to depend on $(\z_{t-1}^{(2)},v_{t-1})$; we require $\sigma(\cdot)$
to be a diagonal matrix (with strictly positive entries), so that
the structural shocks $w_{t}$ remain mutually uncorrelated (cf.\ Section~14.2
of \citealp{KL17book}). By introducing $\{v_{t}\}$, we also extend
the model so as to permit the r.h.s.\ to depend on processes that
are exogenous to the SVAR (such as deterministic processes). We continue
to maintain that $\{\err_{t}\}$ is i.i.d.\ with mean zero and variance
$I_{G}$, and moreover that $\err_{t+1}$ is independent of $(\b z_{0},\{\err_{s},v_{s}\}_{s\leq t})$,
for all $t\geq0$.

Under the assumptions given below, the augmented model \eqref{nlsemhet}
yields the following (time-invariant) density for $z_{t}$ conditional
on $(\b z_{t-1},v_{t-1})$,
\[
\varphi_{z_{t}\mid\b z_{t-1},v_{t-1}}(\xi\mid\b{\xi}_{-1},\upsilon)=\den\{\sigma(\b{\xi}_{-1}^{(2)},\upsilon)^{-1}[\fe_{0}(\xi)-\b{\fe}_{1}(\b{\xi}_{-1}^{(1)},\b{\xi}_{-1}^{(2)},\upsilon)]\}\cdot\smlabs{\det D\fe_{0}(\xi)},
\]
where $\b{\xi}_{-1}\in\reals^{kp}$ is partitioned into $(\b{\xi}_{-1}^{(1)},\b{\xi}_{-1}^{(2)})$
conformably with that of $\b z_{t-1}$ into $(\b z_{t-1}^{(1)},\b z_{t-1}^{(2)})$.
Since the likelihood for $\{z_{t}\}_{t=1}^{n}$ conditional on $(\b z_{0},\{v_{t}\}_{t=0}^{n-1})$
can be expressed entirely in terms of these conditional densities,
we continue to regard two alternative parametrisations as being observationally
equivalent if they yield the same $\varphi_{z_{t}\mid\b z_{t-1},v_{t-1}}$
(up to the usual a.e.\ equivalences), similarly to \secref{svarident}
above. 

The parameters $(\fe_{0},\b{\fe}_{1},\sigma)$ of the model \eqref{nlsemhet}
are, in a quite trivial sense, indistinguishable from $(\Lambda\fe_{0},\Lambda\b{\fe}_{1},\Lambda\sigma)$,
if $\Lambda$ is a diagonal matrix with strictly positive entries.
Such a rescaling has no effect on the (scale-normalised) impulse responses
implied by the model parameters, and is merely a consequence of the
lack of a scale normalisation in \eqref{nlsemhet} -- something that
was previously delivered, in the context of \eqref{nlsvaragain},
by the requirement that $\expect\err_{t}\err_{t}^{\trans}=I_{p}$.
Letting $\set S\ni\sigma$ denote the parameter space for the skedastic
function, we may fix the overall scale of the model by requiring every
$\tilde{\sigma}\in\set S$ to satisfy
\begin{equation}
\tilde{\sigma}(\b z^{(2)\ast},v^{\ast})=I_{p},\label{eq:scalenorm}
\end{equation}
at some (user specified) value of $(\b z^{(2)\ast},v^{\ast})\in\reals^{d_{(2)}}\times\mathcal{V}$.
(To prevent this from being satisfied simply by a modification of
$\tilde{\sigma}$ on a null set, we further maintain that $\tilde{\sigma}$
is continuous at $(\b z^{(2)\ast},v^{\ast})$, and that $\mu_{v}$
has strictly positive measure in every neighbourhood of $v^{\ast}$.)

Here we shall also relax the requirement that $\b f_{1}$ be continuous
in \emph{all} of its arguments: in fact we only require continuity
of $\b z^{(1)}\elmap\b{\fe}_{1}(\z^{(1)},\z^{(2)},v)$, at the cost
of a strengthening of the surjectivity condition given in \assref{svardgp}\ass{.}\ref{enu:shockid:surj}
above. This reflects the crucial role that the variables $\z_{t-1}^{(1)}$,
which are excluded from the skedastic function, now play in delivering
the identification of the model parameters.

\assumpname{EXT}
\begin{assumption}
\label{ass:ext}\assref{svarpar} and \assref{svardgp} hold with
only the following modifications, which apply for every $(\b z^{(2)},v)\in\reals^{d_{(2)}}\times\mathcal{V}$:
\begin{enumerate}[leftmargin=1.5cm]
\item[\ass{PS.1$^\prime$}] for every $\tilde{\fe}_{0}\in\fespc_{0}$ and $\tilde{\b{\fe}}\in\bigfspc_{1}$:
$\tilde{\fe}_{0}$ and $\b z^{(1)}\elmap\tilde{\b{\fe}}_{1}(\b z^{(1)},\b z^{(2)},v)$
are locally Lipschitz;
\item[\ass{DGP.1$^\prime$}] $\b z^{(1)}\elmap\b{\fe}_{1}(\b z^{(1)},\b z^{(2)},v)$ is surjective
(onto $\reals^{p}$), with $\rank D_{\b z^{(1)}}\b{\fe}(\b z^{(1)},\b z^{(2)},v)=p$
for almost every $\b z^{(1)}\in\reals^{d_{(1)}}$.
\end{enumerate}
Moreover, for every $\tilde{\sigma}\in\set S$: $\tilde{\sigma}(\b z^{(2)},v)$
is a $(p\times p)$ diagonal matrix with strictly positive entries,
for every $(\b z^{(2)},v)\in\reals^{d_{(2)}}\times\mathcal{V}$; and
the scale normalisation \eqref{scalenorm} holds.
\end{assumption}

We may thus state the main result of this section, which extends \thmref{shockid}
above by allowing for: (i) ARCH-type heteroskedasticity; (ii) dependence
of the r.h.s.\ of the model on an exogenous process $\{v_{t}\}$,
and (iii) $\b{\fe}_{1}$ to be discontinuous in some arguments.
\begin{thm}
\label{thm:svarhetero}Suppose that \assref{ext} holds. Then there
exists a $\tilde{\sigma}\in\set S$ and a $\tilde{\den}\in\denspc$
such that $(\tilde{\fe}_{0},\tilde{\b{\fe}}_{1},\tilde{\sigma},\tilde{\den})$
is observationally equivalent to $(\fe_{0},\b{\fe}_{1},\sigma,\den)$,
if and only if there exists a $Q\in\orths(p)$ such that, for almost
every $\b z^{(2)}\in\reals^{d_{(2)}}$ and $\mu_{v}$-almost every
$v\in\mathcal{V}$:
\begin{align*}
\tilde{\fe}_{0}(z) & =Q\fe_{0}(z)\sep\forall z\in\reals^{p}, & \tilde{\bigf}_{1}(\z^{(1)},\z^{(2)},v) & =Q\bigf_{1}(\z^{(1)},\z^{(2)},v)\sep\forall\z^{(1)}\in\reals^{d_{(1)}},
\end{align*}
and
\begin{equation}
Q\sigma^{2}(\z^{(2)},v)Q^{\trans}\label{eq:ske-id}
\end{equation}
is a diagonal matrix; in which case $\tilde{\sigma}^{2}(\z^{(2)},v)=Q\sigma^{2}(\z^{(2)},v)Q^{\trans}$.
\end{thm}

Since the skedastic function must be a diagonal matrix, \eqref{ske-id}
may provide further restrictions on $Q$; the extent of these will
depend on the properties of the actual skedastic function $\sigma$.
On the one hand, suppose that $\sigma(\z^{(2)},v)=\lambda(\z^{(2)},v)I_{p}$
is \emph{always} a rescaling of the identity matrix. Then \eqref{ske-id}
yields a diagonal matrix for \emph{every} $Q\in\orths(p)$, and no
further restrictions on $Q$ are implied. On the other hand, suppose
that $\sigma(\z^{(2)},v)$ varies in such a way that it is not always
proportional to the identity matrix, so that the variances of some
of the structural shocks may differ from each other, at least for
certain values of $(\z^{(2)},v)$. In particular, if there exists
a $(\z^{(2)\dagger},v^{\dagger})\in\reals^{d_{(2)}}\times\mathcal{V}$
such that all the (diagonal) entries of $\sigma(\z^{(2)\dagger},v^{\dagger})$
are distinct, then $Q$ must be a signed permutation matrix (as follows
from Theorem~2.5.4 in \citealp{HJ13book}; cf.\ Proposition~1 in
\citealp{LLM10JEDC}), in which case the structural impulse response
functions are identified, up to a signing and economic `labelling'
of the shocks. In this way, we here obtain exactly the same kinds
of restrictions that are familiar from the \emph{linear} SVAR literature
on `identification by heteroskedasticity' (see e.g.\ the discussion
in Sections~14.2--14.3 of \citealp{KL17book}).

{\singlespacing

\bibliographystyle{ecta}
\bibliography{cksvar}

}

\appendix

\section{Proofs of main identification results}

\label{app:nlsem}

\subsection{Reformulation of the problem}

\label{app:nonlinearSEM}

While the nonlinear SVAR of \secref{svarident} is a (dynamic) simultaneous
equations model (SEM), our notion of observational equivalence refers
only to the distribution of $z_{t}$ conditional on its lags. This
allows the proof of \thmref{shockid} to be approached in a manner
that entirely abstracts from the dynamics of the SVAR. To more clearly
connect our underlying identification results with those of the literature
on nonlinear simultaneous equations models (SEMs), in particular \citet{Matz09Ecta},
in this appendix we consider the nonlinear SEM
\begin{equation}
U=r(Y,X)=r_{0}(Y)+r_{1}(X),\label{eq:matzkin}
\end{equation}
where $U$ and $Y$ are random vectors taking values in $\reals^{G}$,
and $X$ is a random vector taking values in $\reals^{K}$, where
$K\geq G$. Let $f_{U}$ denote the density of $U$, location- and
scale-normalised so that $\expect U=0$ and $\expect UU^{\trans}=I_{G}$.
This is the same model as in (2.1) of \citet{Matz09Ecta}, but with
the additional restriction that $r$ is (additively) separable in
the endogenous and exogenous variables, $Y$ and $X$. Our results
on observational equivalence in this model are given as \thmref{matzkin}
below: on the basis of which the proof of \thmref{shockid} will simply
be a matter of translating between the notation of the SVAR in \secref{svarident},
and that of \eqref{matzkin} (see \appref{SVARidproof} below).

Under the regularity conditions given below, if we suppose that $X$
has Lebesgue density $f_{X}$ with support $\reals^{K}$, then the
model implies that the distribution of $Y$ conditional on $X$ has
a Lebesgue density that satisfies (see e.g.\ \citealp{EG15}, Thm.~3.9)
\[
f_{Y\mid X}(y\mid x)=f_{U}[r(y,x)]\cdot\det Dr_{0}(y)=f_{U}[r_{0}(y)+r_{1}(x)]\cdot\det Dr_{0}(y)
\]
a.e.\ $(y,x)\in\reals^{G+K}$; here the `a.e.'\ qualifier is a
consequence both of the usual non-uniqueness of the conditional density
(with respect to modifications on a null set), and more importantly
the fact the Jacobian $Dr_{0}(y)$ need only exist a.e. We will accordingly
say that two alternative parametrisations $(\tilde{r}_{0},\tilde{r}_{1},f_{\tilde{U}})$
and $(r_{0},r_{1},f_{U})$ are \emph{observationally equivalent} if
\begin{equation}
f_{U}[r(y,x)]\cdot\det Dr_{0}(y)=f_{U}[\tilde{r}(y,x)]\cdot\det D\tilde{r}_{0}(y)\label{eq:oe}
\end{equation}
a.e.\ $(y,x)\in\reals^{G+K}$, i.e.\ if they imply the same density
for $Y$ conditional on $X$. (This accords exactly with the definition
of observational equivalence given in \subsecref{nonlinearsvar},
transposed from the nonlinear SVAR to the nonlinear SEM.)

The model is parametrised by the functions $r_{0}:\reals^{G}\setmap\reals^{G}$,
$r_{1}:\reals^{K}\setmap\reals^{G}$, and the density $f_{U}$. Let
$\Gamma_{i}\ni r_{i}$, for $i\in\{0,1\}$, and $\Phi\ni f_{U}$ denote
the sets of functions and densities that together comprise the model
parameter space. We make only weak assumptions on the elements of
those parameter spaces, and some further assumptions on the parameters
$(r_{0},r_{1},f_{U})$ that actually generated the data; for a discussion
of these conditions, as they are mirrored in the nonlinear SVAR, see
\subsecref{nonlinearsvar}.

\assumpname{SEM}
\begin{assumption}
\label{ass:sem} $\Gamma_{0}$, $\Gamma_{1}$ and $\Phi$ collect
every function such that:
\begin{enumerate}[label=\ass{A\arabic*.}, ref=\ass{A\arabic*}]
\item \label{enu:sem:lipschitz}$\tilde{r}_{0}\in\Gamma_{0}$ and $\tilde{r}_{1}\in\Gamma_{1}$
are locally Lipschitz (continuous).
\item \label{enu:sem:bij}$\tilde{r}_{0}\in\Gamma_{0}$ is a bijection $\reals^{G}\setmap\reals^{G}$,
with $\det D\tilde{r}_{0}(y)>0$ for almost every $y\in\reals^{G}$.
\item \label{enu:sem:psdens}$f_{\tilde{U}}\in\Phi$ is continuously differentiable,
with $f_{\tilde{U}}(u)>0$ for all $u\in\reals^{G}$, and
\begin{align*}
\int_{\reals^{G}}f_{\tilde{U}}(u)\diff u & =1, & \int_{\reals^{G}}uf_{\tilde{U}}(u)\diff u & =0, & \int_{\reals^{G}}uu^{\trans}f_{\tilde{U}}(u)\diff u & =I_{G}.
\end{align*}
\end{enumerate}
$(r_{0},f_{1},f_{U})$ are such that:
\begin{enumerate}[label=\ass{B\arabic*.}, ref=\ass{B\arabic*}]
\item \label{enu:sem:surj}$r_{1}:\reals^{K}\setmap\reals^{G}$ is surjective,
with $\rank Dr_{1}(x)=G$ for almost every $x\in\reals^{K}$;
\item \label{enu:sem:r0inv}$r_{0}^{-1}$ is locally Lipschitz; and
\item \label{enu:sem:density}$f_{U}$ has a \emph{local} maximum at some
$u^{\ast}\in\reals^{G}$, and is twice continuously differentiable
in a neighbourhood of $u^{\ast}$, with negative definite Hessian
there.
\end{enumerate}
\end{assumption}

We can now state our main result on observational equivalence in the
model \eqref{matzkin}. Recall that $\orths(m)$ denotes the set of
$m\times m$ orthogonal matrices; further define $\sorths(m)$ to
be the subset of these matrices with positive determinant.

\begin{thm}
\label{thm:matzkin}Suppose that \assref{sem} holds. Let $\tilde{r}_{i}\in\Gamma_{i}$
for $i\in\{0,1\}$. Then there exists an $f_{\tilde{U}}\in\Phi$ such
that $(\tilde{r}_{0},\tilde{r}_{1},f_{\tilde{U}})$ is observationally
equivalent to $(r_{0},r_{1},f_{U})$, if and only if there exists
a $Q\in\sorths(G)$ such that
\begin{equation}
\tilde{r}_{0}(y)+\tilde{r}_{1}(x)=Q[r_{0}(y)+r_{1}(x)]\label{eq:matzeq}
\end{equation}
for all $(y,x)\in\reals^{G}\times\reals^{K}$.
\end{thm}

Only the sum of $r_{0}(y)+r_{1}(x)$ is identified, because in view
of \eqref{matzkin} we cannot distinguish between $(r_{0},r_{1})$
and $(r_{0}-\delta,r_{1}+\delta)$ for any $\delta\in\reals^{G}$.
This indeterminacy can of course be resolved by imposing a location
normalisation on either of these functions, e.g.\ by requiring $\tilde{r}_{0}(0)=0$
for all $\tilde{r}_{0}\in\Gamma_{0}$.

\subsection{Preliminaries}

For ease of reference, the following lemma collects some useful (and
well known) results regarding the properties of locally Lipschitz
functions, that will be relied on in the proof. Note when we say that
a function $g:\reals^{k}\setmap\reals^{\ell}$ is differentiable at
$x_{0}\in\reals^{k}$, we mean that there exists a (Jacobian) matrix
$Dg(x_{0})\in\reals^{\ell\times k}$, such that
\[
g(x)-g(x_{0})=Dg(x_{0})(x-x_{0})+o(\smlnorm{x-x_{0}})
\]
as $x\goesto x_{0}$. When we refer to the `measure' of a subset
of Euclidean space, we always mean its Lebesgue measure, unless otherwise
stated.
\begin{lem}
\label{lem:lipschitz}Suppose that $g:\reals^{k}\setmap\reals^{\ell}$
is locally Lipschitz. Then
\begin{enumerate}
\item \label{enu:rademacher}$g$ is differentiable a.e.;
\item \label{enu:luzin}if $k\leq\ell$, and $N\subseteq\reals^{k}$ has
measure zero (in $\reals^{k}$), then $g(N)$ has measure zero (in
$\reals^{\ell}$);
\item \label{enu:constant}if $Dg(x)=B$ for almost every $x\in\reals^{k}$,
then $g(x)=a+Bx$ for all $x\in\reals^{k}$; and
\item \label{enu:rigidity}if $\ell=k\geq2$, and $Dg(x)\in\sorths(k)$
for almost every $x\in\reals^{k}$, then $g(x)=a+Qx$ for some $Q\in\sorths(k)$. 
\end{enumerate}
Suppose that $k=\ell$, $g$ is bijective, and $g^{-1}$ and $h:\reals^{k}\setmap\reals^{m}$
are locally Lipschitz. Then
\begin{enumerate}[resume]
\item \label{enu:chain}for almost every $x\in\reals^{k}$, $f\defeq h\compose g$
is differentiable at $x$, and
\[
Df(x)=Dh[g(x)]Dg(x).
\]
\end{enumerate}
\end{lem}
\begin{proof}
\ref{enu:rademacher}. This is Rademacher's theorem (e.g.\ Theorem~3.2
in \citealp{EG15}).

\ref{enu:luzin}. This follows by Lemma~2.2(i), Theorem~2.5 and
Theorem~2.8(i) in \citet{EG15}.

\ref{enu:constant}. Fix $x_{0}\in\reals^{k}$. Since the locally
Lipschitz function $f(x)\defeq g(x)-Bx$ has $Df(x)=0$ a.e., and
is absolutely continuous along the segment joining any point $x\in\reals^{k}$
to $x_{0}$, it must be constant along that segment, by the fundamental
theorem of calculus. Hence $f(x)=f(x_{0})\eqdef a$ for all $x$.

\ref{enu:rigidity}. This follows from Theorem~3.1 (and the discussion
on p.\ 1469) in \citet{FJM02CPAM} -- see also Theorem~IV in \citet{John61CPAM}
-- and part~\ref{enu:constant}.

\ref{enu:chain}. Let $G\subseteq\reals^{k}$ and $H\subseteq\reals^{k}$
collect the points at which $g$ and $h$ are respectively differentiable.
Then $\reals^{k}\backslash H$ has measure zero, and since $g^{-1}$
is surjective and locally Lipschitz, it follows from $\reals^{k}=g^{-1}(\reals^{k}\backslash H)\union g^{-1}(H)$
and part~\ref{enu:luzin} that $\reals^{k}\backslash g^{-1}(H)$
also has measure zero. Deduce that the complement of $X\defeq G\intsect g^{-1}(H)$
has measure zero, and that for every $x\in X$, $g$ is differentiable
at $x$, and $h$ is differentiable at $g(x)$. Thus the chain rule
yields the result.
\end{proof}

\subsection{Proof of \thmref{matzkin}}

It is clear that if \eqref{matzeq} holds, then 
\[
\tilde{U}\defeq\tilde{r}_{0}(Y)+\tilde{r}_{1}(X)=Q[r_{0}(Y)+r_{1}(X)]=QU
\]
will be independent of $X$, with a density $f_{\tilde{U}}$ that
satisfies \assref{sem}\ass{.}\ref{enu:sem:psdens}; hence observational
equivalence obtains in this case. It remains therefore to prove the
reverse implication.

To that end, we suppose that $(\tilde{r}_{0},\tilde{r}_{1},f_{\tilde{U}})$
is observationally equivalent to $(r_{0},r_{1},f_{U})$. Taking logs
in \eqref{oe}, as we may under \assref{sem}\ass{.}\ref{enu:sem:bij}\ass{--}\ref{enu:sem:psdens},
yields that 
\begin{equation}
\log f_{U}[r(y,x)]-\log f_{\tilde{U}}[\tilde{r}(y,x)]=\log\det D\tilde{r}_{0}(y)-\log\det Dr_{0}(y)\label{eq:oe-rewritten}
\end{equation}
a.e.\ $(y,x)\in\reals^{G+K}$. In view of \assref{sem}\ass{.}\ref{enu:sem:lipschitz}\ass{--}\ref{enu:sem:bij}
and \assref{sem}\ass{.}\ref{enu:sem:surj}, we may define a set
$\mathcal{A}\subset\reals^{G+K}$, whose complement has measure zero
(in $\reals^{G+K}$), such that for every $(y,x)\in\mathcal{A}$:
\begin{itemize}
\item \eqref{oe-rewritten} holds;
\item $r_{0}$ and $\tilde{r}_{0}$ are differentiable at $y$, with $\det Dr_{0}(y)>0$
and $\det D\tilde{r}_{0}(y)>0$; and
\item $r_{1}$ and $\tilde{r}_{1}$ are differentiable at $x$, with $\rank Dr_{1}(x)=G$.
\end{itemize}
By Tonelli's theorem, we may also define sets $\mathcal{Y}\subset\reals^{G}$
and $\mathcal{X}\subset\reals^{K}$, whose complements (in $\reals^{G}$
and $\reals^{K}$ respectively) have measure zero, such that:
\begin{itemize}
\item for every $y_{0}\in\mathcal{Y}$: $(y_{0},x)\in\mathcal{A}$ for almost
every $x\in\reals^{K}$; and
\item for every $x_{0}\in\mathcal{X}$: $(y,x_{0})\in\mathcal{A}$ for almost
every $y\in\reals^{G}$.
\end{itemize}
The proof now proceeds in five steps. (Had we made imposed the stronger
requirement that $\tilde{r}_{0}$ and $\tilde{r}_{1}$ be twice continuously
differentiable, then the claims proved in the first two steps would
follow more directly as corollaries to the results of \citet{Matz09Ecta},
particularly her Theorem~3.3; and indeed our arguments in those parts
of the proof largely follow hers, suitably modified to allow $\tilde{r}_{0}$
and $\tilde{r}_{1}$ to have points of non-differentiability.)

\subsubsection*{(i) Claim: $\protect\rank D\tilde{r}_{1}(x)=G$ for all $x\in\mathcal{X}$.}

Let $x_{0}\in\mathcal{X}$ be given. Differentiating both sides of
\eqref{oe-rewritten} with respect to $x$, we obtain
\begin{equation}
D(\log f_{U})[r(y,x_{0})]Dr_{1}(x_{0})=D(\log f_{\tilde{U}})[\tilde{r}(y,x_{0})]D\tilde{r}_{1}(x_{0})\label{eq:diffedoe}
\end{equation}
a.e.\ $y\in\reals^{G}$. By the continuity of both sides in $y$,
this holds for all $y\in\reals^{G}$. Recall that $\rank Dr_{1}(x_{0})=G$
by the definition of $\mathcal{X}$; we must show that this is transmitted
to $D\tilde{r}_{1}(x_{0})$.

Under \assref{sem}\ass{.}\ref{enu:sem:density}, it follows from
the inverse function theorem that the map 
\[
u\elmap D(\log f_{U})(u)^{\trans}
\]
is invertible in a neighbourhood of $u=u^{\ast}$, and equals zero
at $u^{\ast}$. Hence by \assref{sem}\ass{.}\ref{enu:sem:bij},
the composite map
\[
y\elmap s(y,x_{0})\defeq D(\log f_{U})[r(y,x_{0})]^{\trans}=D(\log f_{U})[r_{0}(y)+r_{1}(x_{0})]^{\trans}
\]
is also invertible for $y$ in a neighbourhood of
\begin{equation}
y^{\ast}(x_{0})\defeq r_{0}^{-1}[u^{\ast}-r_{1}(x_{0})],\label{eq:yast}
\end{equation}
with the property that
\[
s[y^{\ast}(x_{0}),x_{0}]=D(\log f_{U})(u^{\ast})^{\trans}=0.
\]

Hence there exist $\lambda>0$ and $\{y^{i}\}_{i=1}^{G}$ such that
\[
s(y^{i},x_{0})=\lambda e_{i}
\]
for all $i\in\{1,\ldots,G\}$, where $e_{i}$ denotes the $i$th column
of $I_{G}$. Evaluating \eqref{diffedoe} at each $y^{i}$, we obtain
that 
\[
Dr_{1}(x_{0})^{\trans}e_{i}\subset\spn D\tilde{r}_{1}(x_{0})^{\trans}
\]
for $i\in\{1,\ldots,G\}$, whence $\spn Dr_{1}(x_{0})^{\trans}\subset\spn D\tilde{r}_{1}(x_{0})^{\trans}$.
Since $Dr_{1}(x_{0})^{\trans}$ has rank $G$, it follows that so
too does $D\tilde{r}_{1}(x_{0})^{\trans}$.

\subsubsection*{(ii) Claim: $\log\det D\tilde{r}_{0}(y)-\log\det Dr_{0}(y)$ is constant
on $\mathcal{Y}$.}

Let $J:\reals^{G}\setmap\reals^{G}$ be defined such that
\[
J(y)=\log\det D\tilde{r}_{0}(y)-\log\det Dr_{0}(y)
\]
for all $y\in\mathcal{Y}$, so that it equals the r.h.s.\ of \eqref{oe-rewritten}
there; and set $J(y)=0$ otherwise.

Consider again the map $y^{\ast}:\reals^{K}\setmap\reals^{G}$, defined
in \eqref{yast} above, which is surjective and locally Lipschitz
in view of \assref{sem}\ass{.}\ref{enu:sem:lipschitz} and \assref{sem}\ass{.}\ref{enu:sem:surj}\ass{--}\ref{enu:sem:r0inv}.
Hence the complement of $y^{\ast}(\mathcal{X})$ in $\reals^{G}$
has measure zero, by \lemref{lipschitz}\ref{enu:luzin}. Now fix
$y_{0}\in\mathcal{Y}\intsect y^{\ast}(\mathcal{X})$, whose complement
also has measure zero. By definition of $\mathcal{Y}$, \eqref{oe-rewritten}
holds at $(y_{0},x)$, for almost every $x$. Moreover, since both
sides of \eqref{oe-rewritten} are continuous in $x$, it follows
that
\begin{equation}
\log f_{U}[r(y_{0},x)]-\log f_{\tilde{U}}[\tilde{r}(y_{0},x)]=J(y_{0})\label{eq:oe-J}
\end{equation}
holds for \emph{every} $x\in\reals^{K}$. Since $y_{0}\in\mathcal{Y}$,
the l.h.s.\ is differentiable with respect to $y$, whence so too
is the r.h.s.,\ with
\begin{equation}
D(\log f_{U})[r(y_{0},x)]Dr_{0}(y_{0})-D(\log f_{\tilde{U}})[\tilde{r}(y_{0},x)]D\tilde{r}_{0}(y_{0})=DJ(y_{0})\label{eq:diffcond}
\end{equation}

Since $y_{0}\in y^{\ast}(\mathcal{X})$, there exists an $x_{0}\in\mathcal{X}$
such that $r_{0}(y_{0},x_{0})=u^{\ast}$, and hence 
\[
D(\log f_{U})[r(y_{0},x_{0})]=D(\log f_{U})(u^{\ast})=0.
\]
Since \eqref{diffedoe} holds at $(y_{0},x_{0})$, with $\rank D\tilde{r}_{1}(x_{0})=\rank Dr_{1}(x_{0})=G$
by the preceding part of the proof, it follows that 
\[
D(\log f_{\tilde{U}})[\tilde{r}(y_{0},x_{0})]=0.
\]
Deduce from \eqref{diffcond} that $DJ(y_{0})=0$ for all $y_{0}\in\mathcal{Y}\intsect y^{\ast}(\mathcal{X})$.
It then follows from \eqref{oe-J} above that for all $x\in\reals^{K}$,
the Jacobian of
\[
y\elmap\log f_{U}[r(y,x)]-\log f_{\tilde{U}}[\tilde{r}(y,x)]
\]
is zero at $y_{0}\in\mathcal{Y}\intsect y^{\ast}(\mathcal{X})$, i.e.\ almost
everwhere. Since this map is locally Lipschitz, it is therefore equal
to some constant $C$, by \lemref{lipschitz}\ref{enu:constant}.
Hence
\[
J(y)=\log\det D\tilde{r}_{0}(y)-\log\det Dr_{0}(y)=C
\]
for all $y\in\mathcal{Y}$.

\subsubsection*{(iii) Claim: $\tilde{r}_{1}(x)=\tilde{u}^{\ast}-\tilde{m}_{0}[u^{\ast}-r_{1}(x)]$,
for $\tilde{m}_{0}\protect\defeq\tilde{r}_{0}\protect\compose r_{0}^{-1}$.}

Returning now to \eqref{oe-rewritten}, it follows from the preceding
part of the proof that
\[
\log f_{\tilde{U}}[\tilde{r}_{0}(y)+\tilde{r}_{1}(x)]=\log f_{U}[r_{0}(y)+r_{1}(x)]-C
\]
a.e.\ $(y,x)\in\reals^{G+K}$; and since both sides are continuous
in $(y,x)$, the preceding must hold for \emph{all} $(y,x)\in\reals^{G+K}$.
Setting $\tilde{u}=\tilde{r}_{0}(y)+\tilde{r}_{1}(x)$, and recalling
that $\tilde{r}_{0}$ is invertible (by \assref{sem}\ass{.}\ref{enu:sem:bij}),
this may be equivalently stated as
\begin{align*}
\log f_{\tilde{U}}(\tilde{u}) & =\log f_{U}[r_{0}\{\tilde{r}_{0}^{-1}[\tilde{u}-\tilde{r}_{1}(x)]\}+r_{1}(x)]-C\\
 & =\log f_{U}\{\tilde{m}_{0}^{-1}[\tilde{u}-\tilde{r}_{1}(x)]+r_{1}(x)\}-C
\end{align*}
for all $(\tilde{u},x)\in\reals^{G+K}$, where $\tilde{m}_{0}\defeq\tilde{r}_{0}\compose r_{0}^{-1}$
is invertible and locally Lipschitz, by \assref{sem}\ass{.}\ref{enu:sem:lipschitz}
and \assref{sem}\ass{.}\ref{enu:sem:r0inv}. Since the l.h.s.\ of
the preceding does not depend on $x$, the r.h.s.\ must be invariant
to $x$, and so we have in particular that
\begin{equation}
f_{U}\{\tilde{m}_{0}^{-1}[\tilde{u}-\tilde{r}_{1}(0)]+r_{1}(0)\}=f_{U}\{\tilde{m}_{0}^{-1}[\tilde{u}-\tilde{r}_{1}(x)]+r_{1}(x)\}\label{eq:x-invar}
\end{equation}
for all $(\tilde{u},x)\in\reals^{G+K}$.

By taking $\tilde{u}$ in the preceding to be equal to
\begin{equation}
\tilde{u}^{\ast}\defeq\tilde{m}_{0}[u^{\ast}-r_{1}(0)]+\tilde{r}_{1}(0),\label{eq:uasttilde}
\end{equation}
for $u^{\ast}$ as in \assref{sem}\ass{.}\ref{enu:sem:density},
we obtain that
\begin{equation}
f_{U}(u^{\ast})=f_{U}\{\tilde{m}_{0}^{-1}[\tilde{u}^{\ast}-\tilde{r}_{1}(0)]+r_{1}(0)\}=f_{U}\{\tilde{m}_{0}^{-1}[\tilde{u}^{\ast}-\tilde{r}_{1}(x)]+r_{1}(x)\}\label{eq:fuast}
\end{equation}
for all $x\in\reals^{K}$. Defining the continuous map $\theta:\reals^{K}\setmap\reals^{G}$
as
\[
\theta(x)\defeq\tilde{m}_{0}^{-1}[\tilde{u}^{\ast}-\tilde{r}_{1}(x)]+r_{1}(x),
\]
which by \eqref{uasttilde} has $\theta(0)=u^{\ast}$, we may thus
rewrite \eqref{fuast} as
\begin{equation}
f_{U}(u^{\ast})=f_{U}[\theta(0)]=f_{U}[\theta(x)]\label{eq:fhall}
\end{equation}
for all $x\in\reals^{K}$.

The preceding entails that $f_{U}[\theta(x)]$ does not in fact depend
on $x$; we need to show that this implies that $\theta(x)$ itself
is invariant to $x$. By a second-order Taylor expansion of $f_{U}$
around $u=u^{\ast}$, in view of \assref{sem}\ass{.}\ref{enu:sem:density},
there exist $\epsilon,\eta>0$ such that
\[
\smlabs{f_{U}(u)-f_{U}(u^{\ast})}\geq\eta\smlnorm{u-u^{\ast}}^{2}
\]
for all $\smlnorm{u-u^{\ast}}<\epsilon$. Since $x\elmap\theta(x)$
is continuous with $\theta(0)=u^{\ast}$, the equalities in \eqref{fhall}
can hold for all $x\in\reals^{K}$ only if 
\[
\tilde{m}_{0}^{-1}[\tilde{u}^{\ast}-\tilde{r}_{1}(x)]+r_{1}(x)=\theta(x)=\theta(0)=u^{\ast}
\]
for all $x\in\reals^{K}$. Thus
\begin{equation}
\tilde{r}_{1}(x)=\tilde{u}^{\ast}-\tilde{m}_{0}[u^{\ast}-r_{1}(x)]\label{eq:r1tilde}
\end{equation}
for all $x\in\reals^{K}$.

\subsubsection*{(iv) Claim: $\tilde{m}_{0}$ is affine.}

For $v\in\reals^{G}$, define
\[
\delta(v)\defeq f_{U}\{\tilde{m}_{0}^{-1}[(v+\tilde{u}^{\ast})-\tilde{r}_{1}(0)]+r_{1}(0)\}
\]
which in view of \eqref{fuast} satisfies 
\begin{equation}
\delta(0)=f_{U}\{\tilde{m}_{0}^{-1}[\tilde{u}^{\ast}-\tilde{r}_{1}(0)]+r_{1}(0)\}=f_{U}(u^{\ast}).\label{eq:delta0}
\end{equation}
Noting that \eqref{x-invar} above holds for all $(\tilde{u},x)\in\reals^{G+K}$,
it follows that by taking $\tilde{u}=v+\tilde{u}^{\ast}$ there, we
obtain
\[
\delta(v)=f_{U}\{\tilde{m}_{0}^{-1}[(v+\tilde{u}^{\ast})-\tilde{r}_{1}(0)]+r_{1}(0)\}=f_{U}\{\tilde{m}_{0}^{-1}[(v+\tilde{u}^{\ast})-\tilde{r}_{1}(x)]+r_{1}(x)\}
\]
for all $x\in\reals^{K}$. By the preceding part of the proof (namely,
\eqref{r1tilde}), 
\[
\tilde{m}_{0}^{-1}[(v+\tilde{u}^{\ast})-\tilde{r}_{1}(x)]=\tilde{m}_{0}^{-1}\{v+\tilde{m}_{0}[u^{\ast}-r_{1}(x)]\},
\]
and hence
\[
\delta(v)=f_{U}[\tilde{m}_{0}^{-1}\{v+\tilde{m}_{0}[u^{\ast}-r_{1}(x)]\}+r_{1}(x)],
\]
with the r.h.s.\ being invariant to $x\in\reals^{K}$. Since by \assref{sem}\ass{.}\ref{enu:sem:surj}
the image of $r_{1}$ is the whole of $\reals^{G}$, we may conclude
that
\[
\delta(v)=f_{U}\{\tilde{m}_{0}^{-1}[v+\tilde{m}_{0}(u^{\ast}-w)]+w\}
\]
depends only on $v$, for all $w\in\reals^{G}$; equivalently,
\begin{equation}
\delta(v)=f_{U}\{\tilde{m}_{0}^{-1}[v+\tilde{m}_{0}(w)]+u^{\ast}-w\}\label{eq:deltavnew}
\end{equation}
for all $w\in\reals^{G}$.

To establish that $\tilde{m}_{0}$ is affine, we shall now consider
the behaviour of $\delta(v)$ in a neighbourhood of $v=0$. We first
note that $\delta(0)=f_{U}(u^{\ast})$ by \eqref{delta0} above, and
that by \assref{sem}\ass{.}\ref{enu:sem:density} $f_{U}(u)$ admits
the following second-order Taylor expansion,
\begin{equation}
f_{U}(u)-f_{U}(u^{\ast})=-\tfrac{1}{2}(u-u^{\ast})^{\trans}H(u-u^{\ast})+o(\smlnorm{u-u^{\ast}}^{2})\label{eq:fUexp}
\end{equation}
as $u\goesto u^{\ast}$, where $H$ is positive definite. We note
that for $w\in\reals^{G}$, $\tilde{m}_{0}^{-1}$ is differentiable
at the value of $\tilde{m}_{0}(w)$ if $\tilde{m}_{0}=\tilde{r}_{0}\compose r_{0}^{-1}$
is itself differentiable at $w$ with $\det D\tilde{m}_{0}(w)\neq0$.
Since $\tilde{r}_{0}$ and $r_{0}^{-1}$ are locally Lipschitz, and
the latter is invertible (by \assref{sem}\ass{.}\ref{enu:sem:lipschitz}\ass{--}\ref{enu:sem:bij}
and \assref{sem}\ass{.}\ref{enu:sem:r0inv}) it follows by \lemref{lipschitz}\ref{enu:chain}
that $\tilde{m}_{0}$ is differentiable a.e., with
\begin{equation}
D\tilde{m}_{0}(w)=D\tilde{r}_{0}[r_{0}^{-1}(w)]Dr_{0}^{-1}(w)=D\tilde{r}_{0}[r_{0}^{-1}(w)][Dr_{0}(w)]^{-1}\label{eq:Dm}
\end{equation}
which has nonzero determinant a.e., in view of \assref{sem}\ass{.}\ref{enu:sem:bij}.
Thus there exists a set $\mathcal{B}\subset\reals^{G}$, whose complement
has measure zero, such that $\tilde{m}_{0}^{-1}$ is differentiable
at the value of $\tilde{m}_{0}(w)$, for every $w\in\mathcal{B}$.
Taking $w\in\mathcal{B}$, $\lambda>0$ and $d\in\reals^{G}\backslash\{0\}$,
and setting $v=\lambda d$, we obtain that
\begin{align}
 & \lambda^{-1}[\{\tilde{m}_{0}^{-1}[\lambda d+\tilde{m}_{0}(w)]+u^{\ast}-w\}-u^{\ast}]\nonumber \\
 & \qquad\qquad=\lambda^{-1}[\{\tilde{m}_{0}^{-1}[\lambda d+\tilde{m}_{0}(w)]+u^{\ast}-w\}-\{\tilde{m}_{0}^{-1}[\tilde{m}_{0}(w)]+u^{\ast}-w\}]\nonumber \\
 & \qquad\qquad\goesto(D\tilde{m}_{0}^{-1})[\tilde{m}_{0}(w)]d\nonumber \\
 & \qquad\qquad=[D\tilde{m}_{0}(w)]^{-1}d\label{eq:diffpart}
\end{align}
as $\lambda\goesto0$. Hence \eqref{delta0}, \eqref{deltavnew},
\eqref{fUexp} and \eqref{diffpart} yield
\begin{align*}
\lambda^{-2}[\delta(v)-\delta(0)] & =\lambda^{-2}[f_{U}\{\tilde{m}_{0}^{-1}[\lambda d+\tilde{m}_{0}(w)]+u^{\ast}-w\}-f_{U}(u^{\ast})]\\
 & \goesto-\tfrac{1}{2}d^{\trans}[D\tilde{m}_{0}(w)^{\trans}]^{-1}H[D\tilde{m}_{0}(w)]^{-1}d\\
 & =-\tfrac{1}{2}d^{\trans}\{[D\tilde{m}_{0}(w)]H^{-1}[D\tilde{m}_{0}(w)]^{\trans}\}^{-1}d
\end{align*}
as $\lambda\goesto0$, for all $w\in\mathcal{B}$ and $d\in\reals^{G}\backslash\{0\}$.

Since the l.h.s.\ of the preceding does not depend on $w$ or $d$
(for any value of $\lambda>0$), the limit on the r.h.s.\ cannot
either. Therefore, fixing a $w_{0}\in\mathcal{B}$ we obtain that
\[
[D\tilde{m}_{0}(w)]H^{-1}[D\tilde{m}_{0}(w)]^{\trans}=[D\tilde{m}_{0}(w_{0})]H^{-1}[D\tilde{m}_{0}(w_{0})]^{\trans}\eqdef S
\]
for all $w\in\mathcal{B}$. Taking $A$ and $B$ to be the (lower
triangular) Cholesky roots of the positive definite matrices $H^{-1}=AA^{\trans}$
and $S^{-1}=BB^{\trans}$ respectively, it follows that
\[
B^{\trans}[D\tilde{m}_{0}(w)]AA^{\trans}[D\tilde{m}_{0}(w)]^{\trans}B=B^{\trans}SB=I_{G}
\]
for all $w\in\mathcal{B}$, and hence the map
\[
\tilde{\ell}_{0}(w)\defeq B^{\trans}\tilde{m}_{0}(Aw)
\]
is a locally Lipschitz bijection $\reals^{G}\setmap\reals^{G}$ for
which 
\[
D\tilde{\ell}_{0}(w)=B^{\trans}D\tilde{m}_{0}(Aw)A,
\]
for all $w\in\mathcal{B}$, and hence
\begin{align*}
D\tilde{\ell}_{0}(w)D\tilde{\ell}_{0}(w)^{\trans} & =[B^{\trans}D\tilde{m}_{0}(Aw)A][B^{\trans}D\tilde{m}_{0}(Aw)A]^{\trans}\\
 & =B^{\trans}[D\tilde{m}_{0}(Aw)]AA^{\trans}[D\tilde{m}_{0}(Aw)]^{\trans}B=I_{G}
\end{align*}
for all $w\in\mathcal{B}$, whence also $D\tilde{\ell}_{0}(w)^{\trans}D\tilde{\ell}_{0}(w)=I_{G}$
for all $w\in\mathcal{B}$. Moreover, in view of \eqref{Dm}, \assref{sem}\ass{.}\ref{enu:sem:bij},
and the fact that the determinants of $A$ and $B$ must be strictly
positive, as triangular matrices with strictly positive diagonal entries,
we have
\begin{align*}
\det D\tilde{\ell}_{0}(w) & =(\det B)[\det D\tilde{m}_{0}(Aw)](\det A)>0
\end{align*}
for all $w\in\mathcal{B}$. Deduce $D\tilde{\ell}_{0}(w)\in\sorths(G)$
for all $w\in\mathcal{B}$.

It therefore follows by \lemref{lipschitz}\ref{enu:rigidity} that
there exists a $P\in\sorths(G)$ such that
\[
\tilde{\ell}_{0}(w)=a+Pw.
\]
Thus $\tilde{\ell}_{0}$ is affine, and hence so too is $\tilde{m}_{0}$.

\subsubsection*{(v) Conclusion.}

To conclude the proof, we recall that $\tilde{m}_{0}=\tilde{r}_{0}\compose r_{0}^{-1}$.
By the previous part of the proof, there exist $Q\in\reals^{G\times G}$
and $q\in\reals^{G}$ such that
\[
\tilde{r}_{0}[r_{0}^{-1}(w)]=\tilde{m}_{0}(w)=q+Qw
\]
for all $w\in\reals^{G}$, whence taking $y=r_{0}^{-1}(w)$ yields
\[
\tilde{r}_{0}(y)=q+Qr_{0}(y)
\]
for all $y\in\reals^{G}$.

It similarly follows from \eqref{r1tilde} above that
\[
\tilde{r}_{1}(x)=\tilde{u}^{\ast}-\tilde{m}_{0}[u^{\ast}-r_{1}(x)]=(\tilde{u}^{\ast}-q-Qu^{\ast})+Qr_{1}(x).
\]
Hence, defining $q_{0}\defeq\tilde{u}^{\ast}-Qu^{\ast}$, we obtain
\begin{align*}
\tilde{U} & \defeq\tilde{r}_{0}(Y)+\tilde{r}_{1}(X)=q_{0}+Q[r_{0}(Y)+r_{1}(X)]=q_{0}+QU
\end{align*}
whereupon for the distribution of $\tilde{U}$ to respect to scale
and location normalisation specified in \assref{sem}\ass{.}\ref{enu:sem:psdens},
we must have $q_{0}=0$, and that $Q$ is an orthogonal matrix. Since
\begin{equation}
\det Dr_{0}(y)=(\det Q)[\det D\tilde{r}_{0}(y)]\label{eq:detQpos}
\end{equation}
a.e., it follows from \assref{sem}\ass{.}\ref{enu:sem:bij} that
$\det Q>0$, and hence $Q\in\sorths(G)$.\hfill{}\qedsymbol{}

\subsection{Proof of \thmref{shockid}}

\label{app:SVARidproof}

This is essentially a matter of mapping the notation and assumptions
imposed on the nonlinear SVAR in \subsecref{nonlinearsvar}, into
their counterparts for the nonlinear SEM in \appref{nonlinearSEM},
and then applying \thmref{matzkin}. Making the identification
\begin{align}
(Y,X,U) & =(z_{t},\b z_{t-1},\err_{t}), & (r_{0},r_{1},f_{U}) & =(\fe_{0},\b{\fe}_{1},\den), & (\Gamma_{0},\Gamma_{1},\Phi) & =(\fespc_{0},\bigfspc_{1},\denspc),\label{eq:vartrans}
\end{align}
so that $G=p$ and $K=kp$, and noting that the nonlinear SVAR satisfies
\assref{svarpar} and \assref{svardgp}, it follows that the nonlinear
SEM satisfies \assref{sem}, with the only exceptions that: $\det D\tilde{r}_{0}(y)\neq0$
a.e., for each $\tilde{r}_{0}\in\Gamma_{0}$ rather than necessarily
being strictly positive a.e.; and that the location normalisation
$\tilde{r}_{0}(0)=0$ is now imposed.

However, since the sign of the determinant of the Jacobian of a locally
Lipschitz bijection $\reals^{G}\setmap\reals^{G}$ must be the same
a.e., it must be the case that for every $\tilde{r}_{0}\in\Gamma_{0}$,
either $\det D\tilde{r}_{0}(y)>0$ a.e., or $\det D\tilde{r}_{0}(y)<0$
a.e. Fixing a $Q_{0}\in\orths(p)$ with $\det Q_{0}=-1$, and suppose
e.g.\ that $r_{0}$ has $\det Dr_{0}(y)<0$ a.e. Then simply by multiplying
\eqref{matzkin} through by $Q_{0}$,
\[
Q_{0}U=(Q_{0}r_{0})(Y)+(Q_{0}r_{1})(X),
\]
we obtain a parametrisation $(Q_{0}r_{0},Q_{0}r_{1},f_{Q_{0}U})$
that is observationally equivalent to $(r_{0},r_{1},f_{U})$, but
where now $\det D[Q_{0}r_{0}](y)>0$ a.e. By similarly transforming
any candidate $(\tilde{r}_{0},\tilde{r}_{1},\tilde{f}_{U})$ for which
$\det D\tilde{r}_{0}(y)<0$ a.e., we can thus reduce the situation
to one in which both $\det Dr_{0}(y)>0$ a.e., and $\det D\tilde{r}_{0}(y)>0$
a.e., as is contemplated in \thmref{matzkin}. Because of the possibly
intervening transformation by $Q_{0}$, that result thus implies that
for a given $(\tilde{r}_{0},\tilde{r}_{1})\in\Gamma_{0}\times\Gamma_{1}$,
there exists an $f_{\tilde{U}}\in\Phi$ such that $(\tilde{r}_{0},\tilde{r}_{1},f_{\tilde{U}})$
is observationally equivalent to $(r_{0},r_{1},f_{U})$, if and only
if there exists a $Q\in\orths(G)$ -- which need not now be in $\sorths(G)$
-- such that 
\[
\tilde{r}_{0}(y)+\tilde{r}_{1}(x)=Q[r_{0}(y)+r_{1}(x)]
\]
for all $(y,x)\in\reals^{G}\times\reals^{K}$. Because of the location
normalisation $\tilde{r}_{0}(0)=0=r_{0}(0)$, this is equivalent to
\begin{align*}
\tilde{r}_{0}(y) & =Qr_{0}(y)\sep\forall y\in\reals^{G} & \tilde{r}_{1}(x) & =Qr_{1}(x)\sep\forall x\in\reals^{K}.
\end{align*}
Transposing this back to the notation of the SVAR, via \eqref{vartrans}
above, yields the result.\hfill\qedsymbol{}

\section{Proofs for piecewise affine functions}

\label{app:piecewiseaffine}

For the proof of \propref{pwacondition}, we shall need the following
auxiliary result, whose proof is given in \appref{conhull} below.
Let the convex hull of a collection of matrices $\{A_{i}\}_{i=1}^{k}$
be denoted $\ch\{A_{i}\}_{i=1}^{k}\defeq\{\sum_{i=1}^{k}\lambda_{i}A_{i}\mid\lambda_{i}\geq0\sep\sum_{i=1}^{k}\lambda_{i}=1\}$.
\begin{lem}
\label{lem:pwa-conhull}Suppose $f:\reals^{p}\setmap\reals^{p}$ is
a piecewise affine function. Then for every $x^{\prime},x^{\prime\prime}\in\reals^{p}$,
there exists a $\Phi\in\ch\{\Phi^{(\ell)}\}_{\ell=1}^{L}$ such that
\[
f(x^{\prime\prime})-f(x^{\prime})=\Phi(x^{\prime\prime}-x^{\prime}).
\]
\end{lem}

\subsection{Proof of \propref{pwacondition}}

To simplify the notation, throughout we drop the `$0$' subscript
on $\fe_{0}$ in the statement of the proposition, writing it simply
as $f$. Without loss of generality, we may suppose that \eqref{detcond}
holds with $\det\Phi^{(\ell)}>0$ for all $\ell\in\{1,\ldots,L\}$.

\paragraph{\ref{enu:orig}.}

By either Theorem\ 1 and 4 in \citet{GLM80Ecta}, which are applicable
in the piecewise linear and threshold affine cases respectively, $\fe:\reals^{p}\setmap\reals^{p}$
is invertible. Being continuous by assumption, it is therefore a homeomorphism,
by Theorem~4.3 in \citet{Deim85}. Since a piecewise affine function
is Lipschitz continuous (\citealp{Schol12}, Prop.\ 2.2.7), it remains
only to note that the inverse of an (invertible) piecewise affine
function is itself piecewise affine (\citealp{Schol12}, Prop.\ 2.3.1).

\paragraph{\ref{enu:smooth}.}

Fix $x^{\prime},x^{\prime\prime}\in\reals^{p}$. We have by \lemref{pwa-conhull}
that for every $u\in\reals^{p}$, there exist non-negative $\{\lambda_{\ell}(u)\}_{\ell=1}^{L}$
(which depend also on $x^{\prime},x^{\prime\prime}$) such that $\sum_{\ell=1}^{L}\lambda_{\ell}(u)=1$
and
\[
\fe(x^{\prime\prime}+u)-\fe(x^{\prime}+u)=\left[\sum_{\ell=1}^{L}\lambda_{\ell}(u)\Phi^{(\ell)}\right](x^{\prime\prime}-x^{\prime}).
\]
Hence
\begin{align}
\fe_{K}(x^{\prime\prime})-\fe_{K}(x^{\prime}) & =\int_{\reals^{p}}[\fe(x^{\prime\prime}+u)-\fe(x^{\prime}+u)]K(u)\diff u\nonumber \\
 & =\sum_{\ell=1}^{L}\left[\int_{\reals^{p}}\lambda_{\ell}(u)K(u)\diff u\right]\Phi^{(\ell)}(x^{\prime\prime}-x^{\prime})\eqdef\left[\sum_{\ell=1}^{L}\mu_{\ell}\Phi^{(\ell)}\right](x^{\prime\prime}-x^{\prime}).\label{eq:fKdiff}
\end{align}
where $\sum_{\ell=1}^{L}\mu_{\ell}=1$. Since the bracketed matrix
on the r.h.s.\ is an element of $\ch\{\Phi^{(\ell)}\}_{\ell=1}^{L}$,
it suffices to show that every matrix in that set is invertible.

We first note the following. Suppose $A$ and $B$ are square matrices,
with $\det A>0$ and $\det B>0$, and $B-A\eqdef uv^{\trans}$ having
rank 1. Then by Cauchy's formula for the determinant of a rank-1
perturbation,
\begin{equation}
\det B=\det(A+uv^{\trans})=(\det A)(1+u^{\trans}A^{-1}v),\label{eq:cauchyrank1}
\end{equation}
and so we must have that $u^{\trans}A^{-1}v>-1$. Therefore for every
$\lambda\in[0,1]$,
\begin{equation}
\det(\lambda A+(1-\lambda)B)=\det[A+(1-\lambda)uv^{\trans}]=(\det A)[1+(1-\lambda)u^{\trans}A^{-1}v]>0.\label{eq:cauchydet}
\end{equation}

Now suppose that $\fe$ is threshold affine. Since $f$ is continuous
at the thresholds,
\begin{equation}
\phi^{(\ell-1)}+\Phi^{(\ell-1)}x=\phi^{(\ell)}+\Phi^{(\ell)}x\label{eq:fcty}
\end{equation}
for all $x\in\reals^{p}$ such that $a^{\trans}x=\tau_{\ell-1}$.
Deduce that
\[
(\Phi^{(\ell-1)}-\Phi^{(\ell)})a_{\perp}=0
\]
where $a_{\perp}\in\reals^{p\times(p-1)}$has full column rank, and
$a^{\trans}a_{\perp}=0$. Hence there exists an $m^{(\ell)}\in\reals^{p}$
such that 

\[
\Phi^{(\ell)}-\Phi^{(\ell-1)}=m^{(\ell)}a^{\trans},
\]
and so
\begin{equation}
\Phi^{(\ell)}=\Phi^{(1)}+\sum_{i=2}^{\ell}m^{(i)}P_{a}\eqdef\Phi^{(1)}+n^{(\ell)}a^{\trans}\label{eq:Phiell}
\end{equation}
for every $\ell\in\{1,\ldots,L\}$. It follows from and \eqref{cauchyrank1}
above, and the fact that $\det\Phi^{(\ell)}>0$ for all $\ell\in\{1,\ldots,L\}$,
that $1+n^{(\ell)\trans}(\Phi^{(1)})^{-1}a>0$. Noting that
\[
\sum_{\ell=1}^{L}\mu_{\ell}\Phi^{(\ell)}=\Phi^{(1)}+\left[\sum_{\ell=1}^{L}\mu_{\ell}n^{(\ell)}\right]a^{\trans}
\]
it follows via another application of \eqref{cauchyrank1} that
\begin{align*}
\det\left(\sum_{\ell=1}^{L}\mu_{\ell}\Phi^{(\ell)}\right) & =(\det\Phi^{(1)})\cdot\left[1+\left(\sum_{\ell=1}^{L}\mu_{\ell}n^{(\ell)}\right)^{\trans}(\Phi^{(1)})^{-1}a\right]\\
 & =(\det\Phi^{(1)})\cdot\sum_{\ell=1}^{L}\mu_{\ell}[1+n^{(\ell)\trans}(\Phi^{(1)})^{-1}a]>0,
\end{align*}
as required.

Next suppose that $\fe$ is piecewise linear, and note that since
each $\set X^{(\ell)}$ is a union of cones of the form \eqref{pwlbasis},
we may without loss of generality write
\[
\fe(x)=\sum_{m=1}^{2^{p}}\indic\{x\in\set C_{{\cal I}_{m}}\}\tilde{\Phi}^{(m)}x
\]
where $\{{\cal I}_{m}\}_{m=1}^{2^{p}}$ partitions $2^{\{1,\ldots,p\}}$,
and each for each $m\in\{1,\ldots,2^{p}\}$, there is an $\ell\in\{1,\ldots,L\}$
such that $\tilde{\Phi}^{(m)}=\Phi^{(\ell)}$. Moreover, since $A=[a_{1},\ldots,a_{p}]$
is invertible, we may write
\begin{align*}
x\in\set C_{{\cal I}} & \iff a_{i}^{\trans}A^{-1}Ax\geq0\sep\forall i\in{\cal I}\text{ and }a_{i}^{\trans}A^{-1}Ax<0\sep\forall i\notin{\cal I}\\
 & \iff Ax\in\set D_{{\cal I}}
\end{align*}
where 
\[
\set D_{{\cal I}}\defeq\{x\in\reals^{p}\mid e_{i}^{\trans}x\geq0\sep\forall i\in{\cal I}\text{ and }e_{i}^{\trans}x<0\sep\forall i\notin{\cal I}\}.
\]
Hence
\[
g(Ax)\defeq\fe[A^{-1}(Ax)]=\fe(x)=\sum_{m=1}^{2^{p}}\indic\{Ax\in\set D_{{\cal I}_{m}}\}\tilde{\Phi}^{(m)}A^{-1}(Ax)
\]
and thus it suffices to prove the result with $f$ replaced by
\begin{align*}
g(y) & =\sum_{m=1}^{2^{p}}\indic\{y\in\set D_{{\cal I}_{m}}\}\Psi^{(m)}y=\sum_{i=1}^{p}[\psi_{i}^{+}\indic^{+}(y_{i})+\psi_{i}^{-}\indic^{-}(y_{i})]y_{i}
\end{align*}
where $\Psi^{(m)}\defeq\tilde{\Phi}^{(m)}A^{-1}$, and $\indic^{+}(s)\defeq\indic\{s\geq0\}$
and $\indic^{-}(s)\defeq\indic\{s<0\}$ for $s\in\reals$. The second
equality holds since $g$ is continuous, and so the coefficients on
$y_{i_{0}}=e_{i_{0}}^{\trans}y$ can only change at the point where
$y_{i_{0}}=0$; and therefore $\Psi^{(m)}e_{i_{0}}=\psi_{i_{0}}^{+}$
for all $m$ such that $\mathcal{I}_{m}\ni i_{0}$, while $\Psi^{(m)}e_{i_{0}}=\psi_{i_{0}}^{-}$
for all $m$ such that $\mathcal{I}_{m}\not\ni i_{0}$. By the requirement
\eqref{detcond}, the determinant of each $\Psi^{(m)}$ must have
the same sign (assumed without loss of generality to be positive).
Thus it suffices to show that for each $\lambda\defeq\{\lambda_{m}\}_{m=1}^{2^{p}}$
in the $(2^{p}-1)$-dimensional simplex,
\[
\Psi_{\lambda}\defeq\sum_{m=1}^{2^{p}}\lambda_{m}\Psi^{(m)}
\]
has $\det\Psi_{\lambda}>0$.

To that end, for each $i\in\{1,\ldots,p\}$, define $\mu_{i}\defeq\sum_{m=1}^{2^{p}}\lambda_{m}\indic\{\mathcal{I}_{m}\ni i\}$,
which sums the weights $\{\lambda_{m}\}$ over those $\set D_{{\cal I}_{m}}$
for which $y_{i}\geq0$. Thus the $i$th column of $\Psi_{\lambda}$
is equal to 
\[
\mu_{i}\psi_{i}^{+}+(1-\mu_{i})\psi_{i}^{-}\eqdef\bar{\psi}_{i}.
\]
For $q\in\{0,\ldots,p\}$, consider the $2^{p-q}$ matrices defined
by
\[
\Psi_{q}(s)=\begin{bmatrix}\bar{\psi}_{1}, & \ldots, & \bar{\psi}_{q}, & \psi_{q+1}(s_{1}), & \ldots, & \psi_{p}(s_{p-q})\end{bmatrix}
\]
where $s\in S^{p-q}\defeq\{-1,+1\}^{p-q}$, and 
\[
\psi_{i}(u)\defeq\psi_{i}^{-}\indic\{u=-1\}+\psi_{i}^{+}\indic\{u=+1\}
\]
We will show, via an induction, that: for each $q\in\{1,\ldots,p\}$,
$\det\Psi_{q}(s)>0$ for every $s\in S^{p-q}$. Since $\Psi_{\lambda}=\Psi_{p}$,
the result will then follow.

To that end, suppose that $q=1$. Then for every $s\in S^{p-1}$,
\begin{align*}
\Psi_{1}(s) & =\begin{bmatrix}\mu_{1}\psi_{1}^{+}+(1-\mu_{1})\psi_{1}^{-}, & \psi_{2}(s_{1}), & \ldots, & \psi_{p}(s_{p-1})\end{bmatrix}\\
 & =\mu_{1}\begin{bmatrix}\psi_{1}^{+}, & \psi_{2}(s_{1}), & \ldots, & \psi_{p}(s_{p-1})\end{bmatrix}+(1-\mu_{1})\begin{bmatrix}\psi_{1}^{-}, & \psi_{2}(s_{1}), & \ldots, & \psi_{p}(s_{p-1})\end{bmatrix}\\
 & =\mu_{1}\Psi_{0}[(+1,s^{\trans})^{\trans}]+(1-\mu_{1})\Psi_{0}[(-1,s^{\trans})^{\trans}]
\end{align*}
Since both $\Psi_{0}[(+1,s^{\trans})^{\trans}]$ and $\Psi_{0}[(-1,s^{\trans})^{\trans}]$
are elements of $\{\Psi^{(m)}\}_{m=1}^{2^{p}}$, they each have positive
determinant. Moreover, they differ only by a rank one matrix, and
so it follows by \eqref{cauchydet} that $\det\Psi_{1}(s)>0$ for
all $s\in S^{p-1}$. Thus the inductive hypothesis is true when $q=1$.

Now suppose the inductive hypothesis is true for all $q\in\{1,\ldots,q_{0}\}$,
where $q_{0}\leq p-1$. We must show it holds for $q=q_{0}+1$. Consider
\begin{align*}
\Psi_{q_{0}+1}(s) & =\begin{bmatrix}\bar{\psi}_{1}, & \ldots, & \bar{\psi}_{q_{0}}, & \mu_{q_{0}+1}\psi_{q_{0}+1}^{+}+(1-\mu_{q_{0}+1})\psi_{q_{0}+1}^{-}, & \psi_{q_{0}+1}(s_{1}), & \ldots, & \psi_{p}(s_{p-q_{0}})\end{bmatrix}\\
 & =\mu_{q_{0}+1}\Psi_{q_{0}}[(+1,s^{\trans})^{\trans}]+(1-\mu_{q_{0}+1})\Psi_{q_{0}}[(-1,s^{\trans})^{\trans}].
\end{align*}
By the inductive hypothesis, both $\Psi_{q_{0}}[(1,s^{\trans})^{\trans}]$
and $\Psi_{q_{0}}[(0,s^{\trans})^{\trans}]$ have strictly positive
determinant; and again they differ only be a rank one matrix. Hence
\eqref{cauchydet} implies that $\det\Psi_{q_{0}+1}(s)>0$ for all
$s\in S^{p-(q_{0}+1)}$, and so the inductive hypothesis is true for
$q=q_{0}+1$. Deduce that $\Psi_{\lambda}=\Psi_{p}$ has strictly
positive determinant, and is therefore invertible. Thus the smoothed
counterpart of $g$, and therefore of $f$ also, is invertible.

We have thus shown that $f_{K}$ is invertible in both the piecewise
linear and threshold affine cases, and that moreover $\sum_{\ell=1}^{L}\mu_{\ell}\Phi^{(\ell)}$
in \eqref{fKdiff} has strictly positive determinant. Clearly, $f_{K}$
is Lipschitz, since $\sum_{\ell=1}^{L}\mu_{\ell}\Phi^{(\ell)}$ is
bounded, for $(\mu_{1},\ldots,\mu_{L})$ an element of the $(L-1)$-dimensional
unit simplex $\Delta^{L-1}$. It follows moreover that it is bi-Lipschitz,
since the final term on the r.h.s.\ of
\[
\smlnorm{\fe_{K}(x^{\prime\prime})-\fe_{K}(x^{\prime})}\geq\smlnorm{x^{\prime\prime}-x^{\prime}}\inf_{\smlnorm v=1}\inf_{\{\mu_{\ell}\}\in\Delta^{L-1}}\norm{\left[\sum_{\ell=1}^{L}\mu_{\ell}\Phi^{(\ell)}\right]v},
\]
may not be zero for any (permitted) $v$ and $\{\mu_{\ell}\}$, since
that would otherwise contradict the invertibility of $\sum_{\ell=1}^{L}\mu_{\ell}\Phi^{(\ell)}$.
By continuity and compactness, the infimum must be achieved, and must
therefore be strictly positive. Finally, in view of the integrability
condition \eqref{integderiv}, $\fe_{K}$ must also have $m$ continuous
derivatives, by the dominated derivatives theorem.\hfill\qedsymbol{}

\subsection{Proof of \lemref{pwa-conhull}}

\label{app:conhull}

Let $\phi(x)\defeq\sum_{\ell=1}^{L}\indic\{x\in\set X^{(\ell)}\}\phi^{(\ell)}$
and $\Phi(x)\defeq\sum_{\ell=1}^{L}\indic\{x\in\set X^{(\ell)}\}\Phi^{(\ell)}$,
so that these are constant on each $\set X^{(\ell)}$, and $f(x)=\phi(x)+\Phi(x)x$.
Now let $x^{\prime},x^{\prime\prime}\in\reals^{p}$; with this notation,
\begin{align}
f(x^{\prime\prime})-f(x^{\prime}) & =[\phi(x^{\prime\prime})-\phi(x^{\prime})]+[\Phi(x^{\prime\prime})x^{\prime\prime}-\Phi(x^{\prime})x^{\prime}]\nonumber \\
 & =[\phi(x^{\prime\prime})-\phi(x^{\prime})]+\Phi(x^{\prime})(x^{\prime\prime}-x^{\prime})+[\Phi(x^{\prime\prime})-\Phi(x^{\prime})]x^{\prime\prime}.\label{eq:dfdecomp}
\end{align}
Define
\[
x(\delta)\defeq(1-\delta)x^{\prime}+\delta x^{\prime\prime}
\]
for $\delta\in\reals$. Since $f$ is continuous, so too is $\delta\elmap f[x(\delta)]$.
Because $\phi$ and $\Phi$ are piecewise constant, and $\{\set X^{(\ell)}\}_{\ell=1}^{L}$
is a convex partition of $\reals^{p}$, it follows that $\delta\elmap\phi[x(\delta)]$
and $\delta\elmap\Phi[x(\delta)]$ have $m\in\{0,\ldots,L-1\}$ points
of discontinuity for $\delta\in[0,1]$, located at some $\{\delta_{i}\}_{i=1}^{m}\subset[0,1]$
with $\delta_{i}<\delta_{i+1}$ for all $i$. If $m=0$, then the
result holds with $\Phi=\Phi^{(\ell^{\ast})}\in\ch\{\Phi^{(\ell)}\}_{\ell=1}^{L}$,
where $\ell^{\ast}$ is such that $x^{\prime}\in\set X^{(\ell^{\ast})}$.
We suppose therefore that $m\geq1$.

Set $x_{0}\defeq x^{\prime}$ and $x_{m}\defeq x^{\prime\prime}$;
and when $m\geq2$, let $\{x_{i}\}_{i=1}^{m-1}$ be chosen such that
$x_{i}=x(\delta)$ for some $\delta\in(\delta_{i},\delta_{i+1})$.
By the continuity of $\delta\elmap f[x(\delta)]$ at each $\delta=\delta_{i}$,
we must have
\begin{equation}
0=\lim_{\delta\dto\delta_{i}}f[x(\delta)]-\lim_{\delta\uto\delta_{i}}f[x(\delta)]=[\phi(x_{i})-\phi(x_{i-1})]+[\Phi(x_{i})-\Phi(x_{i-1})]x(\delta_{i}).\label{eq:ctyatdty}
\end{equation}
for $i\in\{1,\ldots,m\}$. Noting also that
\begin{equation}
x^{\prime\prime}-x(\delta_{i})=x^{\prime\prime}-[(1-\delta_{i})x^{\prime}+\delta_{i}x^{\prime\prime}]=(1-\delta_{i})(x^{\prime\prime}-x^{\prime}),\label{eq:xmxd}
\end{equation}
we may write the final term on the r.h.s.\ of \eqref{dfdecomp} as
\begin{align*}
[\Phi(x^{\prime\prime})-\Phi(x^{\prime})]x^{\prime\prime} & =[\Phi(x_{m})-\Phi(x_{0})]x^{\prime\prime}\\
 & =\sum_{i=1}^{m}[\Phi(x_{i})-\Phi(x_{i-1})]x^{\prime\prime}\\
 & =_{(1)}\sum_{i=1}^{m}[\Phi(x_{i})-\Phi(x_{i-1})][(1-\delta_{i})(x^{\prime\prime}-x^{\prime})+x(\delta_{i})]\\
 & =_{(2)}\sum_{i=1}^{m}(1-\delta_{i})[\Phi(x_{i})-\Phi(x_{i-1})](x^{\prime\prime}-x^{\prime})-\sum_{i=1}^{m}[\phi(x_{i})-\phi(x_{i-1})],
\end{align*}
where $=_{(1)}$ follows from \eqref{xmxd}, and $=_{(2)}$ from \eqref{ctyatdty}.
We note that
\[
\sum_{i=1}^{m}[\phi(x_{i})-\phi(x_{i-1})]=\phi(x_{m})-\phi(x_{0})
\]
and that setting $\delta_{0}\defeq0$ and $\delta_{m+1}=1$, we have
\begin{align*}
 & \sum_{i=1}^{m}(1-\delta_{i})[\Phi(x_{i})-\Phi(x_{i-1})]\\
 & \qquad\qquad\qquad=(1-\delta_{m})\Phi(x_{m})+\sum_{i=1}^{m-1}[(1-\delta_{i})-(1-\delta_{i+1})]\Phi(x_{i})-(1-\delta_{1})\Phi(x_{0})\\
 & \qquad\qquad\qquad=\sum_{i=0}^{m}[(1-\delta_{i})-(1-\delta_{i+1})]\Phi(x_{i})-\Phi(x_{0})\\
 & \qquad\qquad\qquad=\sum_{i=0}^{m}(\delta_{i+1}-\delta_{i})\Phi(x_{i})-\Phi(x_{0})
\end{align*}
and whence
\begin{align*}
[\Phi(x^{\prime\prime})-\Phi(x^{\prime})]x^{\prime\prime} & =\left[\sum_{i=0}^{m}(\delta_{i+1}-\delta_{i})\Phi(x_{i})\right](x^{\prime\prime}-x^{\prime})-\Phi(x_{0})(x^{\prime\prime}-x^{\prime})-[\phi(x_{m})-\phi(x_{0})]\\
 & =\left[\sum_{i=0}^{m}\lambda_{i}\Phi(x_{i})\right](x^{\prime\prime}-x^{\prime})-\Phi(x^{\prime})(x^{\prime\prime}-x^{\prime})-[\phi(x^{\prime\prime})-\phi(x^{\prime})]
\end{align*}
where $\lambda_{i}\defeq\delta_{i+1}-\delta_{i}\geq0$ and $\sum_{i=0}^{m}\lambda_{i}=\sum_{i=0}^{m}(\delta_{i+1}-\delta_{i})=\delta_{m+1}-\delta_{0}=1$.
It follows from \eqref{dfdecomp} that
\[
f(x^{\prime\prime})-f(x^{\prime})=\left[\sum_{i=0}^{m}\lambda_{i}\Phi(x_{i})\right](x^{\prime\prime}-x^{\prime}).
\]
Finally, noting that for each $i\in\{1,\ldots,m\}$, there exists
an $\ell_{i}\in\{1,\ldots,L\}$ such that $\Phi(x_{i})=\Phi^{(\ell_{i})}$,
we have $\sum_{i=0}^{m}\lambda_{i}\Phi(x_{i})\in\ch\{\Phi^{(\ell)}\}_{\ell=1}^{L}$
as required.\hfill\qedsymbol{}

\section{Proofs for the extended model}

\label{app:hetero}

\subsection{Identification in the augmented SEM}

Here we return to the setting of the nonlinear SEM from \appref{nonlinearSEM},
augmented to allow for conditional heteroskedasticity of the form
\begin{equation}
\vol(Z)U=r(Y,X,Z)=r_{0}(Y)+r_{1}(X,Z),\label{eq:hetSEM}
\end{equation}
where the skedastic function $\vol(z)$ is a diagonal matrix with
strictly positive entries, for every $z\in\reals^{L}$. We shall now
maintain that $U$ is independent of $(X,Z)$. In this formulation
of the model, the $X$ variables play a special role, in being excluded
from the skedastic function; and identification will now hinge on
there being sufficient dependence of the r.h.s.\ on $X$ \emph{given}
$Z=z$. (Note also that we will \emph{not} require $r_{1}(x,z)$ to
be continuous with respect to $z$.)

To allow $Z$ to be discrete, we shall suppose that it has has some
support $\mathcal{Z}\subset\reals^{L}$, and a distribution thereon
that is equivalent to some measure $\nu$. We shall suppose that \emph{conditional
on $\nu$-almost every $z\in\mathcal{Z}$}, $X$ has a (Lebesgue)
density $f_{X\mid Z}$ with support $\reals^{K}$ (i.e.\ $f_{X\mid Z}$
may depend on $z$, but its support does not). The model then implies,
for $\nu$-a.e.\ $z\in\mathcal{Z}$, that $Y$ has the following
density conditional on $(X,Z)$:
\begin{align*}
f_{Y\mid X,Z}(y\mid x,z) & =f_{U}[\vol(z)^{-1}r(y,x,z)]\cdot\det\vol(z)^{-1}Dr_{0}(y)\\
 & =f_{U}\{\vol(z)^{-1}[r_{0}(y)+r_{1}(x,z)]\}\cdot\det\vol(z)^{-1}Dr_{0}(y),
\end{align*}
a.e.\ $(y,x)\in\reals^{G+K}$. (So long as the distribution of $Z$
is equivalent to $\nu$, this holds irrespective of what the distribution
of $Z$ actually is, a fact that is useful when we come apply our
results to an SVAR in which the distribution of the predetermined
variables may not be stationary.) We will accordingly now say that
two alternative parametrisations $(r_{0},r_{1},\vol,f_{U})$ and $(\tilde{r}_{0},\tilde{r}_{1},\tilde{\vol},f_{\tilde{U}})$
are \emph{observationally equivalent} if for $\nu$-a.e.\ $z\in\mathcal{Z}$,
\begin{equation}
f_{U}[\vol(z)^{-1}r(y,x,z)]\cdot\det\vol(z)^{-1}Dr_{0}(y)=f_{\tilde{U}}[\tilde{\vol}(z)^{-1}\tilde{r}(y,x,z)]\cdot\det\tilde{\vol}(z)^{-1}D\tilde{r}_{0}(y)\label{eq:obseq}
\end{equation}
a.e.\ $(y,x)\in\reals^{G+K}$.

The model \eqref{hetSEM} is now parameterised by the functions $r_{0}:\reals^{G}\setmap\reals^{G}$,
$r_{1}:\reals^{K}\setmap\reals^{G}$, $\vol:\reals^{L}\setmap\reals^{G\times G}$
and the density $f_{U}$; the sets $\Gamma_{i}\ni r_{i}$, for $i\in\{0,1\}$,
$\Sigma\ni\vol$, and $\Phi\ni f_{U}$ define the parameter space.
Our assumptions here amount to only minor modifications of those maintained
in \appref{nonlinearSEM}. Note, in particular, that although we continue
to require $y\elmap r_{0}(y)$ and $x\elmap r_{1}(x,z)$ to be Lipschitz
continuous, we do not require continuity of either $z\elmap\vol(z)$
or $z\elmap r_{1}(x,z)$. To normalise the overall scale of \eqref{hetSEM},
we shall suppose that there is a (known) $z^{\ast}\in\mathcal{Z}$
such that for every $\tilde{s}\in\Sigma$,
\begin{equation}
\tilde{s}(z^{\ast})=I_{p},\label{eq:snorm}
\end{equation}
with $\tilde{s}$ continuous at $z^{\ast}$, and $\nu$ placing strictly
positive mass on every neighbourhood of $z^{\ast}$.

\assumpname{SEM$^{\prime}$}
\begin{assumption}
\label{ass:het} \assref{sem} holds, with only the following modifications
to parts\ \ref{enu:sem:lipschitz} and \ref{enu:sem:surj}: 
\begin{enumerate}[label=\ass{\Alph*1.}, ref=\ass{\Alph*1}]
\item \label{enu:ske:lipschitz}for every $\tilde{r}_{0}\in\Gamma_{0}$
and $\tilde{r}_{1}\in\Gamma_{1}$: $\tilde{r}_{0}$ and $x\elmap\tilde{r}_{1}(x,z)$
are locally Lipschitz, for every $z\in\mathcal{Z}$;
\item \label{enu:het:surj}$x\elmap r_{1}(x,z)$ is surjective, with $\rank D_{x}r_{1}(x,z)=\reals^{G}$
for almost every $x\in\reals^{K}$, for every $z\in\mathcal{Z}$.
\end{enumerate}
Moreover, for every $\tilde{\vol}\in\Sigma$: $\tilde{\vol}(z)$ is
a $(G\times G)$ diagonal matrix with strictly positive entries, for
every $z\in\mathcal{Z}$; and the scale normalisation \eqref{snorm}
holds.
\end{assumption}

We may now state our main result on observational equivalence in the
model \eqref{hetSEM}.
\begin{thm}
\label{thm:semhet}Suppose that \assref{het} holds. Let $\tilde{r}_{i}\in\Gamma_{i}$
for $i\in\{0,1\}$. Then there exist $(\tilde{\vol},f_{\tilde{U}})\in\Sigma\times\Phi$
such that $(\tilde{r}_{0},\tilde{r}_{1},\tilde{\vol},\tilde{f}_{U})$
is observationally equivalent to $(r_{0},r_{1},\vol,f_{U})$, if and
only if there exists a $Q\in\sorths(G)$ such that for $\nu$-a.e.\ $z\in\mathcal{Z}$:
\begin{equation}
\tilde{r}_{0}(y)+\tilde{r}_{1}(x,z)=Q[r_{0}(y)+r_{1}(x,z)]\label{eq:hetoe1}
\end{equation}
for all $(y,x)\in\reals^{G+K}$; and
\begin{equation}
Q\vol^{2}(z)Q^{\trans}\label{eq:hetoe2}
\end{equation}
is a diagonal matrix; in which case $\tilde{\vol}(z)=Q\vol(z)Q^{\trans}$.
\end{thm}
\begin{proof}
Suppose \eqref{matzeq} and \eqref{hetoe2} hold for some $Q\in\sorths(G)$.
Then setting $\tilde{\vol}(z)\defeq Q\vol(z)Q^{\trans}$ for $\nu$-a.e.\ $z\in\mathcal{Z}$,
we have that
\begin{align*}
\tilde{U} & \defeq\tilde{\vol}(Z)^{-1}[\tilde{r}_{0}(Y)+\tilde{r}_{1}(X,Z)]\\
 & =Q\vol(Z)^{-1}Q^{\trans}Q[r_{0}(Y)+r_{1}(X,Z)]=QU
\end{align*}
a.s., which will be independent of $(X,Z)$, with a density $f_{\tilde{U}}$
that satisfies \assref{sem}\ass{.}\ref{enu:sem:psdens}; hence observational
equivalence obtains in this case. It remains therefore to prove the
reverse implication.

Suppose therefore that $(\tilde{r}_{0},\tilde{r}_{1},\tilde{\vol},\tilde{f}_{U})$
and $(r_{0},r_{1},\vol,f_{U})$ are observationally equivalent. Then
there exists a $\mathcal{Z}_{0}\subset\mathcal{Z}$ such that $\nu(\mathcal{Z}_{0})=1$
and \eqref{obseq} holds for every $z\in\mathcal{Z}_{0}$. Fixing
a $z_{0}\in\mathcal{Z}_{0}$, and only allowing $(y,x)$ to vary,
it is evident that the notion of observational equivalence in \eqref{obseq},
i.e.
\[
f_{U}[\vol(z_{0})^{-1}r(y,x,z_{0})]\cdot\det\vol(z_{0})^{-1}Dr_{0}(y)=f_{\tilde{U}}[\tilde{\vol}(z_{0})^{-1}\tilde{r}(y,x,z_{0})]\cdot\det\tilde{\vol}(z_{0})^{-1}D\tilde{r}_{0}(y)
\]
a.e.\ $(y,x)\in\reals^{G+K}$, coincides with that of \eqref{oe}
for the model \eqref{matzkin}: the only difference being that in
\eqref{oe} the dependence on $z_{0}$ is suppressed from the notation.
By \thmref{matzkin}, the preceding equality implies that there exists
a $P(z_{0})\in\sorths(G)$ such that
\[
\tilde{\vol}(z_{0})^{-1}\tilde{r}(y,x,z_{0})=P(z_{0})\vol(z_{0})^{-1}r(y,x,z_{0})
\]
for all $(y,x)\in\reals^{G+K}$, where we have written $P(z_{0})$
because this matrix may depend on the $z_{0}$ that was fixed above.
Since the preceding argument holds for every $z\in\mathcal{Z}_{0}$,
we thus obtain a map $P:\mathcal{Z}_{0}\setmap\sorths(G)$ such that
\begin{equation}
\tilde{\vol}(z)^{-1}[\tilde{r}_{0}(y)+\tilde{r}_{1}(x,z)]=P(z)\vol(z)^{-1}[r_{0}(y)+r_{1}(x,z)]\label{eq:condzid}
\end{equation}
for all $(y,x,z)\in\reals^{G}\times\reals^{K}\times\mathcal{Z}_{0}$.

Note that since we can exchange arbitrary constants between $\tilde{r}_{0}$
and $\tilde{r}_{1}$ (and between $r_{0}$ and $r_{1}$), as per
\[
\tilde{r}_{0}(y)+\tilde{r}_{1}(x,z)=[\tilde{r}_{0}(y)-\tilde{r}_{0}(0)]+[\tilde{r}_{1}(x,z)+\tilde{r}_{0}(0)],
\]
without disturbing \eqref{hetoe1}, we may without loss of generality
suppose that $\tilde{r}_{0}(0)=r_{0}(0)=0$; we maintain this henceforth.
Rearranging \eqref{condzid} as
\[
\tilde{\vol}(z)^{-1}\tilde{r}_{0}(y)-P(z)\vol(z)^{-1}r_{0}(y)=P(z)\vol(z)^{-1}r_{1}(x,z)-\tilde{\vol}(z)^{-1}\tilde{r}_{1}(x,z)
\]
we see that both sides of the equality must be invariant to the values
of $y$ and $x$. Taking $y=0$, and using that $\tilde{r}_{0}(0)=r_{0}(0)=0$,
we thus obtain
\[
\tilde{\vol}(z)^{-1}\tilde{r}_{0}(y)-P(z)\vol(z)^{-1}r_{0}(y)=0=P(z)\vol(z)^{-1}r_{1}(x,z)-\tilde{\vol}(z)^{-1}\tilde{r}_{1}(x,z)
\]
for all $(y,x)\in\reals^{G+K}$ and $z\in\mathcal{Z}_{0}$. Deduce
from the first equality that
\begin{equation}
\tilde{r}_{0}(y)=\tilde{\vol}(z)P(z)\vol(z)^{-1}r_{0}(y)\label{eq:rtildertue}
\end{equation}
for all $(y,z)\in\reals^{G}\times\mathcal{Z}_{0}$. Since only the
r.h.s.\ depends on $z$, and $r_{0}(y)$ is surjective, it follows
-- e.g.\ by considering values $\{y^{i}\}_{i=1}^{G}$ such that
$r_{0}(y^{i})=e_{i}$, for $e_{i}$ the $i$th column of $I_{G}$
-- that $\tilde{\vol}(z)P(z)\vol(z)^{-1}$ cannot depend on $z$.
Hence, fixing a $z_{0}\in\mathcal{Z}_{0}$, we have that 
\begin{equation}
\tilde{\vol}(z)P(z)\vol(z)^{-1}=\tilde{\vol}(z_{0})P(z_{0})\vol(z_{0})^{-1}\eqdef Q\label{eq:sigPsig}
\end{equation}
for all $z\in\mathcal{Z}_{0}$.

It follows from \eqref{condzid} and \eqref{sigPsig} that
\[
\tilde{r}_{0}(y)+\tilde{r}_{1}(x,z)=Q[r_{0}(y)+r_{1}(x,z)]
\]
for all $(y,x,z)\in\reals^{G+K}\times\mathcal{Z}_{0}$. Further, rearranging
\eqref{sigPsig} yields
\[
P(z)=\tilde{\vol}(z)^{-1}Q\vol(z)
\]
for all $z\in\mathcal{Z}_{0}$ Since $P(z)\in\sorths(G)$, we have
\[
I_{G}=P(z)P(z)^{\trans}=\tilde{\vol}(z)^{-1}Q\vol^{2}(z)Q^{\trans}\tilde{\vol}(z)^{-1}
\]
and hence
\begin{equation}
\tilde{\vol}^{2}(z)=Q\vol^{2}(z)Q^{\trans},\label{eq:stildes}
\end{equation}
for all $z\in\mathcal{Z}_{0}$ , so that the r.h.s.\ is indeed a
diagonal matrix, as claimed.

Finally, we recall that the scale normalisation \eqref{snorm} entails
that $\tilde{\vol}^{2}(z^{\ast})=I_{p}=\vol^{2}(z^{\ast})$ for some
$z^{\ast}\in\mathcal{Z}$. If $z^{\ast}\in\mathcal{Z}_{0}$, then
we obtain immediately from \eqref{stildes} that
\[
I_{G}=\tilde{\vol}^{2}(z^{\ast})=Q\vol^{2}(z^{\ast})Q^{\trans}=QQ^{\trans},
\]
and hence $Q\in\orths(p)$. If $z^{\ast}\notin\mathcal{Z}$, then
our assumption that $\nu$ has strictly positive mass in every neighbourhood
of $z^{\ast}$ implies that there exists a sequence $\{z_{n}\}$ in
$\mathcal{Z}_{0}$ with $z_{n}\goesto z^{\ast}$. Hence, by \eqref{stildes},
and the maintained continuity of $\tilde{s}$ and $s$ at $z^{\ast}$,
\[
I_{G}=\tilde{s}^{2}(z^{\ast})=\lim_{n\goesto\infty}\tilde{s}^{2}(z_{n})=Q\left[\lim_{n\goesto\infty}\vol^{2}(z_{n})\right]Q^{\trans}=Qs^{2}(z^{\ast})Q^{\trans}=QQ^{\trans}
\]
so that again $Q\in\orths(p)$. That $\det Q>0$ follows by the same
arguments as which yielded \eqref{detQpos} in the proof of \thmref{matzkin}.
\end{proof}

\subsection{Proof of \thmref{svarhetero}}

The argument is analogous to that given in the proof of \thmref{shockid},
with \thmref{semhet} now playing the role of \thmref{matzkin}. We
now make the identification
\begin{align*}
(Y,X,Z,U) & =(z_{t},\b z_{t-1}^{(1)},(\b z_{t-1}^{(2)},v_{t-1}),\err_{t}), & (r_{0},r_{1},\vol,f_{U}) & =(\fe_{0},\b{\fe}_{1},\sigma,\den),
\end{align*}
and $(\Gamma_{0},\Gamma_{1},\Sigma,\Phi)=(\fespc_{0},\bigfspc_{1},\set S,\denspc)$.
Observe, in particular, that under our assumptions, $Z=(\b z_{t-1}^{(2)},v_{t-1})$
is supported on $\mathcal{Z}=\reals^{d_{(2)}}\times\mathcal{V}$,
with a distribution that (for every $t\geq1$) is equivalent to $\nu=\leb_{\reals^{d_{(2)}}}\otimes\mu_{v}$,
where $\leb_{\reals^{d_{(2)}}}$ denotes Lebesgue measure on $\reals^{d_{(2)}}$
(see the discussion following \eqref{nlVARdens} above, which also
applies here). Since $v_{t-1}$ is independent of $(\b z_{t-1}^{(1)},\b z_{t-1}^{(2)})$,
and the latter has a distribution that is equivalent to $\leb_{\reals^{kp}}$,
it follows that $X=\b z_{t-1}^{(1)}$ has, conditionally on $Z=(\b z_{t-1}^{(2)},v_{t-1})$,
a continuous distribution that is supported on the whole of $\reals^{K}=\reals^{d_{(1)}}$.

With these definitions, it is readily verified that the nonlinear
SEM satisfies \assref{het}, with the only exceptions that $\det D\tilde{r}_{0}(y)\neq0$
a.e., for each $\tilde{r}_{0}\in\Gamma_{0}$ rather than necessarily
being strictly positive a.e.; and that the location normalisation
$\tilde{r}_{0}(0)=0$ is now imposed. These can be handled by the
same arguments as were used in the proof of \thmref{shockid}, whereupon
an application of \thmref{semhet} yields the result.\hfill\qedsymbol{}

\end{document}